\RequirePackage{fix-cm}
\documentclass[twocolumn]{svjour3}          
\smartqed  
\usepackage{graphicx}
\usepackage{multirow}
\usepackage{array}
\usepackage{amssymb}
\usepackage{amsmath}
\usepackage[justification=centering]{caption}

\usepackage{graphicx}
\usepackage{balance}  
\usepackage{subfigure}
\usepackage[ruled,vlined,linesnumbered]{algorithm2e}
\usepackage{enumitem}
\usepackage{booktabs}
\usepackage{cite}
\usepackage{multirow}
\usepackage{array}
\usepackage{tikz}

\newcommand*{\circled}[1]{\lower.7ex\hbox{\tikz\draw (0pt, 0pt)%
		circle (.5em) node {\makebox[1em][c]{\small #1}};}}

\newtheorem{myexample}{Example}

\begin{document}

\title{
PM-LSH: a fast and accurate in-memory framework
for high-dimensional approximate NN and closest pair search
}

\author{
Bolong Zheng \and
Xi Zhao \and
Lianggui Weng \and
Nguyen Quoc Viet Hung \and
Hang Liu \and
Christian S. Jensen
}


\institute{
Bolong Zheng, Xi Zhao and Lianggui Weng \at Huazhong University of Science and Technology, Wuhan, China\\
\email{\{bolongzheng,zhaoxi,liangguiweng\}@hust.edu.cn}
\and
Nguyen Quoc Viet Hung \at Griffith University, Gold Coast, Australia\\
\email{quocviethung.nguyen@griffith.edu.au}
\and
Hang Liu \at Stevens Institute of Technology, Hoboken, USA\\
\email{hang.liu@stevens.edu}
\and
Christian S. Jensen \at Aalborg Universigy, Aalborg, Denmark\\
\email{csj@cs.aau.dk}
}

\date{Received: date / Accepted: date}

\maketitle

\begin{abstract}
Nearest neighbor (NN) search is inherently computationally expensive in high-dimensional spaces due to the curse of dimensionality. As a well-known solution, locality-sensitive hashing (LSH) is able to answer $c$-approximate NN ($c$-ANN) queries in sublinear time with constant probability.
Existing LSH methods focus mainly on building hash bucket-based indexing such that the candidate points can be retrieved quickly. However, existing coarse-grained structures fail to offer accurate distance estimation for candidate points, which translates into additional computational overhead when having to examine unnecessary points. This in turn reduces the performance of query processing.
In contrast, we propose a fast and accurate in-memory LSH framework, called PM-LSH, that aims to compute $c$-ANN queries on large-scale, high-dimensional datasets.
First, we adopt a simple yet effective PM-tree to index the data points.
Second, we develop a tunable confidence interval to achieve accurate distance estimation and guarantee high result quality.
Third, we propose an efficient algorithm on top of the PM-tree to improve the performance of computing $c$-ANN queries.

In addition, we extend PM-LSH to support closest pair (CP) search in high-dimensional spaces. We again adopt the PM-tree to organize the points in a low-dimensional space, and we propose a branch and bound algorithm together with a radius pruning technique to improve the performance of computing $c$-approximate closest pair ($c$-ACP) queries.

Extensive experiments with real-world data offer evidence that PM-LSH is capable of outperforming existing proposals with respect to both efficiency and accuracy for both NN and CP search.
\end{abstract}
\section{Introduction}
Nearest neighbor (NN) querying in high-dimensional spaces is classic functionality that is used in a wide variety of important applications, such as sequence matching \cite{DBLP:journals/scientometrics/AbdulhayogluT18}, recommendation \cite{DBLP:conf/www/DasDGR07}, similar-item retrieval \cite{DBLP:conf/iccv/KulisG09}, and de-duplication \cite{DBLP:conf/edbt/NarangB11}, to name but a few. 
Let $\mathcal{D}$ be a set of points in $d$-dimensional space $\mathbb{R}^d$. Given a query point $q$, an NN query returns a point $o^*$ in $\mathcal{D}$ such that its Euclidean distance to $q$ is the minimum among all points in $\mathcal{D}$.

While the exact NN query in low-dimensional space already has efficient solutions \cite{DBLP:conf/sigmod/BeckmannKSS90, DBLP:conf/icde/ChenGLJC15}, providing an efficient solution for large, high-dimensional datasets remains a challenge, as both the query time and the space cost may increase exponentially with respect to the dimensionality. This phenomenon is called the ``curse of dimensionality.'' Fortunately, it frequently suffices to find an approximate nearest neighbor (ANN). 
Given an approximation ratio $c$ ($c>1$) and a query point $q$, a $c$-ANN query returns a point $o$ whose distance to $q$ is at most $cr^*$, where $r^*$ is the distance between $q$ and its exact NN.

A widely adopted locality-sensitive hashing (LSH) method enables computing $c$-ANN queries in sublinear time with constant probability. Generally, LSH maps the points in the dataset to buckets in hash tables by using a set of predefined hash functions that are designed to be locality-sensitive so that close points are hashed to the same bucket with high probability. A query is answered by examining the points that are hashed to the same bucket as the query point, or to similar buckets. Based on their main ideas, we classify the mainstream LSH methods into three categories: 1) \underline{P}robing \underline{S}equence-based (PS) approaches \cite{DBLP:conf/vldb/LvJWCL07,DBLP:journals/pvldb/LvJWCL17,DBLP:conf/sigmod/LiYZXCLNC18}; 2) \underline{R}adius \underline{E}nlargement-based (RE) approaches \cite{DBLP:conf/sigmod/TaoYSK09,DBLP:conf/sigmod/GanFFN12,DBLP:journals/pvldb/HuangFZFN15}; and 3) \underline{M}etric \underline{I}ndexing-based (MI) approaches \cite{DBLP:journals/pvldb/SunWQZL14}. PS approaches use a carefully derived probing sequence to examine multiple hash buckets that are likely to contain the nearest neighbor of a query. RE approaches process a sequence of range queries by enlarging the query range repeatedly until a qualified point is found. In MI approaches, the points are transformed into a low-dimensional, so-called projected space. The coordinates of a point in the projected space are the point's hash values. MI approaches then use a metric index to organize the points such that the distance between two points in the projected space can be used to approximate the distance between them in the original space.

When evaluating the performance of LSH methods, many pertinent performance metrics for $c$-ANN search exist, including efficiency, accuracy, memory consumption, and preprocessing overhead. Among these, both \textit{efficiency} and \textit{accuracy} are important metrics since a desirable algorithm should return results as soon as possible with a quality that is as high as possible, while the memory consumption and preprocessing overhead must be tolerable in the setting of a commodity machine. The performance of LSH depends on two aspects: 1) the estimation of distances between the query point and candidate points; and 2) the probing order of buckets/points. It is proved \cite{DBLP:journals/pvldb/SunWQZL14} that the ratio of the projected distance to the original distance between any two points follows a $\chi^2$ distribution. Therefore, if we are able to estimate the distance between two points accurately, we are able to find high-quality candidates. In addition, a well-designed index structure is required to quickly locate high-quality candidates.

However, the existing LSH methods suffer from either inaccurate distance estimation or unnecessary point probing overhead. For instance, SRS \cite{DBLP:journals/pvldb/SunWQZL14} is the state-of-the-art algorithm that uses an R-tree to index the points in the projected space. By searching the R-tree, SRS is able to iteratively return the next nearest point to $q$. The problem is that finding the next exact NN in an R-tree generally causes additional computational overhead, while the next NN is not necessarily the best next candidate in the original space. Next, Multi-Probe \cite{DBLP:conf/vldb/LvJWCL07} iteratively identifies the next hash bucket to be examined that has the least distance to $q$. However, most of the points in the identified buckets have to be probed due to poor estimation of the distance between $q$ and the candidate point. Finally, QALSH \cite{DBLP:journals/pvldb/HuangFZFN15} shares the same issue as Multi-Probe, and it uses a large number of hash functions that may incur high space consumption.

We propose a fast and accurate in-memory framework, called PM-LSH, for computing $c$-ANN queries on large-scale, high-dimensional datasets. 
The framework consists of three components, namely data partitioning, distance estimation, and point probing. First, we adopt the simple yet effective PM-tree \cite{DBLP:conf/dasfaa/SkopalPS05} to index the points in the projected space. Second, in order to improve the distance estimation accuracy, we exploit the strong relationship between the original and projected distance of any two points, and we develop a tunable confidence interval on the projected distance w.r.t. a given original distance. Third, we propose an efficient algorithm to search the PM-tree with a sequence of range queries with increasingly large radius. PM-LSH is able to achieve both high efficiency and high accuracy when compared with existing LSH methods.

We extend the PM-LSH technique to solve another classical problem, approximate closest pair (CP) search in high dimensional spaces.
Like NN search,  CP search is used in a wide range of settings, such as unsupervised classification or clustering \cite{DBLP:journals/paa/PirbonyehRPNM19}, user pattern similarity search \cite{DBLP:journals/thms/ZhouWJ18}, and geographic information systems \cite{DBLP:journals/geoinformatica/GutierrezS13}, to name but a few.
For a given approximation ratio $c$ ($c>1$) and a dataset $\mathcal{D}$, a $c$-approximate closest pair ($c$-ACP) query returns a point pair $(o_1,o_2)$ with distance at most $cr^*$, where $r^*$ is the distance of the exact closest pair in $\mathcal{D}$.
Early studies mainly adopt space partitioning indexing techniques to solve exact CP queries in two or three dimensions \cite{DBLP:conf/sigmod/CorralMTV00,DBLP:conf/sigmod/HjaltasonS98,DBLP:conf/ssd/ShanZS03,DBLP:journals/dke/CorralMTV04,DBLP:journals/tkde/ShinML03,DBLP:journals/tkde/KimP10}. However, these methods cannot be extended directly to support high-dimensional CP queries efficiently due to the curse of dimensionality. Therefore, improved indexes are proposed to address the effects of dimensionality \cite{kurasawa2011finding,DBLP:conf/sisap/FredrikssonB13,DBLP:conf/sisap/PearsonS14,DBLP:journals/vldb/GaoCLYC15}. Nonetheless, when faced with hundreds or thousands of dimensions, the performance of these methods still degenerates to nearly brute-force performance. 
Thus, another direction is to use dimension reduction methods to solve $c$-ACP, such as LSH or random projection.
For instance, the LSB-tree \cite{DBLP:journals/tods/TaoYSK10} uses a compound hash function to project points into a low-dimensional space and adopts the Z-curve to transform projected points into one-dimensional values that are indexed by a B-tree. 
The candidate point pairs are generated from points with the same Z-values. To improve the query accuracy, $L=O(\sqrt{n})$ B-trees are built, which yields a large space consumption. 
Next, ACP-P \cite{DBLP:conf/pakdd/CaiRZ18} projects the points directly into a one-dimensional space. The points with close distances in the projected space are considered as candidate point pairs. However, the distance estimation is inaccurate and leads to unnecessary candidate verification. 

To compute approximate CP queries, we still employ the PM-tree to index the points in the projected space, which provides an accurate distance estimation for point pairs. Next, we adopt a branch and bound method combined with a radius pruning technique to improve the query efficiency, which enables generation of sufficient candidate pairs with only a small space consumption.
We also note that our method is tunable and enables different trade-offs between query accuracy and query efficiency.

The major contributions are summarized as follows:
\begin{itemize}
\item We present a unified interpretation of the existing mainstream LSH methods and thoroughly analyze the competitors in relation to our method.
\item We propose an accurate and fast method called PM-LSH for c-ANN querying of large-scale, high-dimensional datasets. First, we use the PM-tree to index the points in the projected space. Second, we develop a tunable confidence interval for distance estimation. Third, we propose a c-ANN query algorithm that uses the PM-tree.
\item We extend the PM-LSH to support CP queries. First, we still employ the PM-tree to index the points in the projected space. Next, we propose a branch and bound algorithm together with a radius pruning technique for computing $c$-ACP queries.
\item We conduct an extensive performance study using real datasets that covers the state-of-the-art algorithms, which indicates that PM-LSH is efficient as well as accurate in terms of both the overall ratio and recall for both NN and CP search.
\end{itemize}

The paper extends its conference version \cite{DBLP:journals/pvldb/ZhengZWHLJ20} in several respects. Key extensions include (1) the extension of PM-LSH to support CP queries, (2) the coverage of related work on high-dimensional CP search, and (3) the paper's report on the experimental evaluations of the corresponding proposals. In addition, other parts of the paper have been revised when compared to the conference version.

The rest of the paper is organized as follows. Section \ref{sec:problem} presents the problem setting and preliminaries. Section \ref{sec:competitors} introduces a unified LSH framework, followed by our PM-LSH framework in Section \ref{sec:ourmethod}. Sections \ref{sec:nn} and \ref{sec:cp} introduce the NN and CP query processing based on PM-LSH, respectively. Section \ref{sec:experiments} covers experimental studies that offer insight into the performance of the proposed PM-LSH and the main competitors for both NN and CP search. Section \ref{sec:relatedwork} reviews related work. Finally, Section \ref{sec:conclusion} concludes the paper.

\section{Preliminaries}\label{sec:problem}
We proceed to present the problem definitions of approximate nearest neighbor (NN) and closest pair (CP) search, and the basic idea of LSH. Frequently used notation is summarized in Table \ref{tb:notation}.

\begin{table}
	\caption{Summary of Notations}
	\label{tb:notation}
	\centering
	\begin{tabular}{cl}
		\toprule
		\textbf{Notation} & \textbf{Definition}\\
		\midrule
		$\mathcal{D}$ & Dataset of points in $\mathbb{R}^d$ \\
		$n = |\mathcal{D}|$ & Dataset cardinality \\
		$d$ & Dataset dimensionality \\
		$o$ & A point in $\mathcal{D}$ \\
		$o'$ & A point $o$ in the projected space \\
		$c$ & Approximation ratio \\
		$h(o)$, $h^*(o)$ & Hash functions \\
		$m$ & The number of hash functions\\
		$T$ & The number of candidate points or pairs\\
		$M$ & The node capacity of the PM-tree \\
		\bottomrule
	\end{tabular}
\end{table}

\subsection{Problem Definition}
Let $\mathcal{D}$ be a set of points in $d$-dimensional space $\mathbb{R}^d$ with cardinality $|\mathcal{D}|=n$. Let $\|o_1,o_2\|$ denote the Euclidean distance between points $o_1, o_2 \in \mathcal{D}$.
We define approximate nearest neighbor and closest pair queries in turn.

\begin{figure*}
	\centering
	\begin{minipage}[c]{0.3\textwidth}
		\centering
		\subfigure[Original Space]{\includegraphics[width=\linewidth]{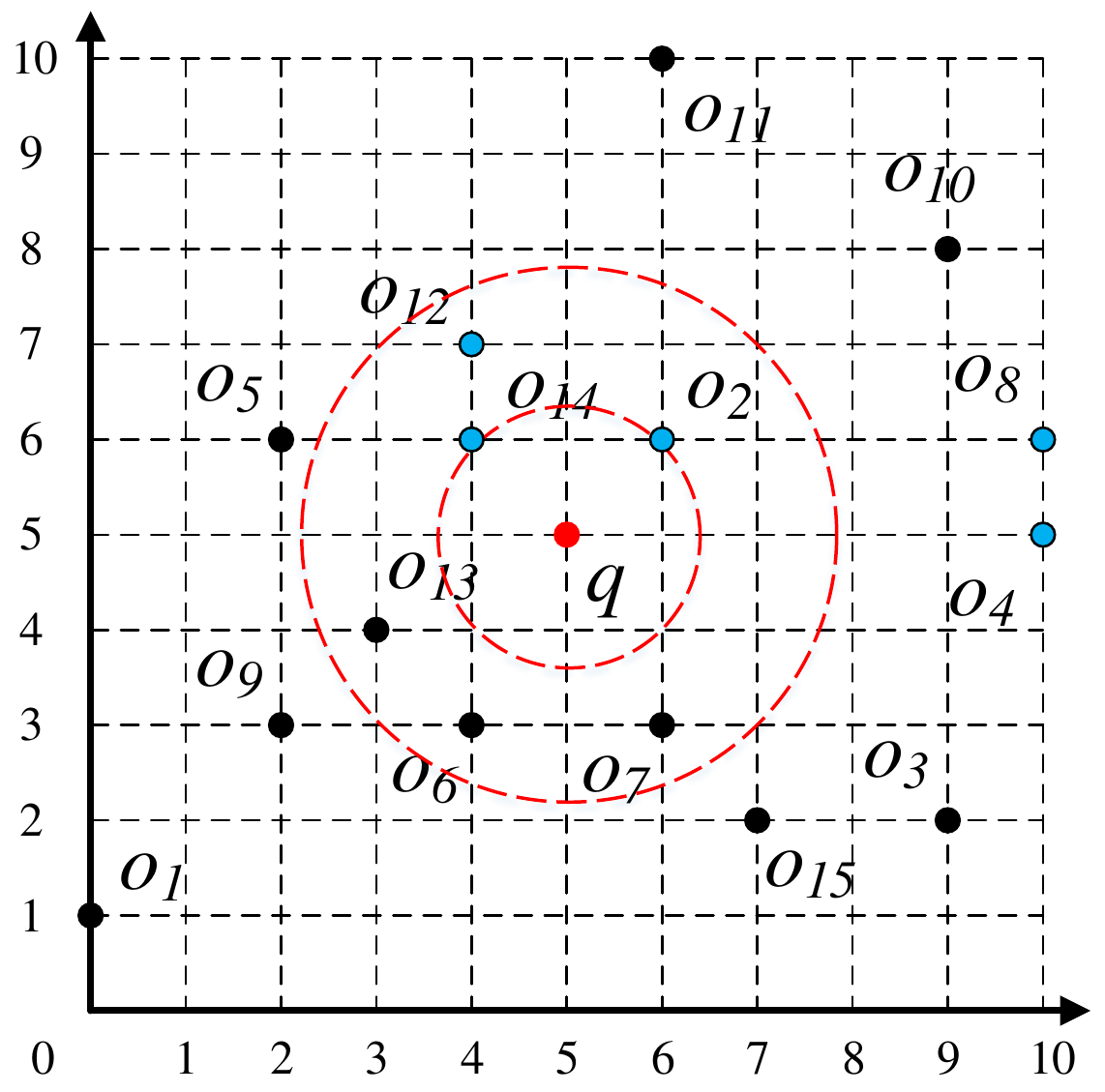}\label{fig:running1}}
	\end{minipage}\hspace{.1in}
	\begin{minipage}[c]{0.27\textwidth}
		\centering
		\subfigure[Projected Space]{\includegraphics[width=\linewidth]{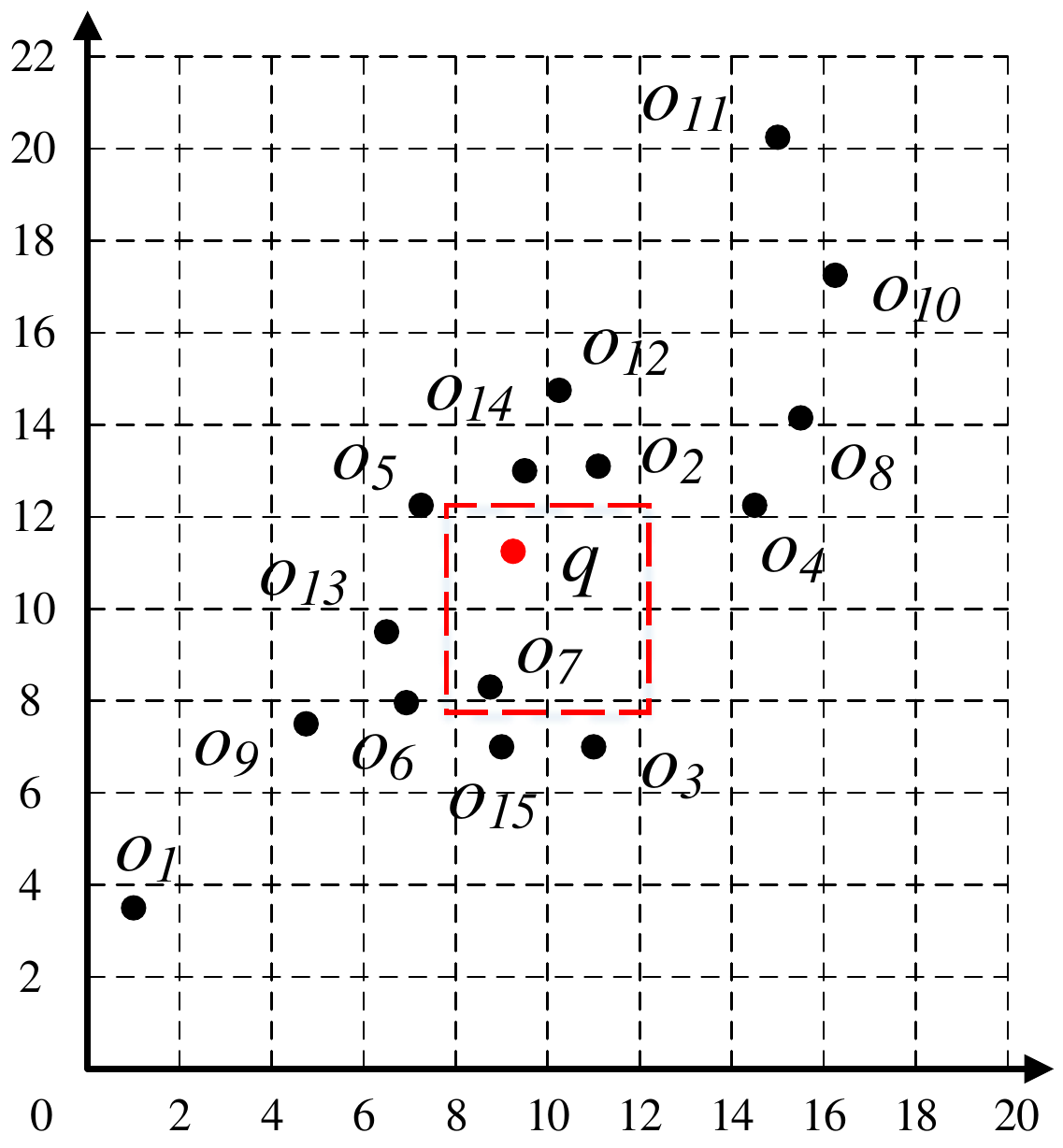}\label{fig:running2}}
	\end{minipage}\hspace{.1in}
	\begin{minipage}[c]{0.33\textwidth}
		\subfigure[Points and Hash Values]{\includegraphics[width=\linewidth]{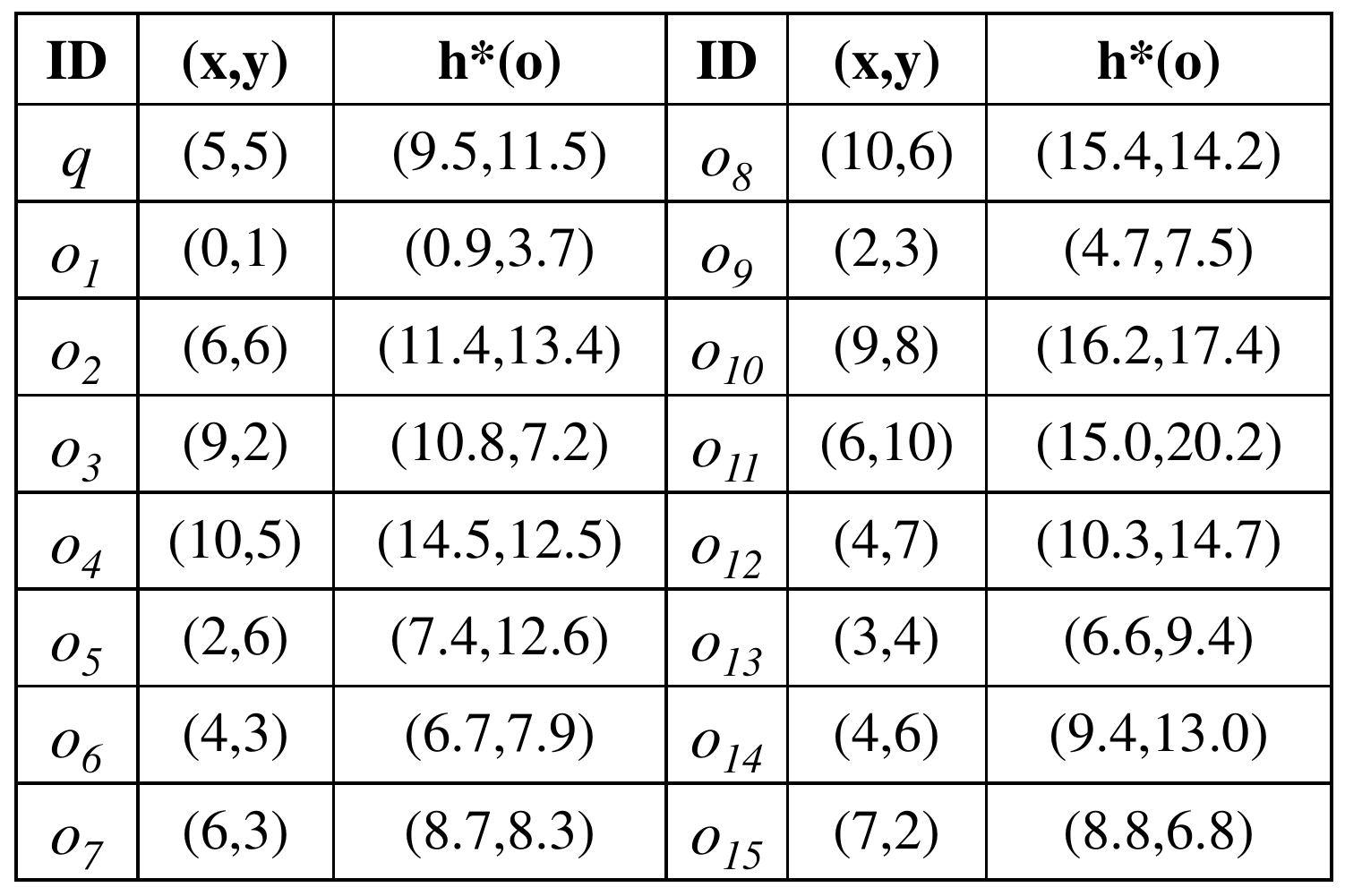}\label{fig:running3}}
	\end{minipage}
	\caption{Running Example with $h_1(o)=\lfloor \frac{\vec{a_1} \cdot \vec{o}}{4} \rfloor$, $h_2(o)=\lfloor \frac{\vec{a_2} \cdot \vec{o}+2}{4} \rfloor$ and $\vec{a_1}=[1.0, 0.9]$, $\vec{a_2}=[0.2, 1.7]$}\label{fig:running}
\end{figure*}

\begin{definition}[$c$-ANN Query]\label{df:cannquery}
Assume a query point $q$ and an approximation ratio $c>1$, and let $o^*$ be the exact nearest neighbor of $q$ in $\mathcal{D}$. A $c$-approximate nearest neighbor query returns a point $o \in \mathcal{D}$ such that $\|q,o\| \leq c \cdot \|q,o^*\|$.
\end{definition}

We generalize the $c$-ANN query to the $(c,k)$-ANN query that returns $k$ approximate nearest points.

\begin{definition}[$(c,k)$-ANN Query]\label{df:ckannquery}
Assume we have a query point $q$, an approximation ratio $c>1$, and a positive integer $k$. Let $o^*_i$ be the $i$-th exact nearest neighbor of $q$ in $\mathcal{D}$. A $(c,k)$-approximate nearest neighbor query returns a sequence of $k$ points $\langle o_1,o_2,\dots,o_k\rangle$ such that for each $o_i$, we have $\|q,o_i\| \leq c \cdot \|q,o^*_i\|$, $i \in [1,k]$.
\end{definition}

\begin{definition}[$c$-ACP Query]\label{df:acpquery}
Assume we have an approximate ratio $c>1$, and let $(o_1^*,o_2^*)$ be the exact closest pair in $\mathcal{D}$. A c-approximate closest pair query returns a point pair $(o_1,o_2) \in \mathcal{D} \times \mathcal{D}$ such that $\|o_1,o_2\| \leq c \cdot \|o_1^*,o_2^*\|$.
\end{definition}

We generalize the $c$-ACP query to the $(c,k)$-ACP query that returns $k$ approximate closest pairs.

\begin{definition}[$(c,k)$-ACP Query]\label{df:ckacpquery}
Assume we have an approximate ratio $c>1$, and a positive integer $k$. Let $(o_{i,1}^*,o_{i,2}^*)$ be the $i$-th exact closest pair in $\mathcal{D}$. A $(c,k)$-approximate closest pair query returns a sequence of $k$ point pairs $\langle (o_{1,1},o_{1,2}),(o_{2,1},o_{2,2}),\dots,(o_{k,1},o_{k,2}) \rangle$ such that for each $(o_{i,1},o_{i,2})$, we have $\|o_{i,1},o_{i,2}\| \leq c \cdot \|o_{i,1}^*,o_{i,2}^*\|$, $i \in [1,k]$.
\end{definition}

\begin{myexample}
As shown in Fig. \ref{fig:running1}, the exact NNs of query $q$ are $o_2$ and $o_{14}$ with distance $\sqrt{2}$. For a $2$-$\mathit{ANN}$ query, any point whose distance to $q$ is within $2\sqrt{2}$ can be considered as a result, i.e., any object in the set $\{o_2, o_{14}, o_{12}, o_{13}, o_6, o_7\}$.

The exact CPs are $(o_4, o_8)$ and $(o_{12}, o_{14})$ with distance $1$. For a $2$-ACP query, any point pair whose distance is within $2$ can be considered as a result, i.e., any pair in the set $\{( o_{6}, o_{7} )$,$( o_{4}, o_{8} )$,$( o_{6}, o_{9} )$,$( o_{6}, o_{13} )$,$( o_{9}, o_{13} )$,\\$( o_{2}, o_{14} )$,$( o_{5}, o_{14} )$,$( o_{12},o_{14} )$,$( o_{3}, o_{15} )$,$( o_{7}, o_{15} )\}$.
\end{myexample}

\subsection{Basic Locality Sensitive Hashing}
We first introduce the LSH scheme, and then explain how to answer the $(r,c)$-ball cover and $c$-ANN queries using the basic LSH \cite{e2lsh,DBLP:conf/compgeom/DatarIIM04}.

\textbf{Hash Family}.
Given a distance $r$, an approximation ratio $c > 1$, two probability values $p_1$ and $p_2$, where $p_1 > p_2$, a family $\mathcal{H}=\{h:\mathbb{R}^d \to U \}$ is called $(r, cr, p_1, p_2)$-locality sensitive, if for any $o_1, o_2 \in \mathbb{R}^d$, it satisfies both of the following conditions:
\begin{enumerate}
	\item If $\|o_1, o_2\| \leq r$ then $Pr[h(o_1)=h(o_2)] \geq p_1$
	\item If $\|o_1, o_2\| \geq cr$ then $Pr[h(o_1)=h(o_2)] \leq p_2$
\end{enumerate}
A well-adopted hash function is formally defined as follows:
\begin{equation}\label{eq:hashfunction}
h(o) = \lfloor \frac{\vec{a} \cdot \vec{o} + b}{w} \rfloor,
\end{equation}
where $\vec{o}$ is the vector representation of a point $o \in \mathbb{R}^d$, $\vec{a}$ is a $d$-dimensional vector where each dimension is drawn independently from a $p$-stable distribution \cite{DBLP:conf/compgeom/DatarIIM04}, $b$ is a real number uniformly and randomly drawn from $[0,w)$, and $w$ is a user-specified constant. The 2-stable distribution is the normal distribution.

Formally, let $\tau = \|o_1,o_2\|$ and let $f(\cdot)$ denote the normal probability distribution function (pdf). We then have:
\begin{equation}\label{eq:collidprob}
	p(\tau)=Pr[h(o_1)=h(o_2)]=\int_{0}^{w} \frac{1}{\tau} \cdot f(\frac{t}{\tau}) \cdot (1-\frac{t}{w})~dt
\end{equation}
The intuition behind Eq. \ref{eq:collidprob} is that, given a fixed $w$, the collision probability of two hash values $h(o_1)$ and $h(o_2)$ grows as the distance $\|o_1,o_2\|$ decreases. Therefore, $h(\cdot)$ in Eq. \ref{eq:hashfunction} is $(r,cr,p_1,p_2)$-sensitive with $p_1=p(r)$ and $p_2=p(cr)$. 

Before we consider how to answer the $c$-ANN query, we define an $(r,c)$-ball cover query that can be answered directly by an $(r,cr,p_1,p_2)$-sensitive hash family.
\begin{definition}[$(r,c)$-BC Query]
Assume a query point $q$, a distance threshold $r$, and an approximation ratio $c > 1$. Let $B(q, r)$ denote a ball centered at $q$ with radius $r$. An $(r,c)$-ball cover query returns the following result:
\begin{enumerate}
	\item If $B(q, r)$ covers at least one point in $\mathcal{D}$, it returns a point in $B(q, cr)$;
	\item If $B(q, cr)$ covers no points in $\mathcal{D}$, it returns nothing.
\end{enumerate}
\end{definition}

E2LSH \cite{e2lsh} is a seminal solution that forms $L$ hash tables and randomly chooses $m$ hash functions for each hash table. By concatenating the $m$ hash functions, a compound hash function $G(o)=(h_1(o),\dots,h_m(o))$ is formed in each hash table, and each point $o \in \mathcal{D}$ is stored in a hash bucket based on $G(o)$. Given a query point $q$, E2LSH computes $G(q)$ and enumerates the points in the corresponding hash bucket. In all $L$ hash tables, it examines at most  $3L$ points and returns a point $o$ if $\|q,o\| \leq cr$. By setting $m=\log_{1/p_2}n$ and $L=1/p_1^k$, the $(r,c)$-BC query can be answered correctly with at least constant probability.

\textbf{From $(r,c)$-BC to $c$-ANN}.
It is easy to see that the ball cover query can be considered as a decision version of the approximate NN query. Processing a sequence of $(r,c)$-BC queries with $r=1,c,c^2,\dots,x$, once a point is returned, we take it as a result of the ANN query. Interestingly \cite{DBLP:conf/stoc/IndykM98}, the ANN query can be answered with approximation ratio $c^2$, i.e., $c^2$-ANN.

\begin{myexample}
In the example in Fig. \ref{fig:running}, we choose $m=2$ hash functions $h_1(o)=\lfloor \frac{\vec{a_1} \cdot \vec{o}}{4} \rfloor$, $h_2(o)=\lfloor \frac{\vec{a_2} \cdot \vec{o}+2}{4} \rfloor$ with $\vec{a_1}=[1.0, 0.9]$, $\vec{a_2}=[0.2, 1.7]$, $b_1=0$, $b_2=2$, and $w=4$. For simplicity, we only construct $L=1$ hash table. Figs. \ref{fig:running2} and \ref{fig:running3} show the coordinates of the objects in the projected space. To answer a $(1,2)$-BC query with $r=1$ and $c=2$, we first compute $G(q)=(h_1(q),h_2(q))=(2,2)$. Then we search the hash bucket $(2,2)$ that is indicated by a red rectangle; the $(1,2)$-BC query returns $o_7$. As $o_{14}$ is the exact NN with $\|q,o_{14}\|=\sqrt{2}$ and
$\|q,o_7\|=\sqrt{5} < 4 \times \sqrt{2}$, we have that $o_7$ is a result of the $4$-ANN query for $q$.
\end{myexample}
\section{A Unified Interpretation of LSH} \label{sec:competitors}
We proceed to introduce the main competitors and give a unified interpretation.

\subsection{Main Competitors}
\textbf{Probing Sequence (PS)}.
The representative PS methods include Multi-Probe \cite{DBLP:conf/vldb/LvJWCL07,DBLP:journals/pvldb/LvJWCL17} and GQR \cite{DBLP:conf/sigmod/LiYZXCLNC18} that use a carefully derived probing sequence to examine multiple hash buckets that are likely to contain the nearest neighbors of a query point. Unlike the basic LSH that builds $L$ hash tables and checks only one hash bucket in each hash table, PS probes multiple nearby buckets in order to achieve higher recall with fewer hash tables. Given a query point $q$, PS adopts a ``generate-to-probe" paradigm that iteratively generates the next hash bucket to be examined with the least distance to $q$ in the remaining buckets. 

\textbf{Radius Enlargement (RE)}.
This category mainly includes the LSB-Tree \cite{DBLP:conf/sigmod/TaoYSK09}, C2LSH \cite{DBLP:conf/sigmod/GanFFN12}, and QALSH \cite{DBLP:journals/pvldb/HuangFZFN15}. These do not build multiple hash tables based on different radii. Generally, RE methods build a hash table like the basic LSH and processes a sequence of $(r,c)$-BC queries by enlarging $r=1,c,c^2,\dots,x$ when a $c$-ANN query is issued. Suppose $r_i = c^i$ and $r_0 = 1$. It has been shown \cite{DBLP:conf/sigmod/GanFFN12} that $h^{r_i}(\cdot)=\lfloor \frac{h(\cdot)}{r_i} \rfloor$ is $(r_i,cr_i,p_1,p_2)$-sensitive. 
Instead of building multiple hash tables with corresponding hash functions $h^{r_i}(\cdot)$ to handle $(r_i,cr_i)$-BC queries, RE methods adopt the smart idea of ``virtual rehashing'' to avoid consuming unnecessary space.
For the $(1,c)$-BC query, RE probes the hash bucket $h(q)$. For the remaining $(r_i,cr_i)$-BC queries, RE probes $r_i^m$ hash buckets near $h(q)$ in the original hash table in the $i$-th iteration. Note that among these $r_i^m$ buckets, $r_{i-1}^m$ buckets were already examined in the previous iteration. Interestingly, it is easy to see that the $r_i^m$ hash buckets in the original hash table correspond to the hash bucket $h^{r_i}(q)$ in the hash table w.r.t. $h^{r_i}(\cdot)$.

\textbf{Metric Indexing (MI)}.
SRS \cite{DBLP:journals/pvldb/SunWQZL14} is the state-of-the-art algorithm that projects the points from the original $d$-dimensional space into a lower $m$-dimensional projected space by using $m$ hash functions. It utilizes an R-tree to index the points based on their hash values and uses the Euclidean distance between two points in the projected space to approximate their distance in the original space. The intuition is that the points close to the query point $q$ in the projected space are also likely close to $q$ in the original space. SRS repeatedly calls an incSearch function that utilizes the R-tree to return the next nearest point to $q$ in the projected space. 

\subsection{A Way of Probing}
We proceed to introduce a unified interpretation of existing LSH methods as shown in Fig. \ref{fig:framework}, which consists of three components, namely data partitioning, distance estimation, and point probing.

\begin{figure}
	\centering
	\includegraphics[width=.4\textwidth]{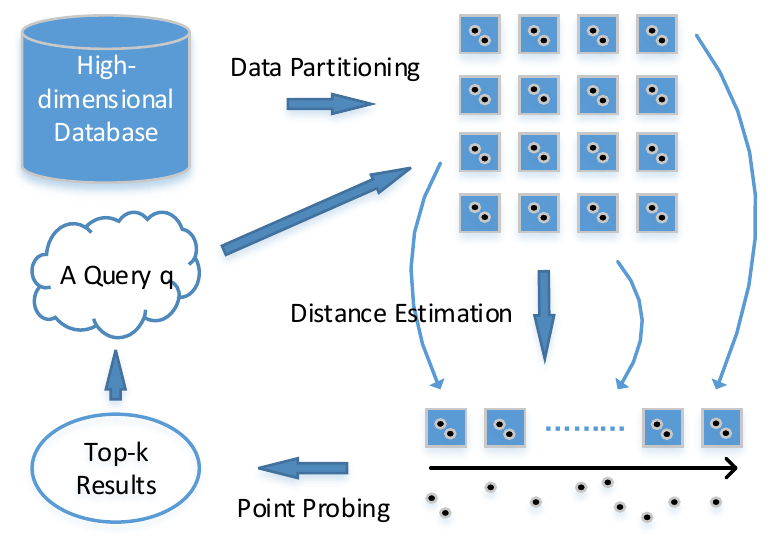}
	\caption{Unified LSH Framework}\label{fig:framework}
\end{figure}

Generally, we adopt a random projection $h(o)$ as the locality sensitive hash functions:
\begin{equation}
h^*(o) = \vec{a} \cdot \vec{o}
\end{equation}
By using $h^*(o)$, the points in the original space are mapped into a projected space, as shown in Figs. \ref{fig:running1} and \ref{fig:running2}. Let $o'=[h^*_1(o),\dots,h^*_m(o)]$ denote point $o$ in the projected space. For any two points $o_1$ and $o_2$, let $r=\|o_1,o_2\|$ and $r'=\|o_1', o_2'\|$ denote the distance between $o_1$ and $o_2$ in the original and in the projected space, respectively. In addition, we let $\rho(o_1,o_2)$ denote an $m$-dimensional vector, where each dimension is the hash value difference between points $o_1$ and $o_2$, i.e., $\rho_i=h_i^*(o_1)-h_i^*(o_2)=o_1'[i]-o_2'[i]$. Therefore, we have $r'=\sqrt{\sum_{i=1}^{m} {\rho_i^2}}$.

Based on a property of a $2$-stable distribution, for any $d$ real numbers $o[1],\dots,o[d]$, independent and identically distributed (i.i.d.) random variables $X_1,\dots,X_d$ (corresponding to $\vec{a}$) following the $2$-stable distribution, $\sum_i o[i] \cdot X_i$ has the same distribution as the variable $(\sum_{i=1}^{d} o[i]^2)^{1/2} \cdot X$, where $X$ is a random variable with distribution $N(0,1)$.
For any two points $o_1$ and $o_2$, since $\rho=h^*(o_1)-h^*(o_2)=\vec{a} \cdot (\vec{o_1} - \vec{o_2})$, we know that $\rho$ is a random variable with distribution $r \cdot X$. In other words, $\rho$ has distribution $N(0, r^2)$, i.e., $\frac{\rho}{r} \sim N(0,1)$. 

\begin{lemma}\label{lemma:X2}
${r'^2}/{r^2}$ follows the distribution $\chi^2(m)$.
\end{lemma}
\begin{proof}
If $Y_1,\dots,Y_m$ are i.i.d. variables with $N(0,1)$ then $\sum_{i=1}^{m} Y_i^2$ follows the $\chi^2$ distribution with $m$ degrees of freedom. Given $m$ hash functions $h_1^*(\cdot),\dots,h_m^*(\cdot)$, for any $o_1$ and $o_2$, we have $\rho_1,\dots,\rho_m$. Thus, ${r'^2}/{r^2}$ follows the distribution $\chi^2(m)$. 
\end{proof}

\textbf{Data Partitioning}.
After mapping the points into the projected space by using hash functions, the existing LSH methods adopt the ``divide-and-conquer'' paradigm that partitions the projected space into subspaces. When a query is issued, the regions that are likely to contain the results are probed, and finally the results of these regions are combined and returned. Generally, there are two kinds of data partitioning approaches in the existing LSH methods:

(1) Interval based Partitioning. The basic LSH constructs hash buckets based on $G(o)$, and each bucket can be viewed as an $m$-dimensional hypercube with equal side lengths $w$. Most of the LSH methods belong to this class, including Multi-Probe, LSB-Tree, C2LSH, and QALSH. Specifically, an LSB-Tree assigns each hypercube a Z-order value and stores the values in a B-tree. In contrast, QALSH does not physically build hypercubes, but stores the values of $h^*(o)$ in a B$^+$-tree. When a query arrives, the length-$w$ intervals are formed virtually on the B$^+$-tree.
	
(2) Metric Space Partitioning. SRS uses an R-tree to index all the points $o'$ in the projected space such that incremental $k$NN search is supported. For in-memory processing, it is also able to use a Cover Tree. In our proposed PM-LSH, we partition the projected space using a PM-tree so that efficient range querying can be supported.

\textbf{Distance Estimation}.
In order to accurately estimate distances, two aspects are considered, i.e., the distance estimator and the estimation granularity.

(1) Distance Estimator.
As $\rho$ follows distribution $N(0,r^2)$, for any $o_1$ and $o_2$, $\rho(o_1,o_2)=[\rho_1,\dots,\rho_m]$.
We estimate the value of $r^2$ by using $r'^2$ as follows.
\begin{lemma}\label{lemma:estimator}
$\hat{r}^2=\frac{r'^2}{m}$ is an unbiased estimator of $r^2$.
\end{lemma}
\begin{proof}
Let $\hat{r}^2$ be the estimated value of $r^2$.
We compute the expectation of $r'^2$ as follows.
$$E[r'^2]=E[\sum_{i=1}^{m} {\rho_i^2}]=\sum_{i=1}^{m} E[{\rho_i^2}]=mr^2$$
Therefore, we have $E[\hat{r}^2]=E[r'^2]/m=r^2$.

We also provide an interesting alternative proof that uses maximum likelihood estimation (MLE) \cite{DBLP:books/lib/HarrisS98}.
MLE is a procedure for finding the value of one or more parameters for a given statistic that maximizes the known likelihood distribution.
As $Pr[\rho = \rho_i]=\frac{1}{\sqrt{2\pi}r } \exp( -\frac{\rho_i^2}{2r^2})$, the probability that the hash value difference $\rho(o_1,o_2)$ between $o_1$ and $o_2$ equals $[\rho_1,$ $\dots,\rho_m]$ is computed as follows.
$$
\begin{aligned}
Pr[\rho(o_1,o_2)&=[\rho_1,\dots,\rho_m]] \\
&=f(\rho_1,\dots,\rho_m | \mu = 0, \sigma = r) \\
&=\prod_{i=1}^{m}(\frac{1}{\sqrt{2\pi} r})^m \exp( -\frac{\sum_{i=1}^m \rho_i^2}{2r^2})\\
\end{aligned}
$$
The objective of the maximum likelihood is to find an $r$ such that the above probability is maximized.
Given $\ln f = -\frac{1}{2}m \ln (2\pi) - m \ln r - \frac{\sum \rho_i^2}{2r^2}$ and $\frac{\partial(\ln f)}{\partial r} = - \frac{m}{r}+\frac{\sum \rho_i^2}{r^3}=0$, we obtain $\hat{r}^2=\frac{\sum_{i=1}^m \rho_i^2}{m}=\frac{r'^2}{m}$.
\end{proof}

\begin{figure}
	\centering
	\subfigure[Recall]{\includegraphics[width=.49\linewidth]{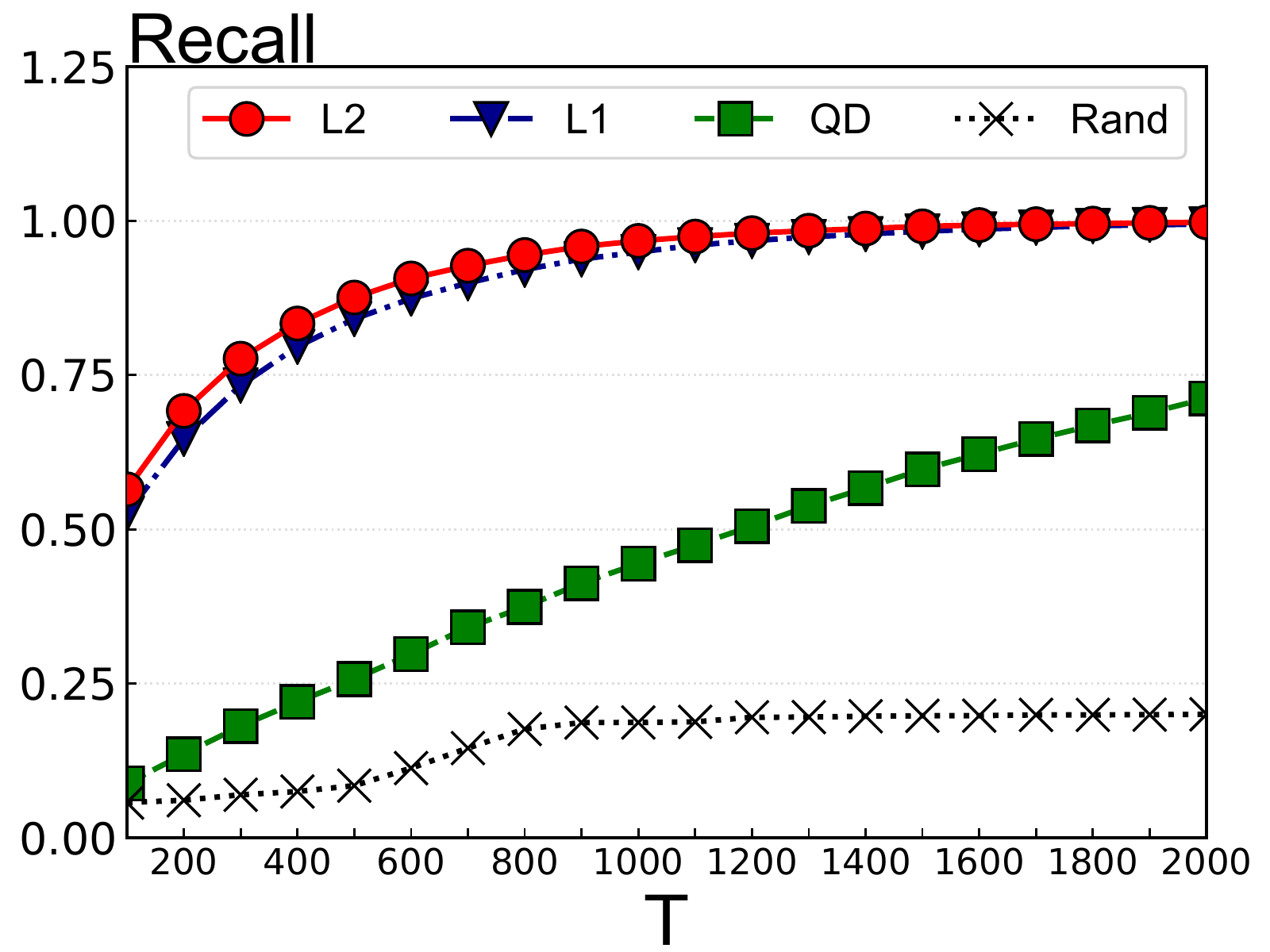}\label{fig:de1}}
	\subfigure[Overall Ratio]{\includegraphics[width=.49\linewidth]{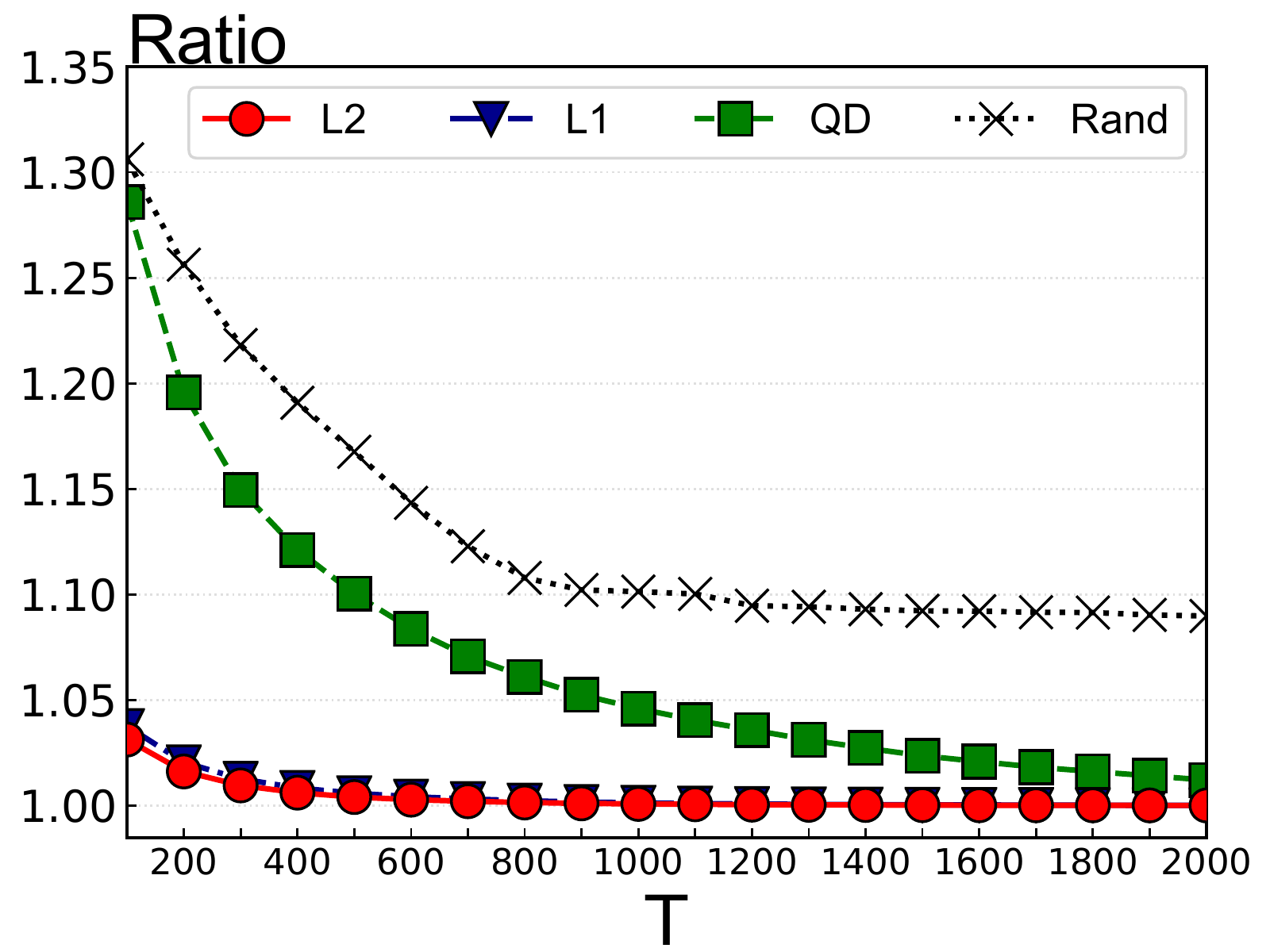}\label{fig:de2}}
	\caption{Comparison on Recall and Overall Ratio}\label{fig:de}
\end{figure}

To evaluate the performance of our estimator in Lemma \ref{lemma:estimator}, i.e., $L_2=r'$ (the same as our estimator when $m$ is fixed), we compare it with three other distance estimators: $L_1$, QD \cite{DBLP:conf/sigmod/LiYZXCLNC18}, and Rand (assign a random value). 
We randomly sample a small dataset that contains 10K points from the \textit{Trevi} dataset \cite{DBLP:journals/tkde/LiZSWLZL20} and choose 100 points as query points. For each query point $q$, we first compute its exact 100NNs. With $m=15$ hash functions, we compute the distances in the projected space between $q$ and all the points based on different estimators. Then, we choose the top-$T$ points with smallest estimated distances ($T$ varies from 100 to 2,000). For each $q$, we compare its exact 100NNs with the 100NNs from the $T$ points.
Finally, we compute the average \textit{recall} and \textit{overall ratio} (discussed in Section \ref{sec:experiments}) of these estimators. As shown in Fig. \ref{fig:de}, we can see that our estimator has the best performance in terms of both the recall and overall ratio. 

(2) Estimation Granularity.
The distance estimation methods may use different granularities:
\begin{itemize}
\item Bucket to Bucket. The hash bucket based indexing methods, such as Multi-Probe, LSB-tree, and C2LSH, store points in hash buckets. When a query is issued, we first find its corresponding bucket and then decide which additional buckets to probe. Therefore, the quality of the distance estimation between buckets is affected by the bucket side length $w$.
\item Point to Bucket. QALSH is an improved version of C2LSH that stores points by a B$^+$-tree instead of using a hash table. When a query $q$ arrives, the length-$w$ intervals are conceptually built on the B$^+$-tree with $q$ as the center. So the distance estimation can be considered as between point $q$ and bucket intervals.
\item Point to Point. SRS uses the projected Euclidean distance between two points to estimate their original distance. This offers a finer precision than the previous two methods due to the fine granularity. Our PM-LSH also adopts this method.
\end{itemize}

\textbf{Point Probing}.
Suppose we probe $T$ points.
In the hash bucket based indexing methods, we directly probe the points in the bucket, where the time cost is $O(T)$.
The second approach is QALSH that searches the points in a B$^+$-tree, where the time cost is $O(\log n + T)$.
Unlike the previous two approaches, SRS indexes the points with an R-tree and iteratively finds the next NN in the projected space. The time cost is $O(\log n \cdot T)$. Our PM-LSH can be considered as a combination of the second and third approaches in that we build a PM-tree in the projected space and execute range queries to retrieve points.
\section{The PM-LSH Framework} \label{sec:ourmethod}
We proceed to present the details of the PM-LSH framework. As mentioned previously, the RE methods quickly probe the points stored in the hash buckets by enlarging the search radius, but suffer from inaccurate distance estimation due to a coarse-grained index structure, which translates into computational overhead when having to examine unnecessary points. 
In contrast, the MI methods index the points with an R-tree and iteratively return the next nearest point to $q$ in the projected space. However, finding the next exact NN in an R-tree is also computationally costly, and the next NN is not necessarily the best next candidate in the original space. To achieve the best of both worlds, PM-LSH combines the ideas of the RE and MI methods. To achieve both efficiency and accuracy, we adopt the PM-tree instead of the R-tree to index the points in the projected space and execute a sequence of range queries with increasingly large radius.

Next, we briefly describe how to construct a PM-tree. Then, we analyze the cost models of the PM-tree and the R-tree to understand how the PM-tree performs better than the R-tree for the relevant range query workload. Finally, we present the details of the algorithms.

\begin{figure}
	\centering
	\subfigure[Space Partitioning]{\includegraphics[width=.5\linewidth]{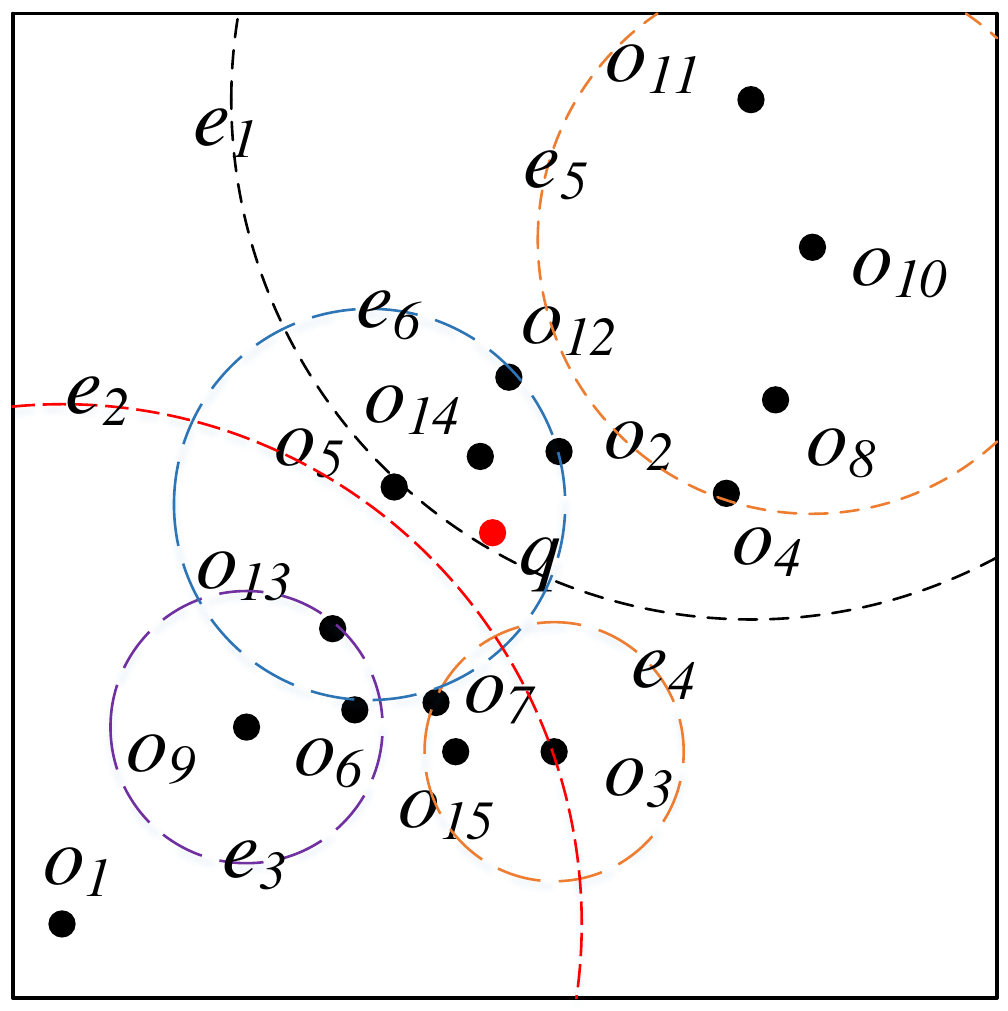}\label{fig:spacepartition}}
	\subfigure[PM-tree]{\includegraphics[width=.98\linewidth]{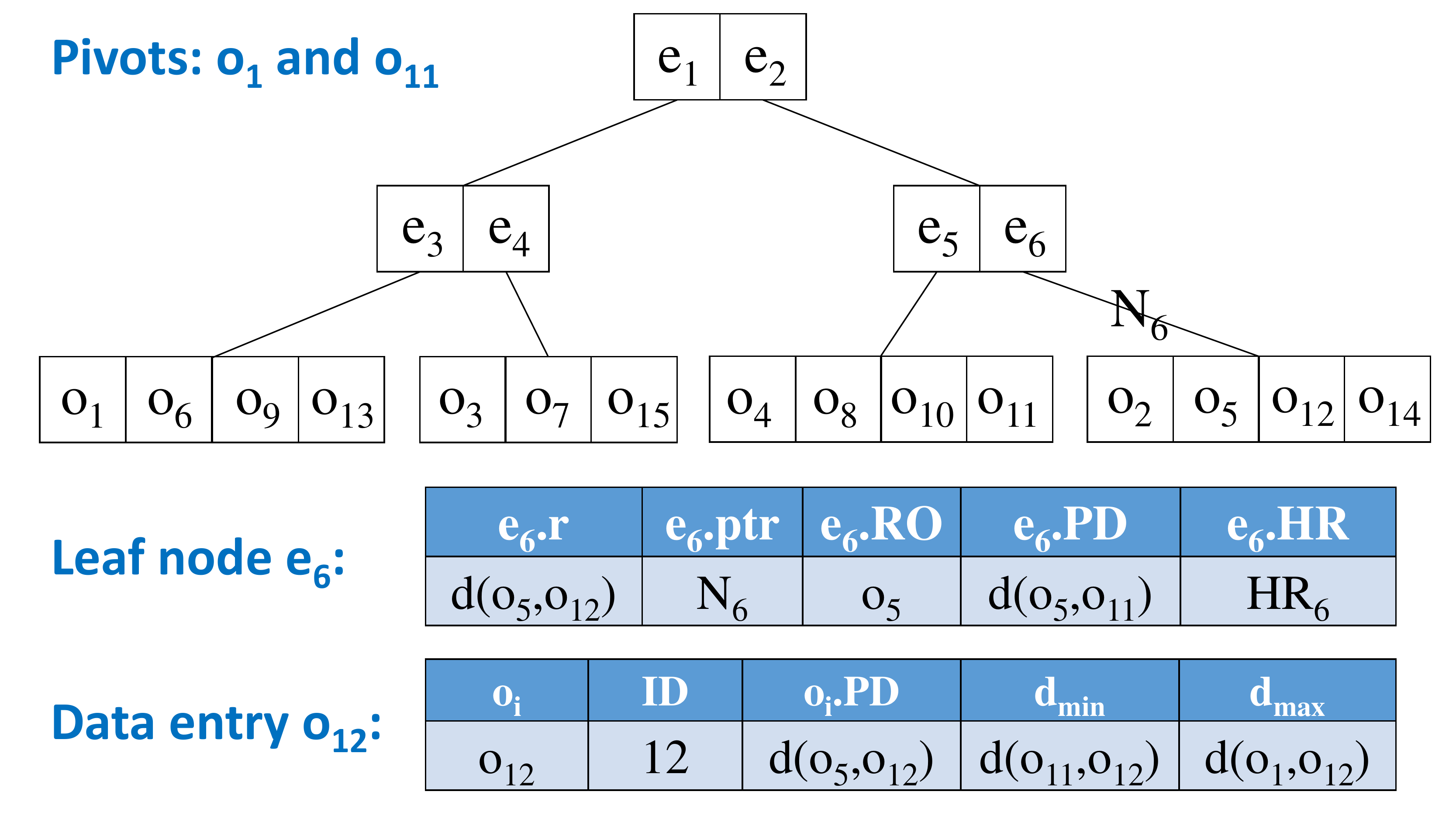}\label{fig:pmtree}}
	\caption{The Structure of PM-LSH}\label{fig:pm}
\end{figure}

\subsection{Building a PM-tree in the Projected Space}
In the projected space, each $o'_i$ w.r.t. $o_i \in \mathcal{D}$ is an $m$ dimensional vector. For the paper to be self-contained, we briefly explain how to build a PM-tree on the $o'_i$s. Interested readers may refer to \cite{DBLP:conf/dasfaa/SkopalPS05} for more details on the PM-tree.

\textbf{Selecting Pivots}.
The PM-tree combines the M-tree with pivot mapping. Methods for selecting an optimal set of pivots have been studied extensively.
For each set of pivots, a PM-tree region is the intersection of the M-tree hyper-spherical region and hyper-rings caused by the pivots.
We choose a set of pivots with the aim of making the overall volume of the corresponding PM-tree region the smallest.

\textbf{PM-tree Structure}.
The structure of a PM-tree is shown in Fig. \ref{fig:pm}.
Being an extension of the M-tree, it retains all the information of the M-tree. 
For each node ${e}$, it stores the covered radius $\sf{e.r}$, a pointer to its covered sub-tree $\sf{e.ptr}$, the center of the covered hyper-sphere $\sf{e.RO}$, the distance $\sf{e.PD}$ between $\sf{e.RO}$ and its parent node, and the smallest interval $\sf{e.HR}$ covering the distances between the pivots and each of the point stored in leaves.
For a data entry $o$, it stores the point data, the ID of the point $o$, the distance $\sf{o.PD}$ between $o$ and its parent entry, and the minimum and the maximum distances to pivots.

\textbf{Range Query Processing}.
A range query, denoted by $\textbf{range}(q,r)$, returns all points that are located in $B(q,r)$. The nodes in the PM-tree are traversed in a depth-first fashion. When a node is accessed, we verify its pruning condition by using the triangle inequality. When a data entry is accessed, we insert the corresponding point into the result set if it is in $B(q,r)$.

\begin{myexample}\label{example:pmtree}
As shown in Fig. \ref{fig:pm}, we choose $o_1$ and $o_{11}$ as pivots, and partition the space by using the ball partitioning, as shown in Fig. \ref{fig:spacepartition}. The nodes $e_1,e_2,\cdots, e_6$ contain the points inside a hyper-sphere, whose center and radius are saved as the part of an entry. When a range query $\textbf{range}(q,2)$ is issued, we check pruning conditions when accessing the nodes. Only $e_4$ and $e_6$ are checked. Finally, we return $\{o_{14}\}$ as the result.
\end{myexample}

\subsection{Cost Models of the PM-tree Versus the R-tree}

To compare the performance of the PM-tree and the R-tree, we adopt a node-based cost model \cite{DBLP:conf/pods/CiacciaPZ98} to examine how the PM-tree performs compared to the R-tree from a theoretical point of view. 

In this cost model, a concept called distance distribution of a dataset $\mathcal{D}$ is computed as follows.
\begin{equation}
F(x)=Pr[\|o_i,o_j\|\leq x],
\end{equation}
where $o_i,o_j \in \mathcal{D}$. 
In addition, for each dataset used in our experiments, we compute its ``homogeneity of viewpoints" (HV), which is shown in Table \ref{tb:datasets}. 
HV evaluates the homogeneity of the distance distributions of the data points. Let $F_{o}(x)$ denote the distribution of the distances between all points to point $o$. Given two points $o_1$ and $o_2$, a higher HV means that $o_1$ and $o_2$ are more likely to have similar distance distributions $F_{o_1}(x)$ and $F_{o_2}(x)$. The HV values of all the datasets are no smaller than 0.9, which enables us to approximate their distance distributions when estimating the cost models of the two trees.

\begin{table*}[t]
	\caption{Computation Cost (CC) of PM-tree and R-tree}
	\label{tb:cc}
	\centering
	\begin{tabular}{cccccccc}
		\toprule
		\textbf{Datasets}& Audio& Cifar& MNIST& Trevi &NUS &GIST &Deep\\
		\midrule
		\midrule
		\textbf{PM-tree} & 38,182 & 35,210 & 56,670 & 34,281 & 201,448 & 739,720 & 964,451\\
		\textbf{R-tree} & 40,565 & 54,869 & 59,043 & 63,884 & 252,187 & 889,974 & 1,017,604\\ 
		\textbf{Reduction}  & 6$\%$ & 36$\%$ & 4$\%$ & 46$\%$ & 20$\%$ & 17$\%$ & 5$\%$ \\
		\bottomrule
	\end{tabular}
\end{table*}

\textbf{Cost Model of the PM-tree}.
Consider a range query \textbf{range}$(q, r_q)$. Assume that a PM-tree has $s$ pivot points $p_1,\cdots,p_s$. A node $e$ is accessed iff the following conditions are satisfied:

\begin{equation}\label{eq:cond_pm}
\begin{cases}
\|q,\textsf{e.RO}\| \leq \textsf{e.r}+r_q\\
\land_{i=1}^{s}\{\|q,p_i\|-r_q \leq \textsf{e.HR}[i].max\}\\
\land_{i=1}^{s}\{\|q,p_i\|+r_q \geq \textsf{e.HR}[i].min\}
\end{cases}
\end{equation}

Therefore, the probability of $e$ being accessed can be computed as follows.

\begin{equation}
\begin{split}
Pr[e]=&F(\textsf{e.r}+r_q)\cdot\prod_{i=1}^s [F(\textsf{e.HR}[i].max+r_q)\\
&-F(\textsf{e.HR}[i].min-r_q)]
\end{split}
\end{equation}
Assume that there are $N$ nodes in the PM-tree. The number of distance computations (computation cost) is estimated by considering the probability that a node is accessed multiplied by its number of entries $N(e)$, thus obtaining the number of distance computations as follows.
\begin{equation}\label{eq:ccpmtree}
\textbf{CC}(\textbf{range}(q, r_q))=\sum_{i=1}^N N(e_i) \cdot Pr[e_i]
\end{equation}

\textbf{Cost Model of the R-tree}.
For each node $e$ of an $m$-dimensional R-tree, we denote its minimum bounding rectangle as $\textbf{MBR}(e)=[l_1,u_1]\times\cdots\times[l_m,u_m]$. Given a range query \textbf{range}$(q, r_q)$, the condition of $e$ being accessed is that $B(q,r_q)$ intersects with $\textbf{MBR}(e)$. Since it is hard to quantify the probability that a ball intersects with a high-dimensional rectangle, we substitute an isochoric hyper-cube for the ball. Specifically, an $m$-dimensional ball with radius $r_q$ is substituted by a hyper-cube with the length of sides $l=\sqrt[m]{\frac{2\pi^{m/2}}{m\Gamma(m/2)}}r_q$ \cite{DBLP:conf/stoc/IndykM98}. We also denote the data distribution of dataset $\mathcal{D}$ on the $i$-th dimension as follows.
\begin{equation}
G_i(x)=Pr[X_i\leq x],
\end{equation}
where $X_i$ is the $i$-th dimension of a random point in $\mathcal{D}$. 
Similarly, we let $N$ be the number of nodes in the R-tree and let $N(e_i)$ be the number of entries in node $e_i$. We obtain the number of distance computations as follows (details are omitted for brevity).
\begin{equation}\label{eq:ccrtree}
\textbf{CC}(\textbf{range}(q, r_q))=\sum_{i=1}^N N(e_i) \cdot\prod_{i=1}^m [G_i(u_i+l)-G_i(l_i-l)]
\end{equation}

\textbf{Comparison of the PM-tree and the R-tree}.
In order to compare the computation costs for the two trees, we construct PM-trees and R-trees for the points in all the datasets (introduced in Table \ref{tb:datasets}) after transforming them into the projected space. We choose $m=15$ hash functions and set the maximum number of entries per node to 16. For each dataset, we choose the same range $r$ to estimate the cost of computing a range query. The value of $r$ is chosen to return approximately the nearest $8\%$ of all points, since these points usually suffice to return a $c$-ANN result. The estimated computation costs are computed based on Eqs. \ref{eq:ccpmtree} and \ref{eq:ccrtree}, and the results are presented in Table \ref{tb:cc}. We can see that using the PM-tree reduces the number of distance computations by about $5\%-46\%$ for the different datasets. This observation offers evidence that the PM-tree has better performance than the R-tree in our setting.

\subsection{Tunable Confidence Interval}
Based on Lemma \ref{lemma:estimator}, we further estimate the confidence interval of $r'$ between $o_1$ and $o_2$ for a given $r=\|o_1,o_2\|$.
\begin{lemma}\label{lemma:range}
	Given two points $o_1$ and $o_2$, we have:
	\begin{itemize}	
		\item \textbf{P1}: The probability that $r'<r\sqrt{\chi^2_{1-\alpha}(m)}$ is $\alpha$
		\item \textbf{P2}: The probability that $r'>r\sqrt{\chi^2_{\alpha}(m)}$ is $\alpha$
	\end{itemize}
	Here, $\chi^2_{\alpha}(m)$ is the upper quantile of a $\chi^2$ distribution with $m$ degrees of freedom, where
	$$
	\int_{\chi^2_{\alpha}(m)}^{+\infty} f(x;m)dx=\alpha,
	$$
	and $f(x;m)$ is the probability density function of a $\chi^2$ distribution with $m$ degrees of freedom.
\end{lemma}

\begin{proof}
	From Lemmas \ref{lemma:X2} and \ref{lemma:estimator}, we know $\frac{r'^2}{r^2}\sim \chi^2(m)$. Constructing a confidence interval $I=[u,v]$ for $\frac{r'^2}{r^2}$ requires that the probability that $\frac{r'^2}{r^2}$ falls into $I$ is $1-2\alpha$ for any given $\alpha$. A standard approach is to select $u$ and $v$ that make $Pr[\frac{r'^2}{r^2}<u]=\alpha$, i.e., $Pr[\frac{r'^2}{r^2}>u]=1-\alpha$, and $Pr[\frac{r'^2}{r^2}>v]=\alpha$. Further, $\int_{u}^{+\infty} f(x;m)dx=1-\alpha$ and $\int_{v}^{+\infty} f(x;m)dx=\alpha$. According to the definition of upper quantile, we have $u=\chi^2_{1-\alpha}(m)$ and $v=\chi^2_{\alpha}(m)$. The confidence interval and its corresponding probability are shown in Fig. \ref{fig:ci}.
\end{proof}

\begin{figure}
	\centering
	\includegraphics[width=.3\textwidth]{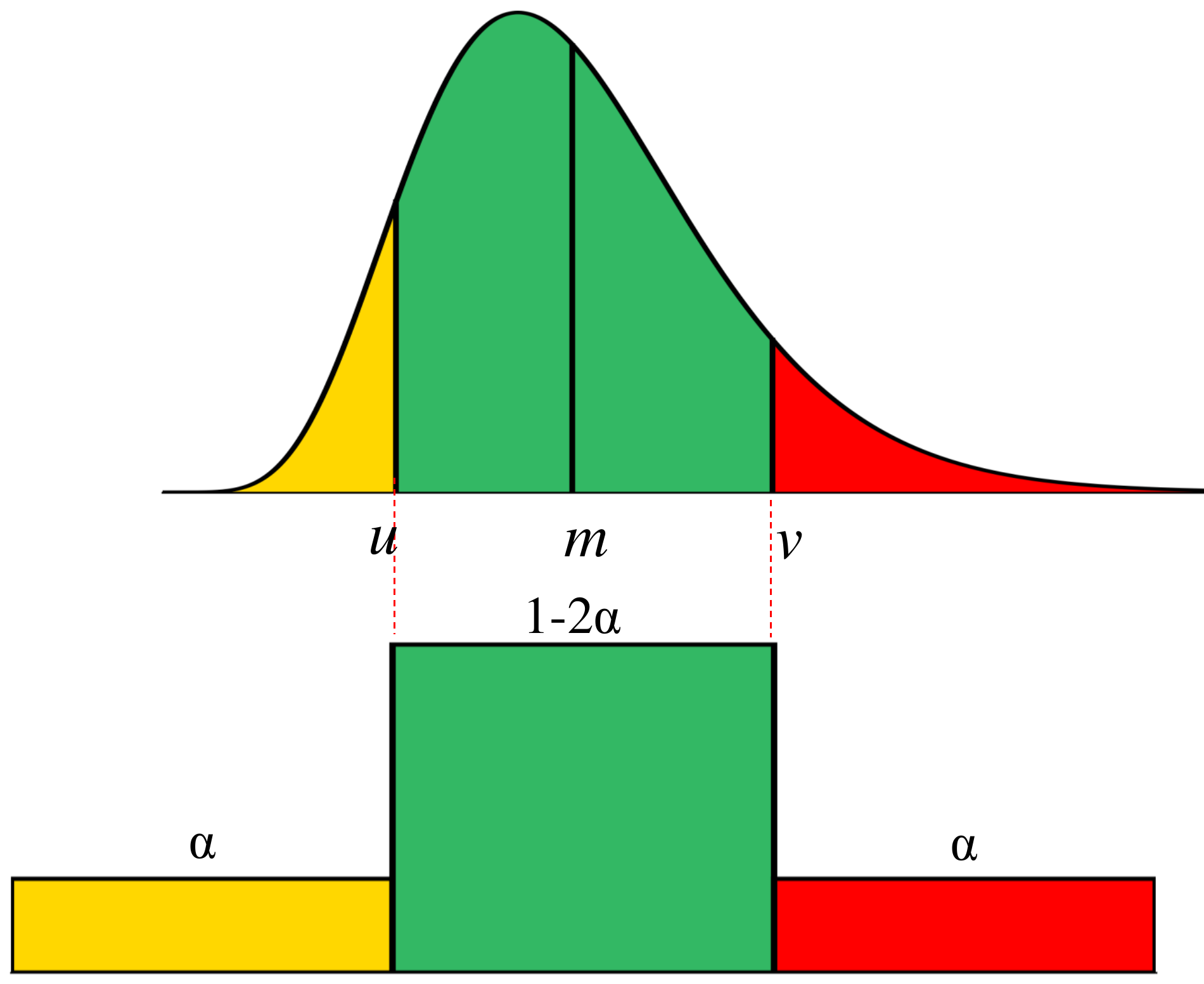}
	\caption{Confidence Interval}\label{fig:ci}
\end{figure}

Lemma \ref{lemma:range} establishes a strong relationship between an original distance and the confidence interval of the corresponding projected distance, which we use to answer $(r,c)$-BC and $c$-ANN queries. 

\section{Nearest Neighbor Query Processing}\label{sec:nn}
We proceed to cover the nearest neighbor query processing based on PM-LSH. First, we present the details of $(r,c)$-BC query processing. Then, we extend the coverage to include $(c,k)$-ANN query processing.

\subsection{The $(r,c)$-BC Query}
An $(r,c)$-BC query can be computed directly by Algorithm \ref{algo:rcbc}. Given a query $q$ and $m$ hash functions, we compute the hash value $q'=(h_1^*(q),\dots,h_m^*(q))$ and use the PM-tree to answer a range query $\textbf{range}(q',tr)$, where $t$ is a parameter that guarantees that a point inside $B(q,r)$ in the original space will fall into $B(q',tr)$ in the projected space with a constant probability. Then we collect the result of the range query into a candidate set $C$. 

According to Lemma \ref{lemma:e1e2}, introduced in Section \ref{sec:analysis}, the correctness of the $(r,c)$-BC query can be guaranteed. In other words, by properly choosing a parameter $\beta$, we examine a sufficient number of $\beta n$ candidate points such that the following two situations hold with constant probability.

\begin{itemize}	
	\item
	If the total number of points in $C$ exceeds $\beta n$, there must be at least one point from $C$ in $B(q,cr)$.
	\item
	If no point from $C$ is in $B(q,cr)$, there exists no point in $\mathcal{D}$ is in $B(q,r)$.
\end{itemize}
Therefore, we can correctly answer an $(r,c)$-BC query by processing a range query using the PM-tree. In Section \ref{sec:analysis}, we consider how to set parameters $t$ and $\beta$.

\begin{algorithm}[t]
	\caption{$(r,c)$-BC Query}\label{algo:rcbc}
	\KwIn{A query point $q$ and parameters $\beta$, $n$, $t$, $c$, $r$}
	\KwOut{A point $p$ in $B(q,cr)$ or nothing}
	Compute $q'=(h_1^*(q),\dots,h_m^*(q))$;\\
	Initialize a candidate set $C \gets$ the results of a range query $q'$ with radius $t \cdot r$ on the PM-tree;\\
	\If{$|C| \geq \beta n + 1$}{\Return{$p$} in $C$ that is closest to $q$;}
	\Else{
		\If{$|\{p~|~p \in C \land {\|p,q\|}\leq c \cdot r\}| \geq 1$}{\Return{$p$} in $C$ that is closest to $q$;}
		\Else{\Return{$\varnothing$;}}
	}
\end{algorithm}

\subsection{The $(c,k)$-ANN Query}
Answering a $c$-ANN query is more complicated than answering an $(r,c)$-BC query since we do not know the distance $\|q,o^*\|$ in advance. In order to answer a $(c,k)$-ANN query with a constant probability, we must ensure that we access enough points, i.e., at least $\beta n$ points. Therefore, we have to enlarge the search radius in the projected space when fewer than $\beta n$ points are found until $k$ points inside $B(q,cr)$ have been obtained.

The details of computing a $(c,k)$-ANN query can be found in Algorithm \ref{algo:ckann}. 
Most of the steps are similar to ones in Algorithm \ref{algo:rcbc}. The difference is that when both termination conditions (Lines 4 and 8) are violated, another range query with a larger radius is required.

\textbf{Selecting the Radius $r$ of a Range Query}.
As executing multiple range queries is time consuming, it is attractive to reduce the number of iterations in the while-loop. Intuitively, we hope to find a ``magic'' $r_{min}$ such that the process terminates quickly. An ideal $r_{min}$ must yield a number of points inside $B(q',tr_{min})$ that exceeds $\beta n+k$ such that Algorithm \ref{algo:ckann} is able to terminate after processing the range query $B(q',tr_{min})$. In addition, to avoid returning a large number of unnecessary points, which also is costly, the number of points inside $B(q',tr_{min}/c)$ should be below $\beta n+k$. Otherwise, a range query $B(q',tr_{min}/c)$ with a smaller radius is able to return enough points.

As the $r_{min}$ can be selected from a relatively large range, we design a selection scheme as follows. 
Suppose that we have obtained the distance distribution $F(x)$ of all datasets. Due to a good HV value, the distance distribution of a query point can be estimated for the dataset. Then we can find a suitable $r$ that satisfies $n \cdot F(r)=\beta n+k$, which implies that $\beta n+k$ points locate in $B(q,r)$ on average. However, to avoid the case where the number of points in $B(q,r)$ exceeds $\beta n+k$, we choose an $r_{min}$ that is slightly smaller than $r$. 
As the choice of $r_{min}$ is not unique and the selection range is relatively large, and since the performance is not strongly dependent on the specific choice, the effect of the estimation is expected to be small.

\begin{myexample}\label{example:cann}
	Setting $\beta n = 4$, we need to retrieve at least $5$ points for a $(2,1)$-ANN query. Initially, we set $r_{min}=r'=2$. As explained in Example \ref{example:pmtree}, $o_{14}$ is returned.
	As the number of returned points is below $5$, we set $r'=4$. In this round, only the subtree of $e_5$ can be discarded, and we check the points in $e_3$, $e_4$, and $e_6$ and obtain $\{o_2,o_5,o_7,o_{12},o_{13},o_{14}\}$.
	The number of returned points is $6$, and the process terminates. Finally, we return the $(2,1)$-ANN result $o_{14}$.
\end{myexample}

\begin{algorithm}[t]
	\caption{$(c,k)$-ANN Query}\label{algo:ckann}
	\KwIn{A query point $q$, and parameters $r_{min}$, $\beta$, $n$, $t$, $c$, $k$}
	\KwOut{$k$ points}
	Initialize a candidate set $C \gets \varnothing$ and $r \gets r_{min}$;\\
	Compute $q'=(h_1^*(q),\dots,h_m^*(q))$;\\
	\While{true}{
		\If{$|\{p~|~p \in C \land {\|p,q\|}\leq c \cdot r\}| \geq k$}{\Return{top-$k$} points closest to $q$ in $C$;}
		Initialize a range query $q'$ with radius $t \cdot r$ in the PM-tree;\\
		\While{$|C| < \beta n + k$}{
			Find a node in $B(q',t \cdot r)$ on the PM-tree;
			$C \gets\; C\cup\{$the points in the node$\}$;\\
		}
		\If{$|C| \geq \beta n + k$}{\Return{top-$k$} points closest to $q$ in $C$;}
		$r \gets c \cdot r$;\\
	}
\end{algorithm}

\subsection{Analysis}\label{sec:analysis}
\textbf{Quality Guarantee}.
In Algorithms \ref{algo:rcbc} and \ref{algo:ckann}, we execute a range query on the PM-tree with a radius $tr$ in the projected space. Therefore, we have to compare the projected distances of candidate points to $q$ with $tr$. Specifically, two types of points need to be discussed: true positives (the points inside $B(q,r)$) and false positives (the points outside $B(q,cr)$).

\begin{lemma}\label{lemma:e1e2}
	Given a query $q$, we set probabilities $\alpha_1$ and $\alpha_2$, and parameter $t$ such that they satisfy Eq. \ref{eq:paramsetting}:
	
	\begin{equation}\label{eq:paramsetting}
	\begin{cases}
	t^2=\chi^2_{\alpha_1}(m)\\
	t^2=c^2\chi^2_{1-\alpha_2}(m)
	\end{cases}
	\end{equation}
	
	We then have:
	\begin{itemize}	
		\item \textbf{E1:} 
		If a point $o$ exists in $B(q,r)$, its projected distance to $q$ is smaller than $tr$.
		\item \textbf{E2:} 
		There are fewer than $\beta n$ $(\beta>\alpha_2)$ points outside $B(q,cr)$ whose projected distances to $q$ are smaller than $tr$.
	\end{itemize}
	
	The probability that E1 occurs is at least $1-\alpha_1$, and the probability that E2 occurs is at least $1-\frac{\alpha_2}{\beta}$.
\end{lemma}

\begin{proof}
	Given a point $o\in B(q,r)$, let $r_o=\|o,q\| \leq r$ and $r_o'=\|o',q'\|$ be the original and projected distances to $q$, respectively. By setting $t = \sqrt{\chi^2_{\alpha_1} (m)}$, according to the Lemma \ref{lemma:range}, we have $Pr[r_o' > r_o \sqrt{\chi^2_{\alpha_1} (m)}]= Pr[r_o' > tr_o] =\alpha_1$. Since $r_o \leq r$, $Pr[r_o' > tr]$ is at most $\alpha_1$. Therefore, we know that $Pr[E1] = Pr[r_o' \leq tr] > 1-\alpha_1$.
	Likewise, given a point $o \notin B(q,cr)$, let $r_o=\|o,q\| > cr$ and $r_o'=\|o',q'\|$ be the original and projected distances to $q$, respectively. By setting $t = c\sqrt{\chi^2_{1-\alpha_2}(m)}$, according to the Lemma \ref{lemma:range}, we have $Pr[r_o'<r_o\sqrt{\chi^2_{1-\alpha_2}(m)}]=Pr[r_o'< t\frac{r_o}{c}]=\alpha_2$. Since $\frac{r_o}{c} > r$, $Pr[r_o' < tr]$ is at most $\alpha_2$. Therefore, by using Markov's inequality, we have $Pr[E2] >1-\frac{\alpha_2}{\beta}$.
\end{proof}

Note that if $E1$ and $E2$ hold at the same time, then Algorithm \ref{algo:rcbc} computes the $(r,c)$-BC query correctly.

\begin{lemma}\label{lemma:rcbc}
	Algorithm \ref{algo:rcbc} answers an $(r,c)$-BC query with at least a constant probability.
\end{lemma}
\begin{proof}
	Let $m=O(1)$. If $\alpha_1$ is a constant, $\alpha_2$ is also a constant due to Eq. \ref{eq:paramsetting}. By setting $\beta=2\alpha_2$, the lower bound probabilities of $E1$ and $E2$, i.e., $1-\alpha_1$ and $1-\frac{\alpha_2}{\beta}$, will also be constant. Therefore, we can guarantee that $E1$ and $E2$ hold at the same time with at least a constant probability. Thus, if we access at least $\beta n+1$ points with projected distances to $q$ smaller than $tR$, due to $E2$, there are at most $\beta n$ points outside $B(q,cr)$, and we thus obtain at least one point inside $B(q,cr)$. On the other hand, if we access no more than $\beta n+1$ points with projected distances to $q$ smaller than $tR$, the correctness of $E2$ is not guaranteed. Therefore, it is safe to return either no points or the points whose distances to $q$ are at most $cr$ for an $(r,c)$-BC query.
\end{proof}

As a typical setting in the LSH methods, we choose parameters that satisfy $Pr[E1]=1-1/e$ and $Pr[E2]=1/2$. Note that we can also choose parameters that achieve more accurate results. In our setting, we have $\alpha_1=1/e$ and $t=\sqrt{\chi^2_{\alpha_1}(m)}$. Based on Eq. \ref{eq:paramsetting}, both $\alpha_2$ and $\beta$ can be determined easily.

\begin{theorem}\label{the:cknn}
	Algorithm \ref{algo:ckann} returns a $c^2$-ANN with probability at least $1/2-1/e$.
\end{theorem}
\begin{proof}
	Due to Lemma \ref{lemma:rcbc}, we find that $E1$ and $E2$ can hold at same time with probability at least $1/2-1/e$ in our setting. Now we show that when $E1$ and $E2$ hold, the output of Algorithm \ref{algo:ckann} is $c^2$-approximate. We denote the set of points whose projected distances to $q$ are smaller than $tr$ as $C(r)$. When enlarging $r$ according to the sequence $1,c,c^2,\cdots$, there must exist a radius $r_{opt}$ that makes $\lvert C(r_{opt})\rvert \geq 1+\beta n$ and $\lvert C(r_{opt}/c)\rvert< 1+\beta n$ hold. Then, if $r^*=||o^*,q||\leq r_{opt}/c$, its projected distance to $q$ is smaller than $tr_{opt}/c$ according to $E1$, and we must have found it in $C(r_{opt})$ because $C(r_{opt})\supset C(r_{opt}/c)$. As a result, Algorithm \ref{algo:ckann} returns the exact NN. If $r=||o^*,q|| > r_{opt}/c$, according to $E2$, there is at least one point in $C(r_{opt})$ whose distance to $q$ is at most $cr_{opt}$. Therefore, we return a point whose distance to $q$ is smaller than $c^2r^*$.
\end{proof}

\noindent
\textbf{Algorithm Analysis of PM-LSH}.
In PM-LSH, if we choose a large $m$, it will be costly to process a sequence of range queries in the projected space. So we consider $m$ as a constant and fix its value at 15 in all experiments.

\begin{theorem}
	PM-LSH has space cost $O(n)$ and time cost $O(\log n+\beta n)$, where $\beta$ is much smaller than 1.
\end{theorem}

\begin{proof}
	The space consumption is due mainly to the PM-tree, which has $n$ items. Each item consumes $m+O(1)$ space, so the overall space consumption is $O(n)$ as $m=O(1)$. The query time cost comes from two parts: 1) finding candidate points in the PM-tree; and 2) verifying the real distances of candidate points to $q$. The former has cost $O(\log n)$, and the latter has cost $O(\beta n)$ when $d$ is considered as a constant. Therefore, the total query time is $O(\log n+\beta n)$.
\end{proof}

\section{Closest Pair Query Processing}\label{sec:cp}
We proceed to cover closest-pair query processing based on PM-LSH. First, we propose a branch and bound algorithm that processes the nodes in the PM-tree in best-first manner. Due to the low efficiency of the branch and bound algorithm, we develop a radius filtering method to improve the query efficiency while sacrificing only slightly the accuracy of the candidate pairs found in the projected space.

\subsection{Branch and Bound Algorithm}
A straightforward method is to employ a branch and bound search strategy on the PM-tree. First, we aim to find $T$ point pairs in the PM-tree with the smallest distances in the projected space. Next, we verify their distances in the original space. Finally, we report $k$ closest pairs as the result. 

For any two nodes $e_1$ and $e_2$, we denote the minimum distance of any point pair $(o_1,o_2) \in e_1 \times e_2$ by $\textsf{Mindist}(e_1,e_2)$, which is computed as follows.
\begin{equation}
\begin{split}
&\textsf{Mindist}(e_1,e_2)=\\ 
&\max
\begin{cases}
\max_{i} LB(p_i),\\
\|\mathsf{e_1.RO},\mathsf{e_2.RO}\|-\mathsf{e_1.r}-\mathsf{e_2.r}\\
\end{cases}
\end{split}
\end{equation}

\begin{figure}[t]
	\centering
	\includegraphics[width=.4\textwidth]{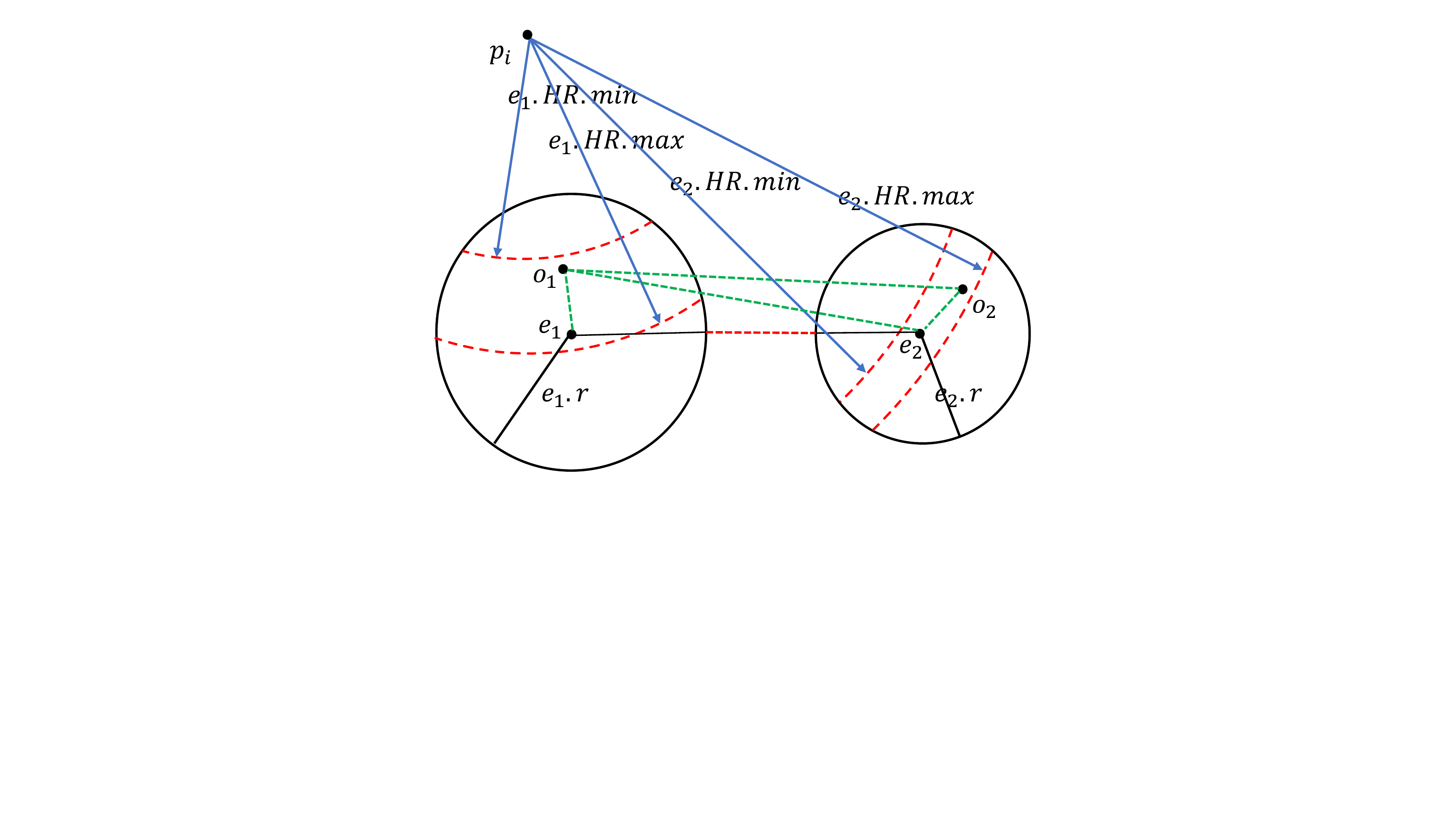}
	\caption{Illustration of Computing $\textsf{Mindist}$} \label{fig:Mindist}
\end{figure}

For the first term, we define a pivot-based lower bound $LB(p_i)$ of the minimum distance between $e_1$ and $e_2$ w.r.t. $p_i$, where $p_i$ is the $i$-th global pivot.
In Fig. \ref{fig:Mindist}, we have two points $o_1 \in e_1$ and $o_2 \in e_2$.
According to the property of the PM-tree, we know that $\|o_1,p_i\|$ is in the range $I_1$:
$$
I_1 = [\mathsf{e_1.HR}[i].min,\mathsf{e_1.HR}[i].max]
$$
Likewise, $\|o_2,p_i\|$ is in the range $I_2$:
$$
I_2 = [\mathsf{e_2.HR}[i].min,\mathsf{e_2.HR}[i].max]
$$
We compute $LB(p_i)$ based on the triangle inequality. 
Since $\|o_1,o_2\| \geq |\|o_1,p_i\| - \|o_2,p_i\||$, if $I_1$ overlaps $I_2$, we have $LB(p_i)=0$. Otherwise, $LB(p_i)$ is the distance between $I_1$ and $I_2$.
In Fig. \ref{fig:Mindist}, we have $LB(p_i) = \mathsf{e_2.HR}.min-\mathsf{e_1.HR}.max$.

For the second term, we estimate the minimum distance between $e_1$ and $e_2$ using their centers. We compute $\|o_1,o_2\|$ with $\mathsf{e_2.RO}$ as follows.
$$
\|o_1,o_2\|\geq \|o_1,\mathsf{e_2.RO}\|-\|\mathsf{e_2.RO},o_2\|
$$
We continue to compute $\|o_1,\mathsf{e_2.RO}\|$ with $\mathsf{e_1.RO}$ as follows.
$$
\|o_1,\mathsf{e_2.RO}\|\geq \|\mathsf{e_1.RO},\mathsf{e_2.RO}\|-\|\mathsf{e_1.RO},o_1\|
$$
Combined with the fact that $\|\mathsf{e_1.RO},o_1\|\leq \mathsf{e_1.r}$ and $\|\mathsf{e_2.RO},o_2\|\leq \mathsf{e_2.r}$, we obtain the second term.

Let $d_T$ be the current $T$-th smallest distance in the projected space. 
We access the node pairs in best-first manner according to the ascending $\textsf{Mindist}$ order. When the $\textsf{Mindist}$ of the next node pair to process exceeds $d_T$, the search terminates, and $T$ point pairs are returned for verification.

The details of Algorithm \ref{alg:branch_and_bound} are explained as follows.
\begin{enumerate}
\item We initialize a point pair candidate set $C$ of size $|C|=T$. We apply a self-join on each leaf node in the PM-tree and update $C$ and $d_T$ accordingly.
\item We maintain a priority queue $PQ$ to store node pairs in ascending $\textsf{Mindist}$ order. We initialize $PQ$ by inserting $(e_r,e_r)$, where $e_r$ is the root of the PM-tree.
\item We pop the top element $ ( e_{1},e_{2}) $ from $PQ$. If we have $ \textsf{Mindist}(e_1,e_2) > d_T $, the procedure stops; otherwise, we continue to examine $ (e_{1},e_{2}) $. The PM-tree is a balanced tree, and we only consider node pairs at the same level. Therefore, if $e_1$ and $e_2$ are leaf nodes, we compute the distance of each point pair in $e_1\times e_2$ and update $C$ and $d_T$ accordingly. If $e_1$ and $e_2$ are non-leaf nodes, for each child node $e_{1}'$ of $e_1$ and each child node $e_{2}'$ of $e_2$, we insert $( e_{1}',e_{2}')$ into $PQ$. This process terminates when $PQ$ is empty if it did not terminate earlier.
\item We verify the original distance of each point pair in $C$ and return top-$k$ point pairs.
\end{enumerate}

\begin{algorithm}[t]
	\caption{Branch and Bound Algorithm}\label{alg:branch_and_bound}
	\KwIn{A dataset $\mathcal{D}$, a PM-tree $\mathcal{T}$ indexing the projected data and parameters $T$, $n$, $k$}
	\KwOut{$k$ point pairs}
	Apply a self-join on each leaf node in $\mathcal{T}$ and store $k$ found pairs with the smallest distance in the projected space;\\
	$d_T$ $\gets$ maximum distance of pairs in $C$;\\ 
	Initialize a priority queue $PQ$ to store node pairs in ascending \textsf{Mindist} order;\\
	$PQ\gets (\mathcal{T}.root,\mathcal{T}.root)$;\\
		\While{$PQ$ is not empty}{
			$(e_1,e_2)\gets PQ.Pop$;\\
			\If{\textsf{Mindist}$ (e_1,e_2)>d_T $}{\textbf{Break};\\}
			\ForEach{child node $e_1'$ of $e_1$}
			{
				\If{$e_2'$ is a leaf node}{
					\ForEach{point pair $( o_{1}',o_{2}')$ in $e_1'\times e_2'$}
					{
						Compute $\|o_{1}',o_{2}'\|$ and update $C$ and $d_T$ accordingly;\\
					}
				}
				\Else{
					\ForEach{child node $e_2'$ of $e_2$}
					{
						Insert $(e_1',e_2')$ into $PQ$;\\
					}
				}	
			}	
		}
	
	Verify the original distance of each point pair in $C$;\\

	\textbf{Return} Top-$k$ results from the verified pairs;\\
\end{algorithm}

\begin{myexample}
	In Fig. \ref{fig:pm}, for a $(2,2)$-ACP query, we set $T=3$. First, we apply a self-join to all leaf nodes $e_3$, $e_4$, $e_5$, and $e_6$, obtaining the top-$3$ result $(o_7,o_{15})$, $(o_2,o_{14})$, and $(o_{6},o_{13})$ with $d_T=1.70$. 
	Then, we consider pairs of points in different leaf nodes. We initialize $PQ$ with $(e_r,e_r)$. As $\textsf{Mindist}(e_r,e_r)=0<d_T$, we continue to insert $(e_1,e_1)$, $(e_2,e_2)$, and $(e_1,e_2)$ into $PQ$. Next, $(e_1,e_1)$ and $(e_2,e_2)$ are examined. For $e_1$'s child nodes $e_3$ and $e_4$, since $(e_3,e_3)$ and $(e_4,e_4)$ have been examined, we only need to insert $(e_3,e_4)$ into $PQ$. After employing a similar operation for $e_2$, the node pairs in $PQ$ are $\langle (e_1,e_2),$ $(e_5,e_6),(e_3,e_4) \rangle$. 
	This process proceeds until we examine $(e_4,e_6)$, since \textsf{Mindist}$(e_4,e_6)=2.91>d_T$. We return the top-3 pairs $(o_7,o_{15})$, $(o_2,o_{14})$, and $(o_{6},o_{13})$ in the projected space. We verify their distances in the original space and return $(o_7,o_{15})$ and $(o_{6},o_{13})$ as the result.
\end{myexample}

\subsection{Limitations of the Branch and Bound Algorithm}
In the branch and bound algorithm, the search procedure terminates when $\textsf{Mindist}>d_T$, where $\textsf{Mindist}$ is a lower bound distance on unexamined pairs. However, this bound is often so loose that the algorithm efficiency suffers. Specifically, due to the property of the PM-tree, the ranges covered by two nodes at the same level overlap with high probability. No matter how small the overlap is, $\textsf{Mindist}(e_1,e_2)=0$.

To understand this issue better, we conduct an experiment on dataset \textit{Audio} to count the number of node pairs with $\textsf{Mindist}=0$. We employ the branch and bound algorithm to search the PM-tree, and we count the number of node pairs with $\textsf{Mindist}=0$ among all verified node pairs. We find that more than 70$\%$ of the node pairs have $\textsf{Mindist}=0$, which indicates that most node pairs overlap.

This phenomenon may be explained by the fact that PM-trees are built so that structured clusters are achieved for the subtrees of each node. The differences between nodes are not considered during construction, due to the high computational cost. Therefore, when points are located in a dense region, the tree nodes constructed for this region are likely to overlap substantially due to their limited node capacity.

Consequently, we have to examine about 90\% of all pairs in the branch and bound algorithm when using a PM-tree with $m=15$, which makes the algorithm degenerate to nearly a brute-force nested loop algorithm. We observe that we can lower $m$ to reduce the cost of finding exact closest pairs in the projected space. However, a small $m$ may lead to an inaccurate confidence interval when estimating the correlation between original and projected distances. As a result, we have to verify more candidate pairs to achieve a high recall.

\subsection{Improvement with Radius Filtering}
To fewer pairs, we provide a radius filtering method. The idea is to compute an upper bound distance of the $k$-th best point pair in the original space. We then estimate a candidate distance in the projected space based on the upper bound and use this distance to prune unnecessary node pairs. 

Specifically, we still apply a self-join on each individual leaf node in the PM-tree. Let $ub$ denote the upper bound distance in the original space. We verify the original distances of all self-join pairs and initialize $ub$ to be the current $k$-th smallest distance. According to Lemma \ref{lemma:e1e2}, if a point pair exists whose original distance is smaller than $ub$, its projected distance is smaller than $t \cdot ub$ with a high probability. Therefore, we aim to find point pairs in the PM-tree whose projected distance is within $t \cdot ub$. As we have already examined all point pairs in leaf nodes via self-joins, we only need to check pairs of points from different leaf nodes.

Let $(o_1',o_2')$ be the point pair of $(o_1,o_2)$ in the projected space. We observe that there is a strong relationship between the projected distance $\|o_1',o_2'\|$ and the radius of their lowest common ancestor in the PM-tree. We define the concept of lowest common ancestor as follows.
\begin{definition}[Lowest Common Ancestor]\label{def:commonancestor}
The lowest common ancestor (LCA) of two points $o_1'$ and $o_2'$ is a node $e$ in the PM-tree such that:
\begin{itemize}
\item Points $o_1'$ and $o_2'$ are stored in the subtree of $e$.
\item No child node $e'$ of $e$ exists such that $o_1'$ and $o_2'$ are also stored in the subtree of $e'$.
\end{itemize}
\end{definition}
Let $R=\textsf{e.r}$ denote the radius of the LCA node $e$ of $o_1'$ and $o_2'$. We assume that $\gamma \cdot \|o_1',o_2'\| \leq R$ holds with high probability, where the setting of parameter $\gamma$ is explained later. Therefore, in order to find point pairs with projected distance smaller than $t \cdot ub$, we only have to examine the points of nodes in the PM-tree whose radius is smaller than $\gamma \cdot t \cdot ub$.

We explain the details of Algorithm \ref{alg:rsj} as follows.

\begin{enumerate}
	\item We initialize a point pair candidate set $C$ of size $|C|=k$. We apply a self-join on each leaf node in the PM-tree, and we compute the original distances of all pairs found. We then update $C$ and $ub$ accordingly.
	
	\item Let $R = \gamma \cdot t \cdot ub$ be the radius used for node filtering in the PM-tree.
	
	\item We employ Algorithm $\textsf{FindLCA}()$ that traverses the PM-tree to find the nodes with radius smaller than $R$. A node $e$ returned by $\textsf{FindLCA}()$ may not be an LCA of the points it covers. 
	But we can find the LCA of any point pair it covers in the subtree of $e$, and the radius of the LCA is smaller than $R$. Therefore, it suffices to examine the point pairs covered by $e$.
	
	\item We consider the nodes returned by $\textsf{FindLCA}()$ in ascending order of their radii. The intuition is that a node with a small radius is likely to cover point pairs with small projected distances. 
	
	
	\item We examine the nodes in turn. For any two points $o_1'$ and $o_2'$ in the sub-tree of a node $e$, we compute $\|o_1',o_2'\|$, and if $\|o_1',o_2'\| < t \cdot ub$, we consider $(o_1,o_2)$ as a candidate pair. Then, we compare $\|o_1,o_2\|$ with $ub$ and update both $ub$ and $C$ if necessary. This process stops when we have $T$ candidate pairs.

	\item We return $C$ as the result.
\end{enumerate}

\begin{algorithm}[t]
	\caption{Radius Filtering Method}\label{alg:rsj}
	\KwIn{A dataset $\mathcal{D}$, a PM-tree $\mathcal{T}$ indexing the projected data and parameters $T$, $n$, $t$, $k$, $\gamma$}
	\KwOut{$k$ point pairs}
	
	Apply a self-join on each leaf node in $\mathcal{T}$ and verify all found point pairs;\\
	$count\gets $ The number of verified pairs;\\
	$ub$ $\gets$ The $k$-th smallest real distance in found pairs;\\ 
	
	$R\gets \gamma\cdot t \cdot ub$;\\
	$C \gets \varnothing$;\\
	Initialize an array $A$ to store the nodes;\\
	
	$ \textsf{FindLCA}(\mathcal{T}.root,R,A) $;\\
	Sort the nodes in $A$ in ascending order of their radii;\\
	\ForEach{node $e$ in $A$}
	{
		\ForEach{point pair $( o_{1},o_{2})$ in $e$'s subtree}
		{
			\If{$\|o_1',o_2'\|<t \cdot ub$}{
				Verify $( o_{1},o_{2})$ and update $ub$;\\
				$count$++;\\
			}
			\If{$count>T$}{\textbf{Break};\\}
		}
		\If{$count>T$}{\textbf{Break};\\}
	}
	\textbf{Return} All pairs in $C$;\\
\end{algorithm}

\begin{algorithm}[t]
	\caption{\textsf{FindLCA}$ (e, R, A) $}\label{alg:findNode}
	\KwIn{A PM-tree node $e$, a radius $R$, and an array $A$}
	\KwOut{$A$}
	
	\If{$e$ is an inner node}{
		\If{$\mathsf{e.r}<R$}{Insert $e$ into $A$;\\}
		\Else{
			\ForEach{child node $e_i$ of $e$}
			{
				$ \textsf{FindLCA}(e_i,R,A) $;\\
			}
		}
	}
\end{algorithm}

\begin{myexample}
	In the example in Fig. \ref{fig:pm}, the PM-tree has 4 leaf nodes $e_3$, $e_4$, $e_5$, and $e_6$. To compute a $(2,2)$-ACP query, we first apply a self-join to all leaf nodes and obtain the preliminary top-$2$ pairs $(o_4,o_8)$ and $(o_{12},o_{14})$, both with distance $1$. We set $ub=1$. Setting $t=3$ and $\gamma=3$, we get $t\cdot ub=3$ and $R=9$. We find all inner nodes whose ranges are within $9$ and obtain $e_2$. The unverified pairs in the subtree of $e_2$ come from $e_5\times e_6$. As $\|o_4',o_2'\|=3.2>3$, we skip it and process the remaining pairs. Finally, we obtain $R = \langle (o_4,o_8),$ $(o_{12},o_{14}) \rangle$.
\end{myexample}

\textbf{Determining the Setting of $\gamma$.}
For any two points $o_1'$ and $o_2'$ in the projected space, we observe that $\|o_1',o_2'\|$ and the radius of their LCA have a strong correlation. Let $\gamma=\frac{R}{\|o_1',o_2'\|}$ be the ratio of $R$ over $\|o_1',o_2'\|$. To ensure the quality of the nodes returned by the radius filtering, we need to find an appropriate setting for $\gamma$. To do so, we study the probability density functions of $\gamma$ on real datasets.

Let us take dataset \textit{Audio} (Details are provided in Sec. \ref{sec:experiments}) as an example. We use $m=15$ hash functions. First, we randomly select $10$K data points. We then index these points in the projected space using two PM-trees with node capacity $M=2$ and $M=16$, respectively.
We obtain some $50$ million point pairs from 10K points. For each pair, we compute the value of $\gamma$. Fig. \ref{fig:rca} shows the probability density functions $f_{\gamma}(x)$ for $M=2$ and $M=16$. It is easy to see that the two functions have similar trends. Both peak quickly and then decline quickly. An appropriate value of $\gamma$ is very likely to be within the neighborhood of the peak, which indicates that $\gamma$ varies slightly for different pairs. With $\Pr(\gamma)$ being the success probability, we choose $\gamma$ such that $\Pr(\gamma) = \int_{0}^{\gamma} f_{\gamma}(x)dx = 85\%$ for all datasets. Note that we can enlarge the value of $\Pr(\gamma)$ to examine more nodes. But this represents a tradeoff between accuracy and efficiency, and $\Pr(\gamma)=85\%$ already provides good performance.
We analyze the cost of computing $\gamma$ experimentally in Section \ref{sec:experiments}. Specifically, the cost is the time it takes to compute the distances of $50$ million point pairs, which is acceptable when compared with the total cost.

\begin{figure}
	\centering
	\includegraphics[width=.3\textwidth]{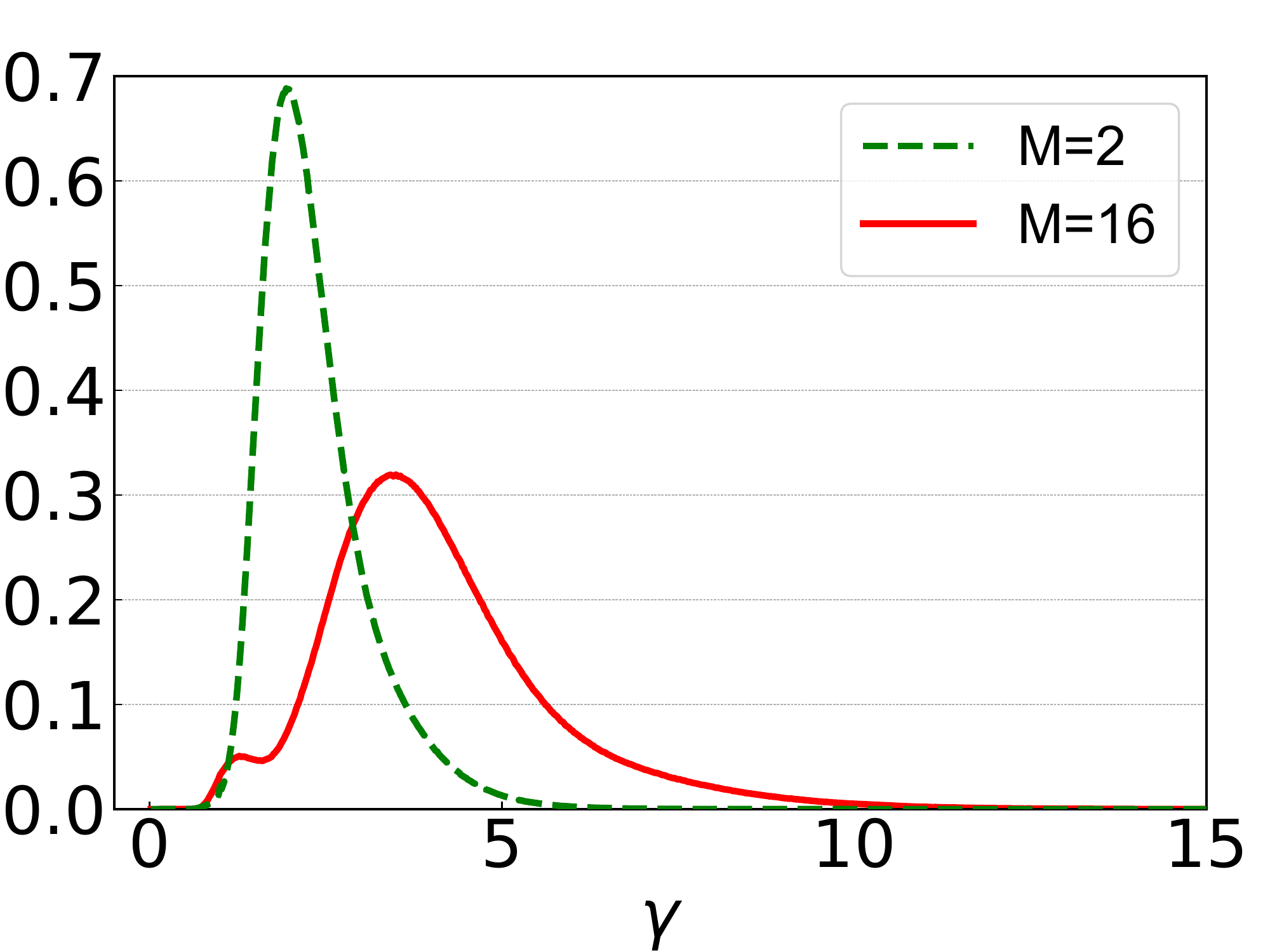}
	\caption{The Probability Density Function of $\gamma$}
	\label{fig:rca}
\end{figure}

\textbf{Promote Methods for the PM-tree.}
The PM-tree is built bottom-up by inserting the data points one by one. When a node $e$ overflows after inserting $M+1$ entries, we allocate a new node $e'$ at the same level and distribute the $ M + 1 $ entries among the two nodes.
One study \cite{DBLP:conf/sebd/CiacciaPRZ97} contributes the concept of a \textsf{Promote} method that selects two points as the centers of two nodes $e$ and $e'$. It is easy to see that different centers may lead to different partitioning results, which affects the algorithm performance. We consider two \textsf{Promote} methods as follows.

\begin{itemize}
\item
\textsf{m$\_$RAD} selects two points from all possible combinations as the centers such that the sum of the two covering radii is the minimum after partitioning. 
This method incurs many distance computations but also yields high-quality partitioning. 
\item
\textsf{RANDOM} selects two points as node centers at random.
\end{itemize}

It is obvious that \textsf{m$\_$RAD} provides no worse partitioning than does \textsf{RANDOM}, since \textsf{m$\_$RAD} aims to minimize the sum of the two covering radii, which represents a locally optimal partitioning of the $M+1$ entries. Consequently, the two nodes are covered by a parent node with a small radius. 
In this case, the radius filtering strategy enables to obtain $T$ candidate pairs with higher quality.

\noindent
\textbf{Algorithm Analysis of Radius Filtering}.
In the radius filtering method, as we have $n(n-1)/2$ pairs, we set $T=\beta n(n-1)/2+k$, which is similar to the setting for the NN query.

\begin{theorem}
	PM-LSH answers an ACP query with space cost $O(n)$ and time cost $O(\beta n^2)$, where $\beta$ is much smaller than 1.
\end{theorem}

\begin{proof}
	The space consumption is due mainly to the PM-tree with $n$ points. Each point consumes $m+O(1)$ space, so the overall space consumption is $O(n)$ as $m=O(1)$. The query time cost stems from two operations: 1) finding candidate pairs in the PM-tree, and 2) verifying the real distances of candidate pairs. Both operations have cost $O(T)$ when $d$ is considered as a constant. According to the setting of $T$, the total query time is $O(\beta n^2)$.
\end{proof}
\section{Experiments} \label{sec:experiments}
We report on extensive experiments with real datasets that offer insight into the performance of PM-LSH for both NN and CP queries.

\subsection{Experimental Settings}
All the algorithms are implemented in C++, and compilation is done with the O3 optimization. All experiments are run on a Linux machine with an Intel 3.4GHz CPU and 32GB memory.

\begin{table}
	\caption{Datasets}
	\vspace*{0.1in}
	\label{tb:datasets}
	\centering
	\begin{tabular}{ccccccc}
		\hline
		\textbf{Dataset} &\textbf{$n$} ($\times 10^3$) & \textbf{$d$} & \textbf{HV} & \textbf{RC} & \textbf{LID}\\
		\hline
		Audio & 54 & 192 &0.9273 &2.97 &5.6\\
		Deep & 1,000 & 256 &0.9393 &1.96 &12.1\\
		NUS & 269 & 500 &0.9995 &1.67 &24.5\\
		MNIST & 60 & 784 &0.9531 &2.38 &6.5\\
		GIST & 983 & 960 &0.9670 &1.94 &18.9\\
		Cifar & 50 & 1,024 &0.9457 &1.97 &9.0\\
		Trevi & 100 & 4,096 &0.9432 &2.95 &9.2\\
		\hline
	\end{tabular}
\end{table}

\begin{figure*}[t]
	\centering
	\subfigure[Time]{
		\begin{minipage}[c]{0.23\linewidth}
			\centering
			\includegraphics[width=1\textwidth]{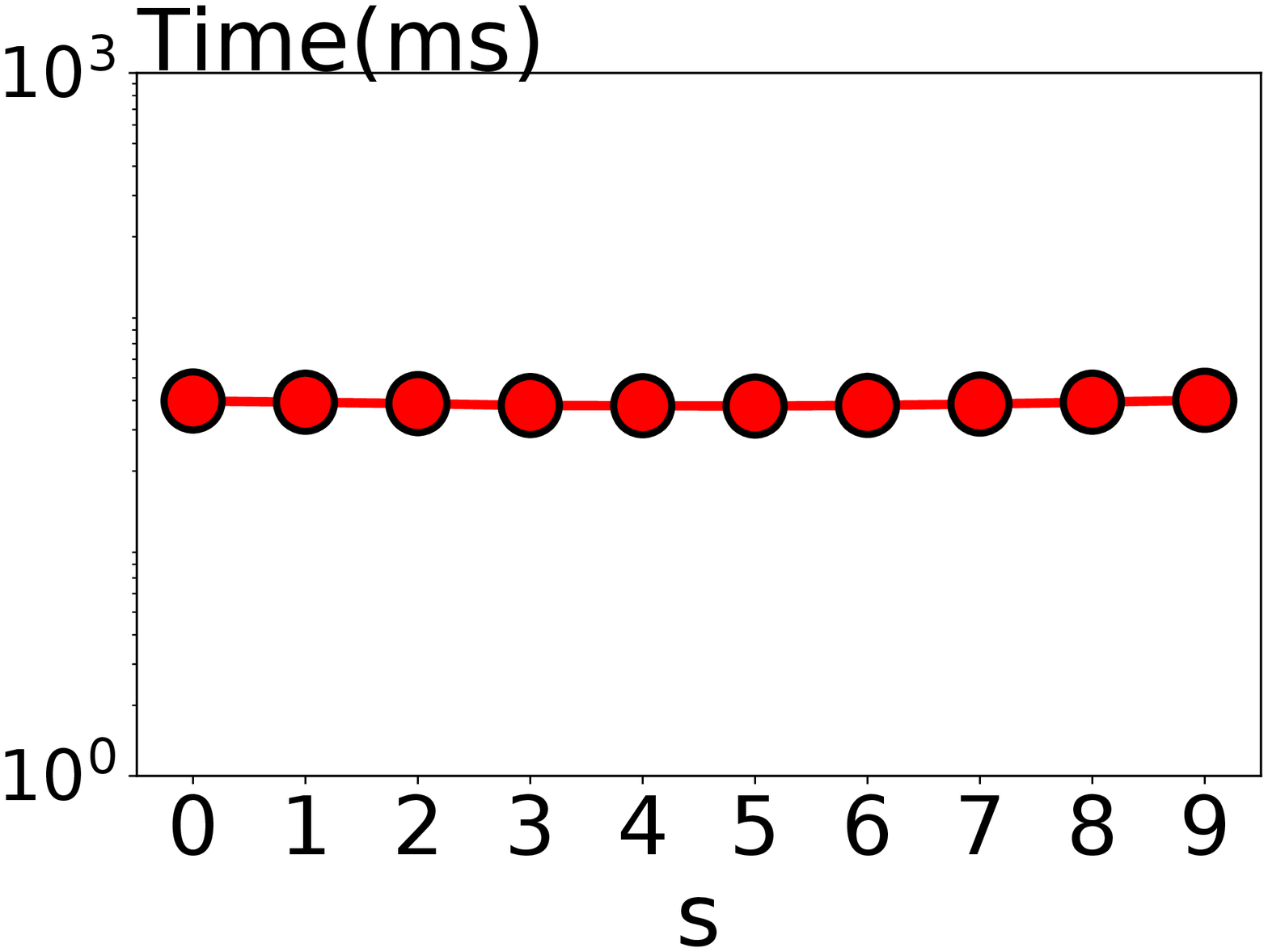}
		\end{minipage}
		\label{fig:s}
	}
	\subfigure[Time]{
		\begin{minipage}[c]{0.23\linewidth}
			\centering
			\includegraphics[width=1\textwidth]{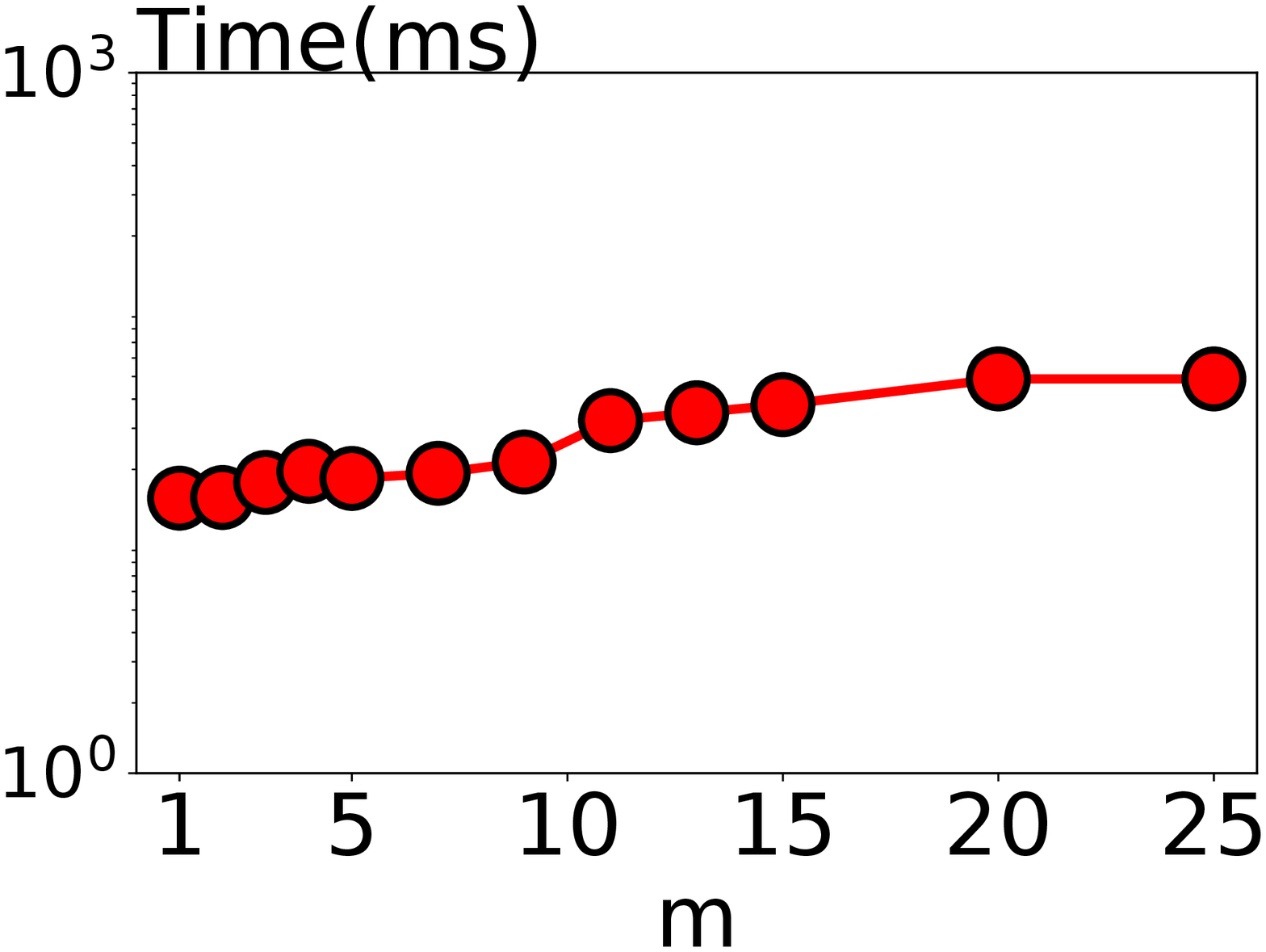}
		\end{minipage}
	}
	\subfigure[Recall]{
		\begin{minipage}[c]{0.23\linewidth}
			\centering
			\includegraphics[width=1\textwidth]{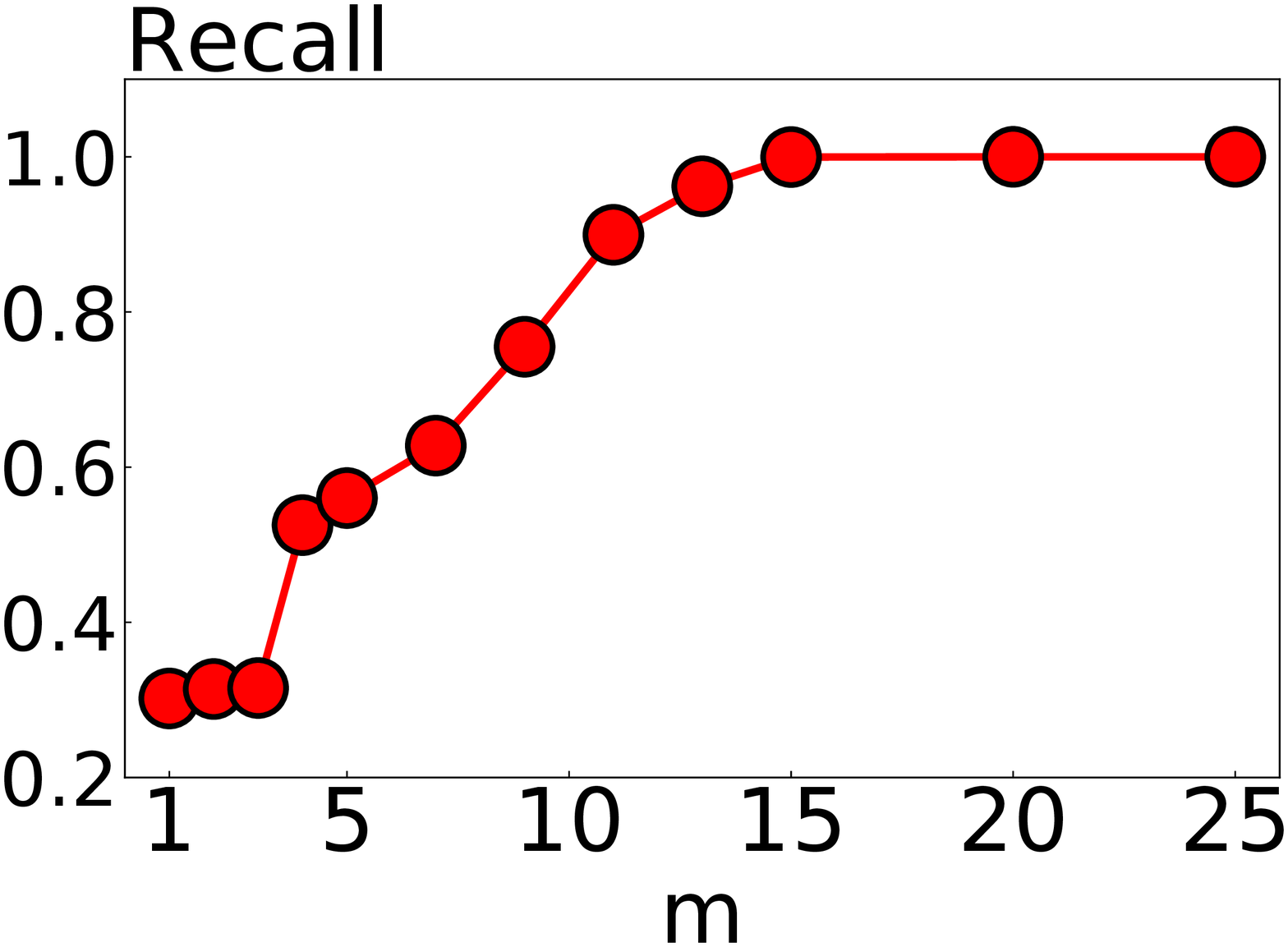}
		\end{minipage}
	}
	\subfigure[OverRatio]{
		\begin{minipage}[c]{0.23\linewidth}
			\centering
			\includegraphics[width=1\textwidth]{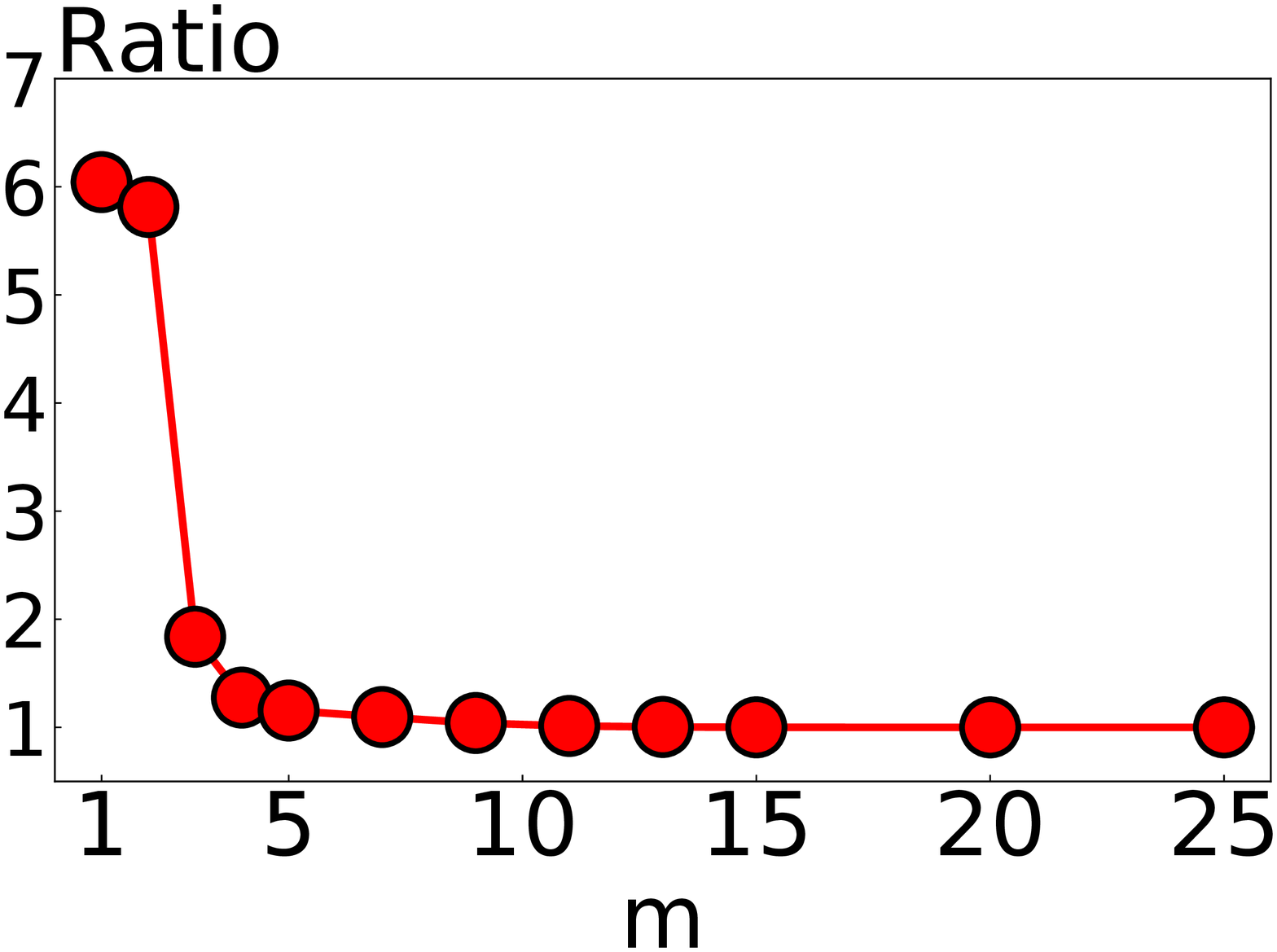}
		\end{minipage}
	}
	\caption{Performance of PM-LSH when Varying $s$ and $m$}
	\label{fig:m}
\end{figure*}

\textbf{Datasets}.
We use seven real datasets: \textit{Audio}, \textit{Deep}, \textit{NUS}, \textit{MNIST}, \textit{GIST}, \textit{Cifar}, and \textit{Trevi}, which are used widely in existing LSH studies \cite{DBLP:journals/tkde/LiZSWLZL20, DBLP:conf/sigmod/GanFFN12,DBLP:journals/pvldb/HuangFZFN15,DBLP:journals/pvldb/SunWQZL14,DBLP:conf/sigmod/LiYZXCLNC18}.
Table \ref{tb:datasets} reports key statistics of the datasets: \textit{Homogeneity of Viewpoints} (HV\cite{DBLP:conf/pods/CiacciaPZ98}), \textit{Relative Contrast} (RC \cite{DBLP:conf/icml/HeKC12}), and \textit{Local Intrinsic
Dimensionality} (LID \cite{DBLP:conf/kdd/AmsalegCFGHKN15}). 
RC computes the ratio of the mean distance over the NN distance for the data points. LID computes the local intrinsic dimensionality. 
A small RC value and a large LID value imply that it is challenging to compute NN results for the dataset.
HV evaluates the homogeneity of the distance distributions of the data points.
A higher HV means that the points are more likely to have similar distance distributions.

\textbf{Query Set}.
For NN queries, we randomly select 200 points from each dataset, and we repeat each experiment 20 times and report the average value.
We set the default value of $c$ to 1.5, and vary its value in $\{1.1, 1.2,\dots,2.0\}$.
We vary the value of $k$ in $\{1,10,20,\dots,100\}$ and set the default value to 50.
For CP queries, we repeat each experiment 20 times and report the average value.
We vary the value of $k$ in $\{1,10,10^2,\dots,10^4\}$ and set the default value to $10^3$.
The default value of $c$ is 4 in PM-LSH and the LSB-tree.

\textbf{Competing Algorithms.}
For NN queries, we compare PM-LSH with the following competitors:
\begin{enumerate}
\item \textbf{Multi-Probe} \cite{DBLP:conf/vldb/LvJWCL07}: A probing sequence (PS) based algorithm.
\item \textbf{QALSH} \cite{DBLP:journals/pvldb/HuangFZFN15}: A radius enlargement (RE) based algorithm.
\item \textbf{SRS} \cite{DBLP:journals/pvldb/SunWQZL14}: A metric indexing (MI) based algorithm.
\item \textbf{R-LSH}: In order to compare the PM-tree and the R-tree, we index the points in the projected space with an R-tree instead of a PM-tree to see how PM-LSH then performs. We call this method R-LSH.
\item \textbf{LScan}: We consider a linear scan algorithm called LScan that randomly selects a portion of points (default 70\%) and returns the top-$k$ points with the smallest distances to the query.
\end{enumerate}

For CP queries, we compare PM-LSH with the following competitors:
\begin{enumerate}
\item \textbf{LSB-tree} \cite{DBLP:journals/tods/TaoYSK10}: The LSB-tree supports both NN and CP queries.
\item \textbf{M$ k $CP} \cite{DBLP:journals/vldb/GaoCLYC15}: M$k$CP supports CP queries with the M-tree. We choose the variant called GMA that uses grouping and N-consider techniques that enables trade-offs between time and accuracy.
\item \textbf{ACP-P} \cite{DBLP:conf/pakdd/CaiRZ18}: The state-of-the-art solution for CP queries.
\item \textbf{NLJ}: Nested loop join (NLJ) is an exact algorithm that computes the distance between any two points with two nested loops and then returns the top-$ k $ CPs.
\end{enumerate}

\textbf{Parameter Settings.}
For NN queries, we choose $m=15$ hash functions for all the algorithms except QALSH and Multi-Probe.
In our method, we set the number of pivots $s=5$ and $\alpha_1=1/e$, so $\alpha_2=0.1405$ and $\beta=0.2809$ are obtained according to Eq. \ref{eq:paramsetting}, and $r_{min}$ is determined according to the description in the previous section.
For QALSH, the false-positive percentage $\beta =100/n$, and the error probability $\delta=1/e$.
For SRS, the threshold of its early-termination condition $p_{\tau}'= 0.8107$, and the maximum percentage of points accessed in the projected space is $T= 0.4010$ when $c=1.5$.

For CP queries, we choose $m=15$ hash functions for our algorithm. We set the number of pivots $s=5$, $\Pr(\gamma)=0.85$, and $\alpha_1=1/e$, so $\alpha_2=0.0024$ are obtained according to Eq. \ref{eq:paramsetting}, and thus $T=\alpha_2 n(n-1)+k$. For ACP-P, we set the hyper parameter $ h=5 $ and the range value is set to 5 according to the advice of its authors.
For M$ k $CP, we set the number of groupings to $ N=2 $. 
For the LSB-tree, the approximation ratio is set to $c=4$.

\textbf{Evaluation Metrics.}
We adopt three metrics to assess the performance of the algorithms: query time (ms for NN, s for CP), overall ratio, and recall, where the query time quantifies the algorithm efficiency and the overall ratio and recall capture the result quality.
For an NN query $q$, we denote the result of a $(c,k)$-ANN query by $R=\langle o_1,o_2,\cdots,o_k\rangle$. Let $R^*=\langle o_1^*,o_2^*,\cdots,o_k^* \rangle$ be the exact $k$NNs. The overall ratio and recall are computed as follows.
\begin{equation}
\mathit{OverallRatio} = \frac{1}{k} \sum_{i=1}^{k} \frac{\|q,o_i\|}{\|q,o_i^*\|}
\end{equation}
\begin{equation}
\mathit{Recall}=\frac{\lvert R \cap R^* \rvert}{\lvert R^* \rvert}
\end{equation}

For a CP query, we denote the result of a $(c,k)$-ACP query by $R=\langle (o_{1,1},o_{1,2}),(o_{2,1},o_{2,2}),\dots,(o_{k,1},o_{k,2}) \rangle$. Let $R^*=\langle (o_{1,1}^*,o_{1,2}^*),(o_{2,1}^*,o_{2,2}^*),\dots,(o_{k,1}^*,o_{k,2}^*) \rangle$ be the exact $k$CPs. The recall is the same as for the NN query, and the overall ratio is computed as follows.
\begin{equation}
\mathit{OverallRatio} = \frac{1}{k} \sum_{i=1}^{k} \frac{\|o_{i,1},o_{i,2}\|}{\|o_{i,1}^*,o_{i,2}^*\|}
\end{equation}

\renewcommand\arraystretch{1.2}
\begin{table*}[htbp]
	\centering
	\caption{Performance Overview of NN Queries}
	\label{tb:overview}
	\begin{tabular}{|c|c| m{1.5cm}<{\centering} |m{1.5cm}<{\centering}|m{1.5cm}<{\centering}|m{2cm}<{\centering}|m{1.5cm}<{\centering}|m{1.5cm}<{\centering}|}
		\cline{1-8}
		\multicolumn{2}{|c|}{} & \textbf{PM-LSH} & \textbf{SRS} & \textbf{QALSH} & \textbf{Multi-Probe} & \textbf{R-LSH} & \textbf{LScan} \\
		\cline{1-8}
		\multirow{3}*{\textbf{Audio}}&Query Time (ms)&\textbf{13.5}&15.3&22.5&15.3&14.2&19.6\\	
		\cline{2-8}
		\multicolumn{1}{|c|}{}&Overall Ratio&\textbf{1.0014}&1.0025&1.0043&1.0242&1.0019&1.0073\\
		\cline{2-8}
		\multicolumn{1}{|c|}{}&Recall&\textbf{0.9662}&0.9126&0.9003&0.8669&0.9633&0.6839\\
		\cline{1-8}
		\multirow{3}*{\textbf{MNIST}}&Query Time (ms)&\textbf{12.3}&18.4&24.7&19.1&16.2&60.3\\	
		\cline{2-8}
		\multicolumn{1}{|c|}{}&Overall Ratio&\textbf{1.0076}&1.0101&1.0085&1.0103&1.0095&1.0276\\
		\cline{2-8}
		\multicolumn{1}{|c|}{}&Recall&\textbf{0.8857}&0.8514&0.8655&0.8502&0.8705&0.7073\\
		\cline{1-8}
		\multirow{3}*{\textbf{NUS}}&Query Time (ms)&\textbf{125.7}&142.1&133.2&125.9&129.6&176.8\\	
		\cline{2-8}
		\multicolumn{1}{|c|}{}&Overall Ratio&\textbf{1.0009}&1.0015&1.0027&1.0025&1.0011&1.0053\\
		\cline{2-8}
		\multicolumn{1}{|c|}{}&Recall&\textbf{0.9257}&0.9247&0.8677&0.8782&0.9214&0.7057\\
		\cline{1-8}
		\multirow{3}*{\textbf{Trevi}}&Query Time (ms)&\textbf{37.2}&47.9&145.5&239.3&63.9&57.68\\	
		\cline{2-8}
		\multicolumn{1}{|c|}{}&Overall Ratio&\textbf{1.0004}&1.0015&1.0029&1.0057&1.0044&1.0084\\
		\cline{2-8}
		\multicolumn{1}{|c|}{}&Recall&\textbf{0.9961}&0.9342&0.8240&0.8534&0.9568&0.7103\\
		\cline{1-8}
		\multirow{3}*{\textbf{Cifar}}&Query Time (ms)&\textbf{11.6}&16.1&38.3&26.8&35.6&58.2\\	
		\cline{2-8}
		\multicolumn{1}{|c|}{}&Overall Ratio&\textbf{1.0009}&1.0025&1.0057&1.0038&1.0056&1.0125\\
		\cline{2-8}
		\multicolumn{1}{|c|}{}&Recall&\textbf{0.9746}&0.9624&0.7917&0.8011&0.9610&0.7081\\
		\cline{1-8}
		\multirow{3}*{\textbf{GIST}}&Query Time (ms)&\textbf{398.7}&452.5&627.7&782.9&425.3&1528.3\\	
		\cline{2-8}
		\multicolumn{1}{|c|}{}&Overall Ratio&\textbf{1.0047}&1.0049&1.0037&1.0053&1.0059&1.0076\\
		\cline{2-8}
		\multicolumn{1}{|c|}{}&Recall&\textbf{0.8436}&0.8145&0.8534&0.8122&0.8098&0.7023\\
		\cline{1-8}
		\multirow{3}*{\textbf{Deep}}&Query Time (ms)&\textbf{227.8}&252.9&458.2&401.4&457.5&507.5\\	
		\cline{2-8}
		\multicolumn{1}{|c|}{}&Overall Ratio&\textbf{1.0037}&1.0077&1.0124&1.0112&1.0152&1.0145\\
		\cline{2-8}
		\multicolumn{1}{|c|}{}&Recall&\textbf{0.8816}&0.8894&0.646&0.8118&0.8801&0.6938\\
		\hline
	\end{tabular}
\end{table*}

\subsection{Evaluation of NN Query Processing}
To evaluate the performance of PM-LSH for NN query processing, we first determine parameter settings. Then, we compare the performance of all algorithms with default parameter settings on all datasets. Finally, we compare the algorithms by studying the changes of the overall ratio and recall under fixed query times.

\textbf{Parameter Study on PM-LSH for NN Query.}
We consider two parameters that may affect the performance of PM-LSH, i.e., the number of pivots $s$ and the number of hash functions $m$. 
Here, we only show results from the \textit{Trevi} dataset.
It is easy to see that $s$ only affects the query time. The overall ratio and recall do not change when we vary $s$.
As we can see from the Fig. \ref{fig:s}, when $s$ changes, the query time remains steady, which indicates that PM-LSH is largely unaffected by different settings for $s$. When using a larger number of pivots, we have a higher chance to prune subtrees in the PM-tree. However, the cost of checking the pruning condition also increases. In conclusion, we set $s=5$.

As shown in Fig. \ref{fig:m}, when the value of $m$ increases, we obtain a higher overall ratio and recall, but the query time also increases. The higher quality occurs because a larger $m$ leads to more accurate distance estimation. However, the average cost to retrieve a point from the PM-tree also increases. Taking both efficiency and accuracy into consideration, we set $m=15$.

When comparing PM-LSH with R-LSH, we observe in all the experiments that PM-LSH outperforms R-LSH on all metrics, which confirms the expected superiority of the PM-tree over the R-tree.

\textbf{Performance Overview of NN Query}.
To compare all the algorithms with default parameter settings, we report the query time (ms), overall ratio, and recall on all datasets in Table \ref{tb:overview}. 
PM-LSH is more efficient than the competitors on all datasets, and its overall ratio and recall are also better than those of its competitors. Moreover, we find that either query time, overall ratio, or recall depend only slightly on the dataset dimensionality. For instance, \textit{Audio}, \textit{MNIST}, and \textit{Cifar} have nearly the same cardinality, but different dimensionality, i.e., $192$, $784$, and $1024$. However, the query times of PM-LSH on them are different and it is not only affected by data dimensionality. So we explain this by the query time being affected by the data distribution.
In Table \ref{tb:datasets}, we can see that dataset \textit{GIST} has large LID value and small RC value, so it is considered as challenging dataset. As shown in Table \ref{tb:overview}, it has larger query times than the other datasets.

\begin{figure*}[htbp]
	\centering
	\subfigure{
		\centering
		\includegraphics[width=0.8\textwidth]{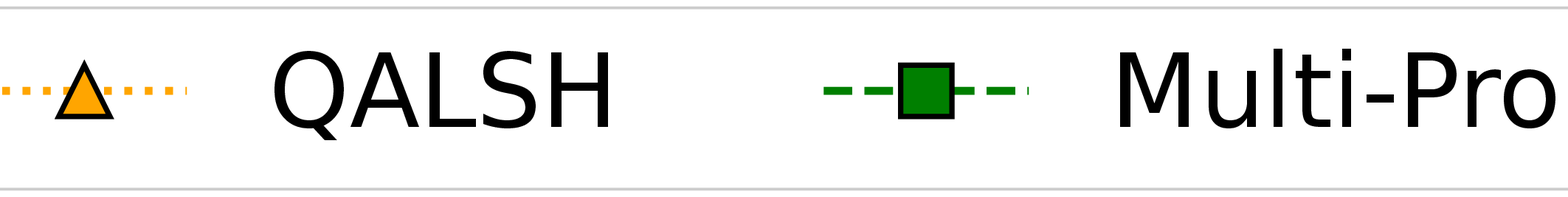}
	}
	\subfigure[Time on Cifar]{
		\begin{minipage}[c]{0.3\linewidth}
			\centering
			\includegraphics[width=1\textwidth]{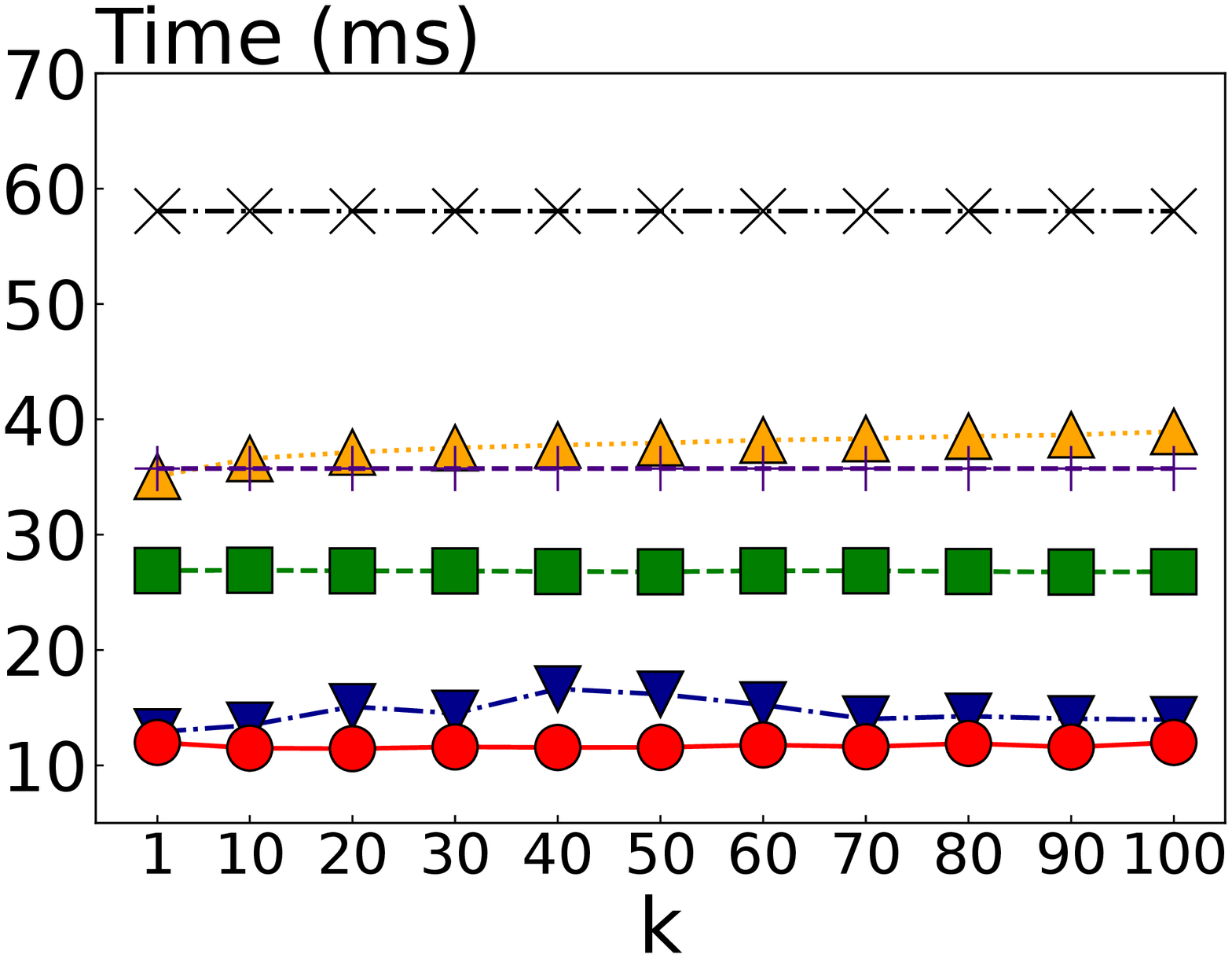}
		\end{minipage}
	}
	\subfigure[Recall on Cifar]{
		\begin{minipage}[c]{0.3\linewidth}
			\centering
			\includegraphics[width=1\textwidth]{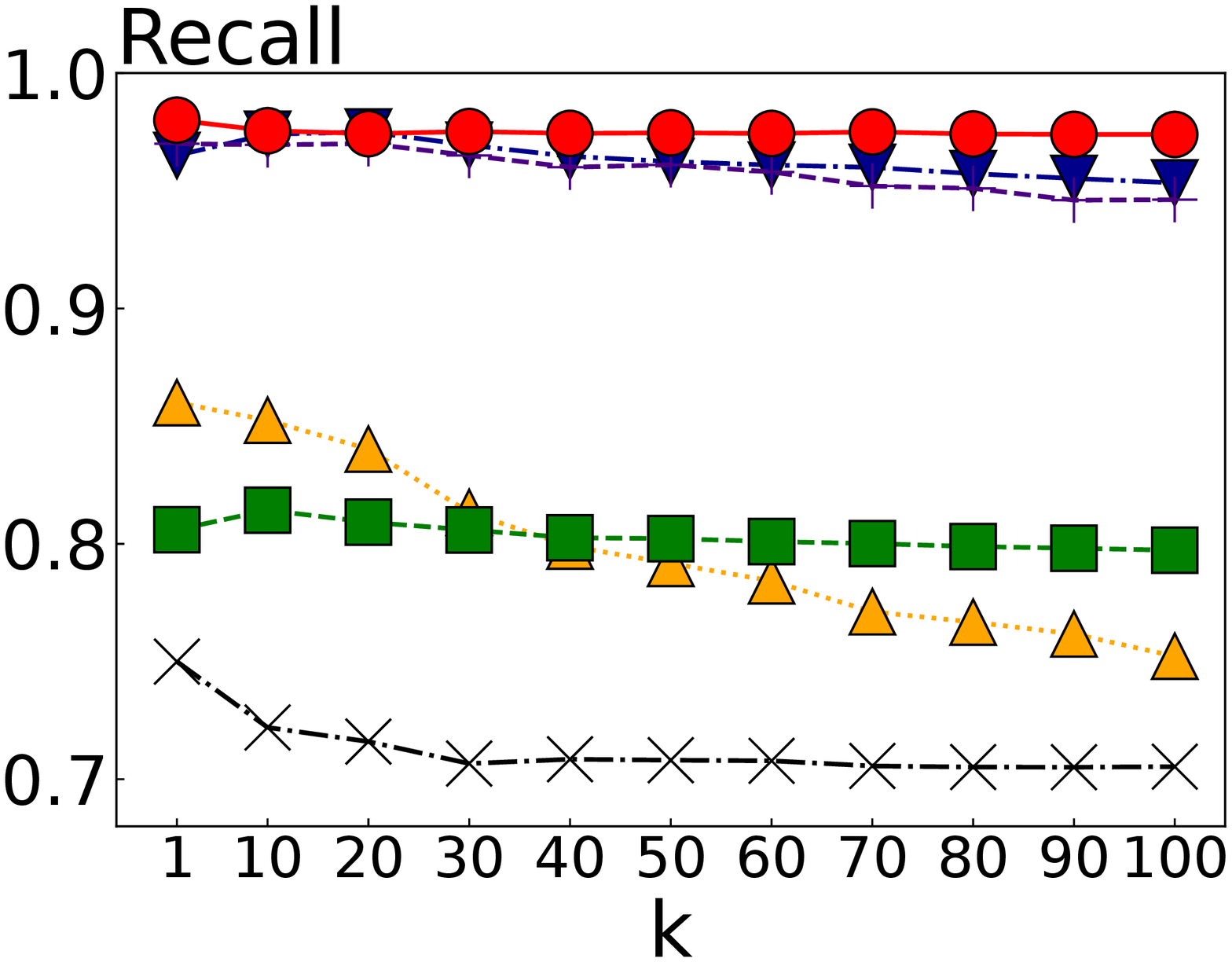}
		\end{minipage}
	}
	\subfigure[OverRatio on Cifar]{
		\begin{minipage}[c]{0.3\linewidth}
			\centering
			\includegraphics[width=1\textwidth]{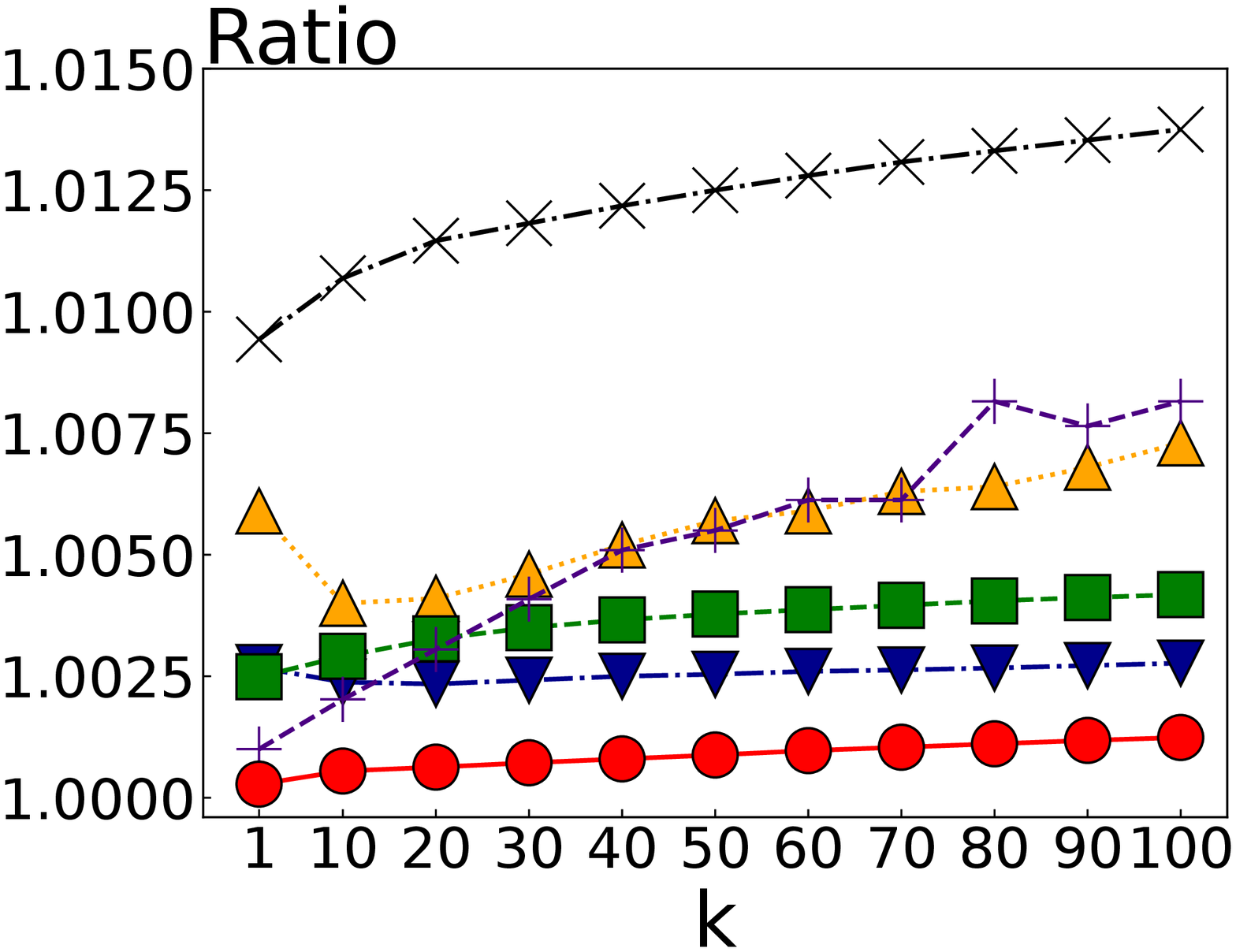}
		\end{minipage}
	}
	\caption{Performance on Cifar when Varying $k$ of NN Queries}
	\label{result_Cifar}
\end{figure*}

\begin{figure*}[htbp]
	\centering
	\subfigure[Time on Deep]{
		\begin{minipage}[c]{0.3\linewidth}
			\centering
			\includegraphics[width=1\textwidth]{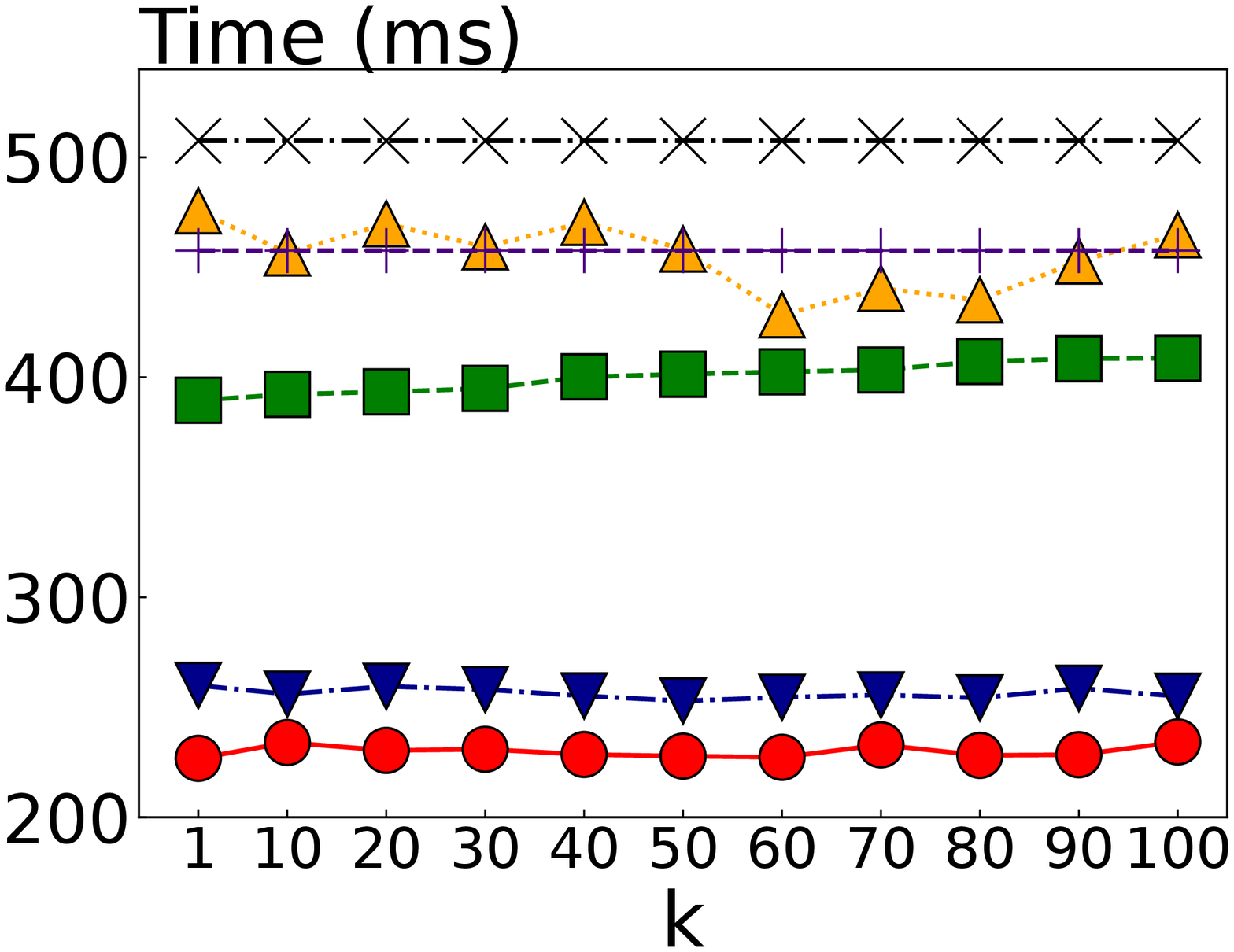}
		\end{minipage}
	}
	\subfigure[Recall on Deep]{
		\begin{minipage}[c]{0.3\linewidth}
			\centering
			\includegraphics[width=1\textwidth]{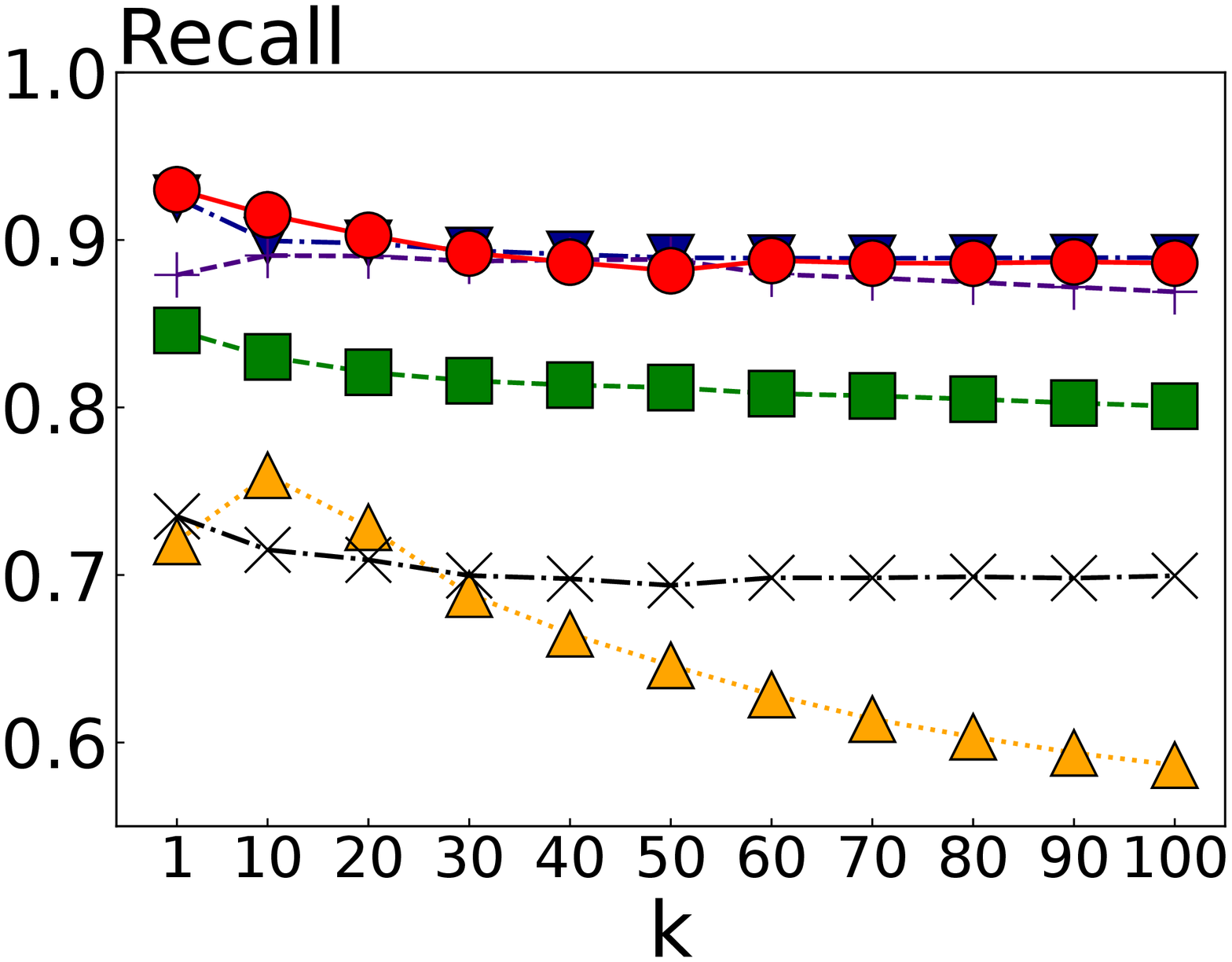}
		\end{minipage}
	}
	\subfigure[OverRatio on Deep]{
		\begin{minipage}[c]{0.3\linewidth}
			\centering
			\includegraphics[width=1\textwidth]{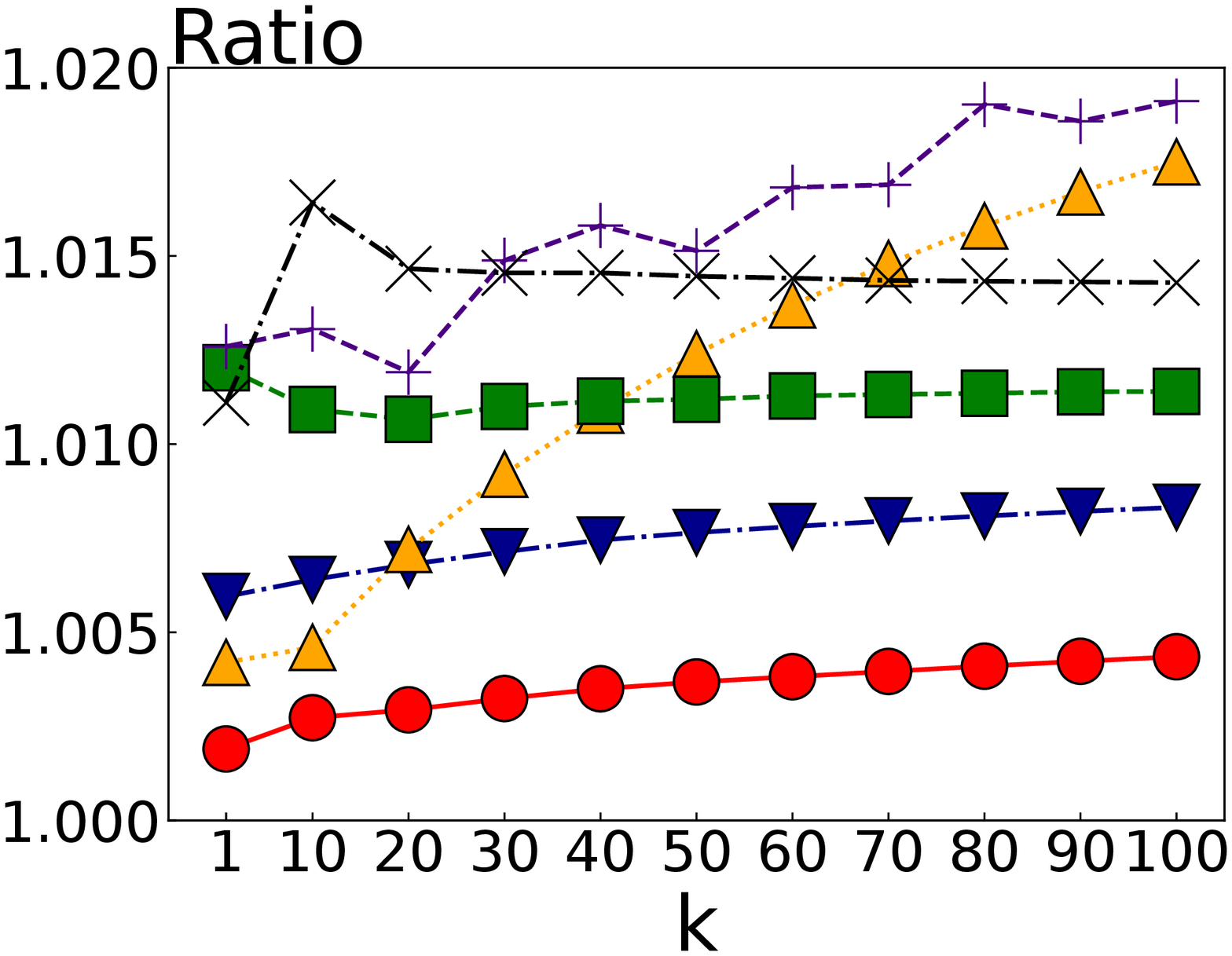}
		\end{minipage}
	}
	\caption{Performance on Deep when Varying $k$ of NN Queries}
	\label{result_Deep}
\end{figure*}

\begin{figure*}[htbp]
	\centering
	\subfigure[Time on Trevi]{
		\begin{minipage}[c]{0.3\linewidth}
			\centering
			\includegraphics[width=1\textwidth]{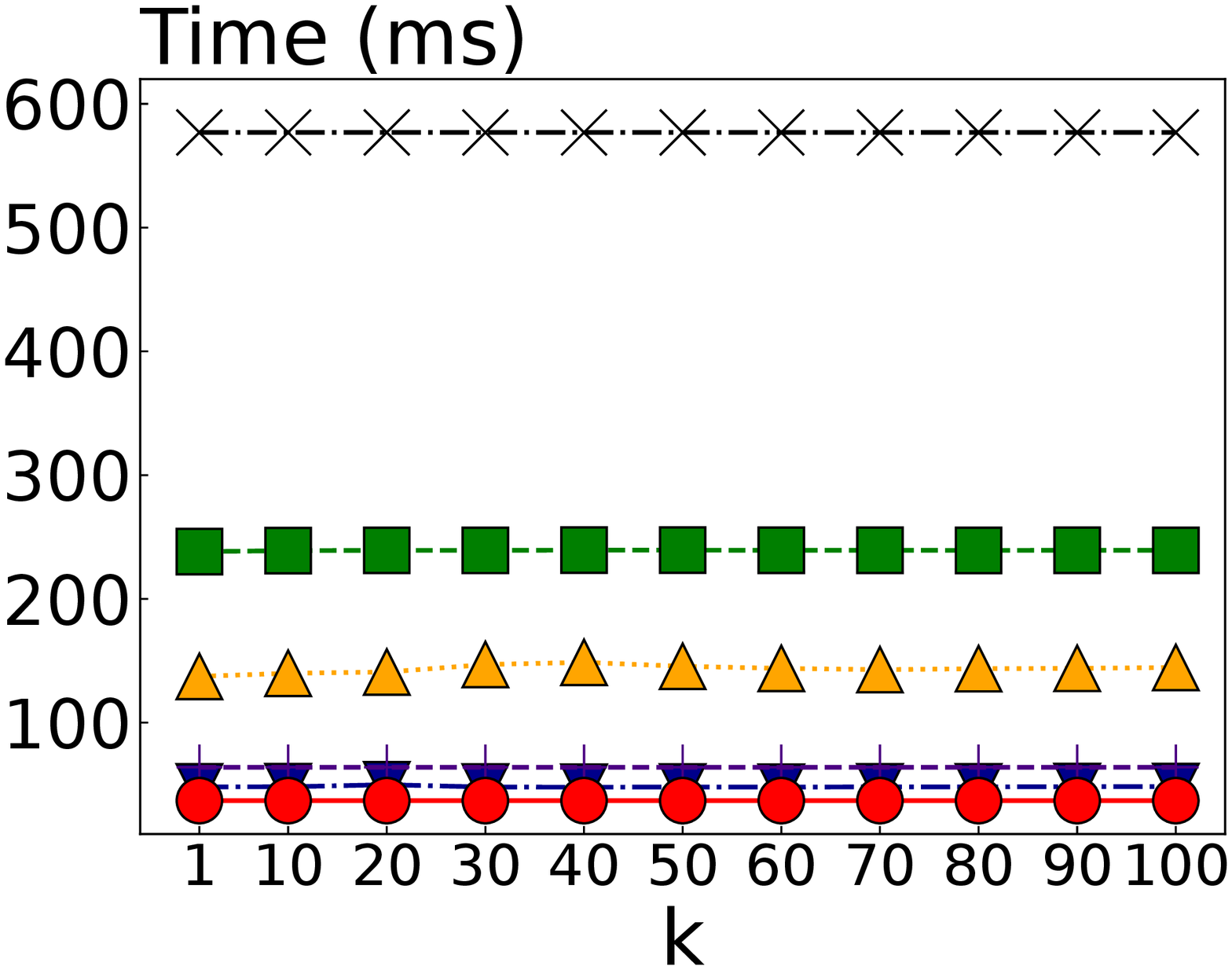}
		\end{minipage}
	}
	\subfigure[Recall on Trevi]{
		\begin{minipage}[c]{0.3\linewidth}
			\centering
			\includegraphics[width=1\textwidth]{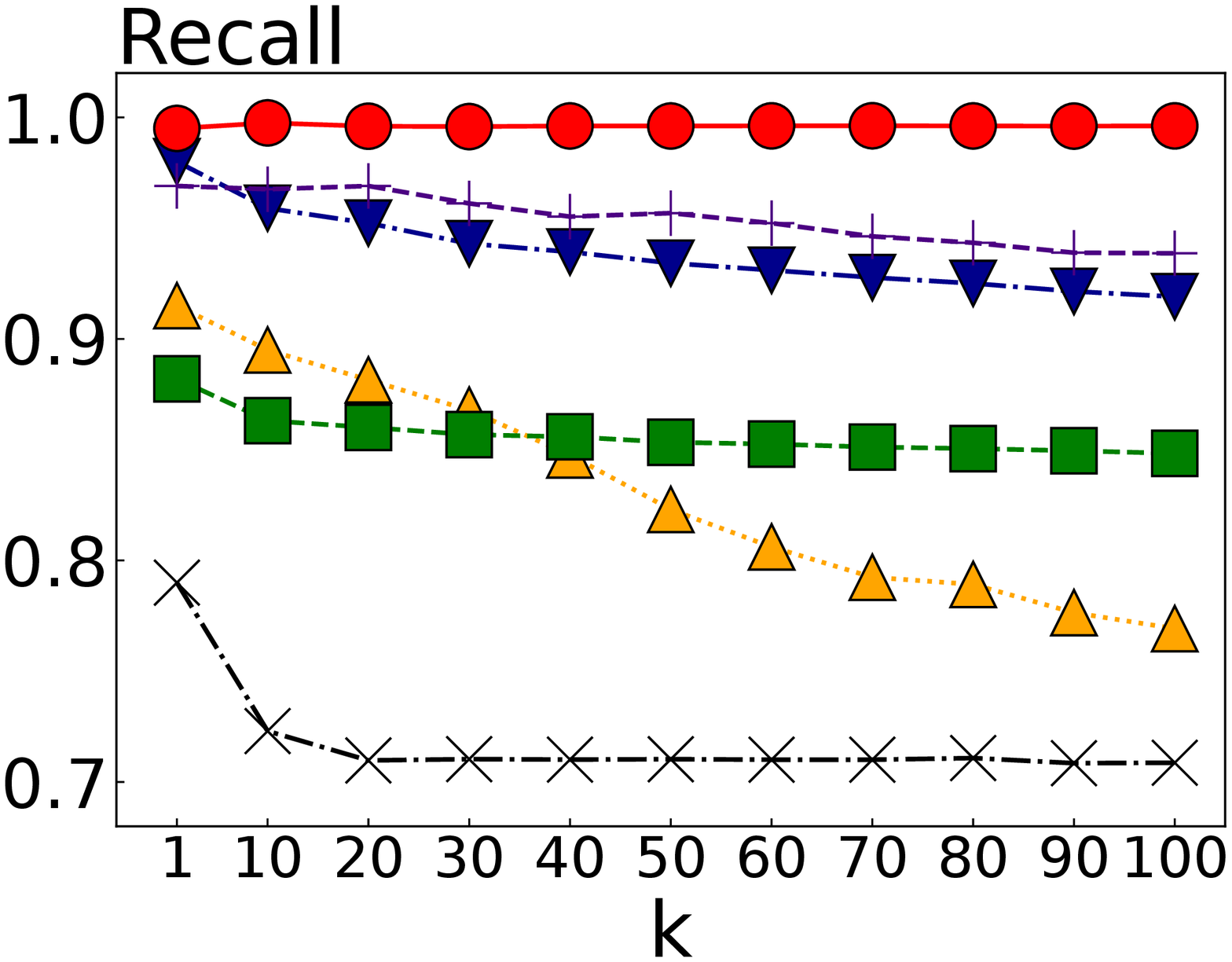}
		\end{minipage}
	}
	\subfigure[OverRatio on Trevi]{
		\begin{minipage}[c]{0.3\linewidth}
			\centering
			\includegraphics[width=1\textwidth]{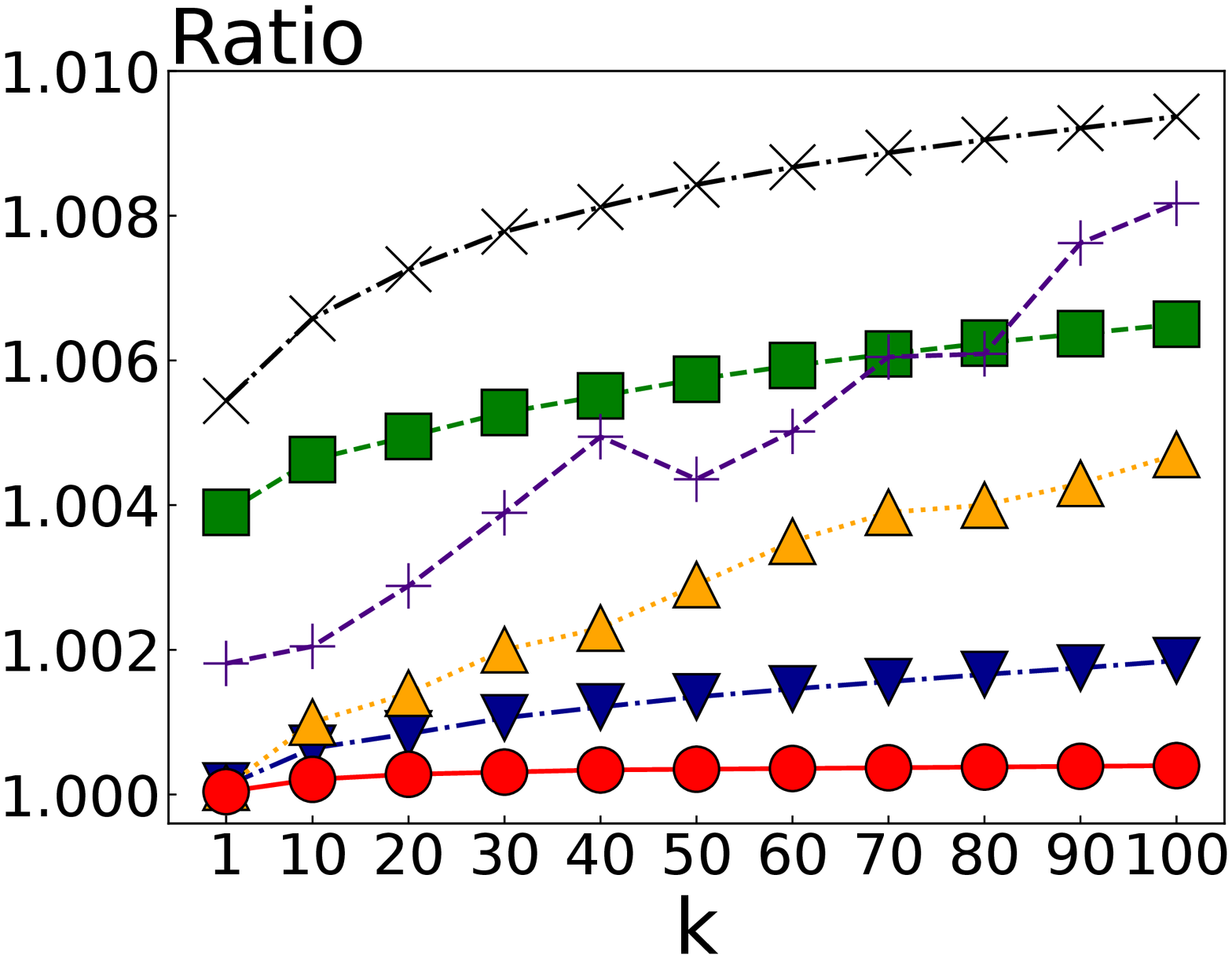}
		\end{minipage}
	}
	\caption{Performance on Trevi when Varying $k$ of NN Queries}
	\label{result_trevi}
\end{figure*}

\begin{figure*}[htbp]
	\centering
	\subfigure[Recall-Time on Cifar]{
		\begin{minipage}[c]{0.3\linewidth}
			\centering
			\includegraphics[width=1\textwidth]{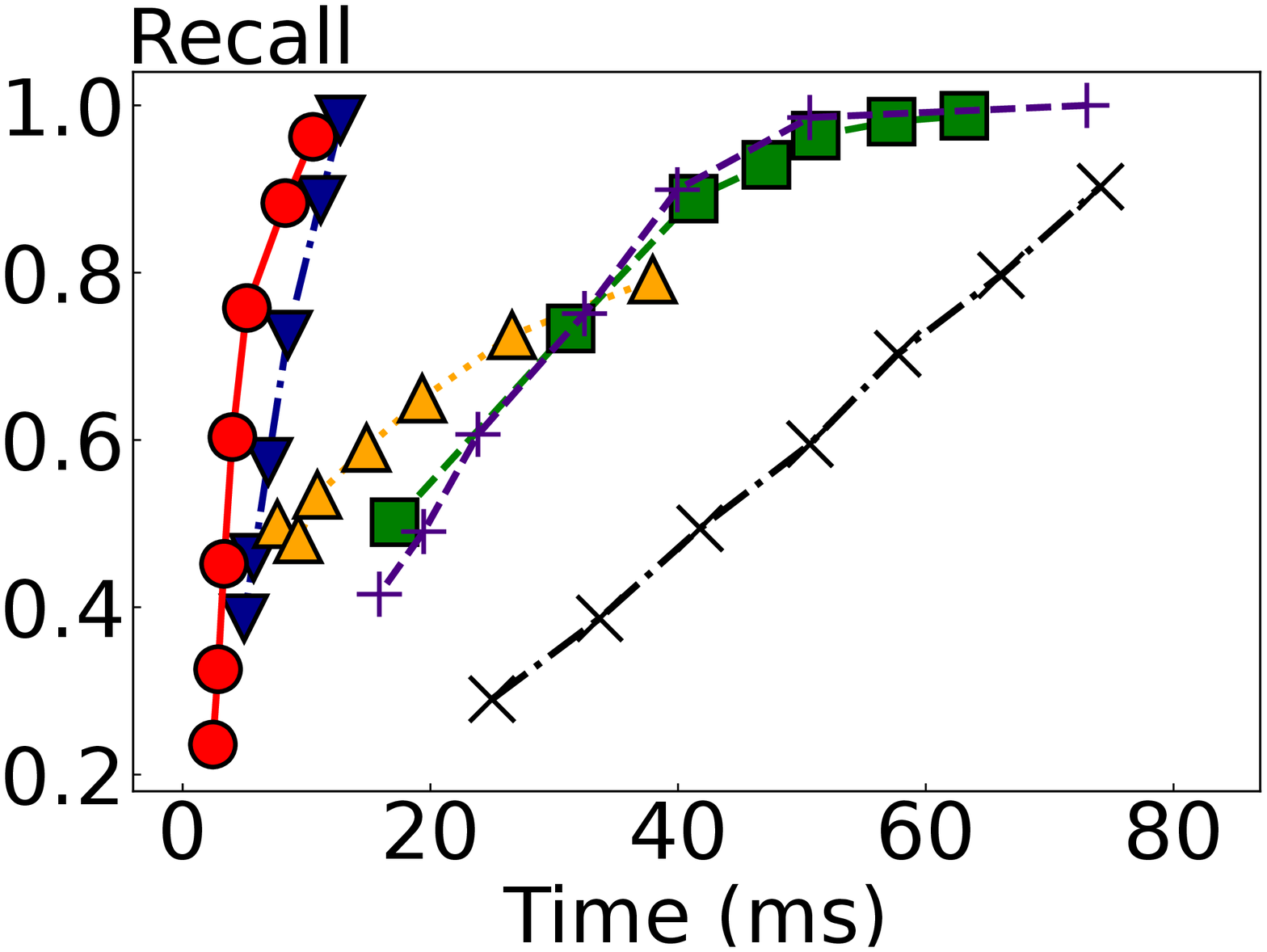}
		\end{minipage}
	}
	\subfigure[Recall-Time on Trevi]{
		\begin{minipage}[c]{0.3\linewidth}
			\centering
			\includegraphics[width=1\textwidth]{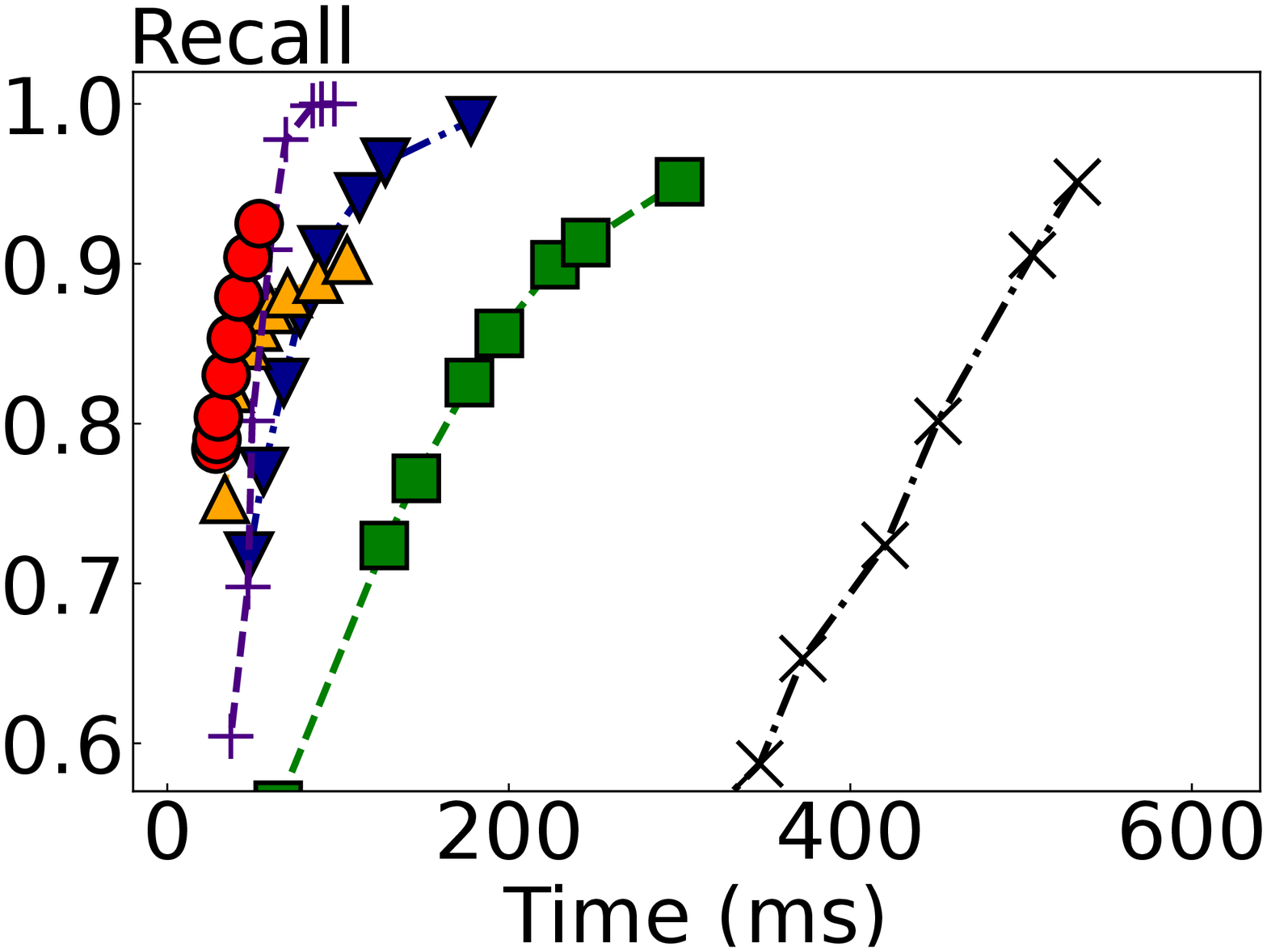}
		\end{minipage}
	}
	\subfigure[Recall-Time on Deep]{
		\begin{minipage}[c]{0.3\linewidth}
			\centering
			\includegraphics[width=1\textwidth]{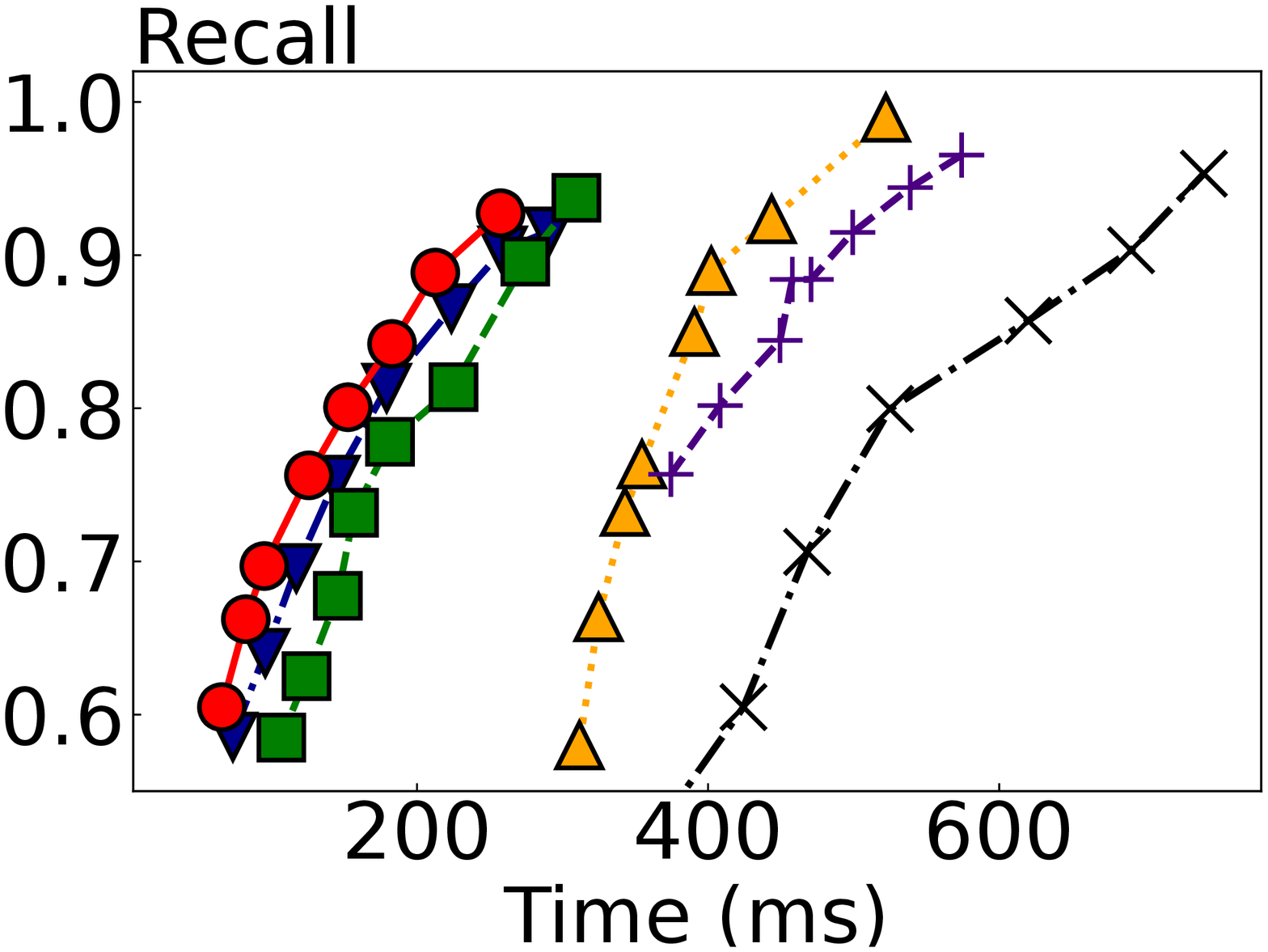}
		\end{minipage}
	}
	\caption{Recall-Time Curve for NN Queries}
	\label{result_recall_time}
\end{figure*}

\begin{figure*}[htbp]
	\centering
	\subfigure[Ratio-Time on Cifar]{
		\begin{minipage}[c]{0.3\linewidth}
			\centering
			\includegraphics[width=1\textwidth]{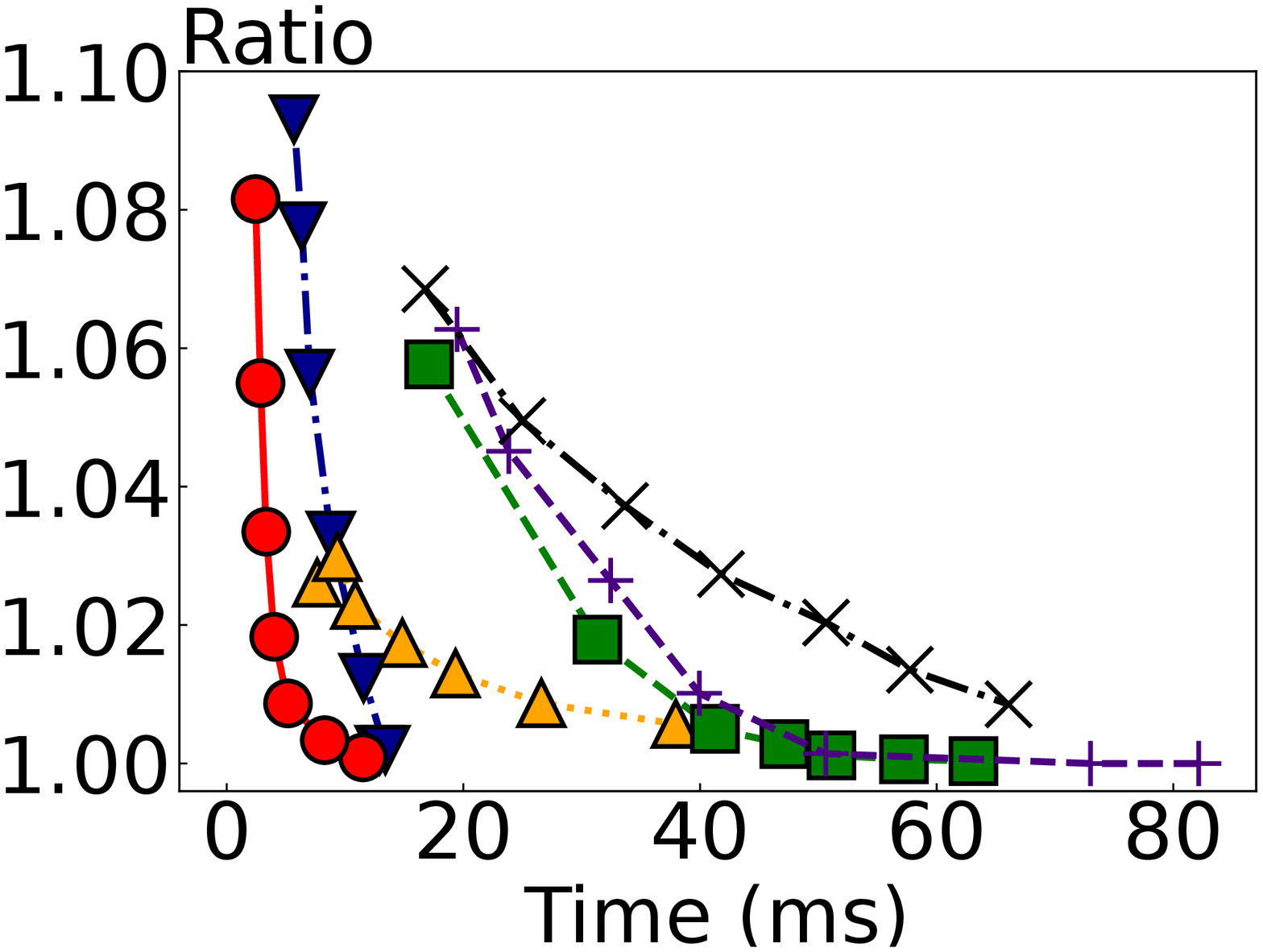}
		\end{minipage}
	}
	\subfigure[Ratio-Time on Trevi]{
		\begin{minipage}[c]{0.3\linewidth}
			\centering
			\includegraphics[width=1\textwidth]{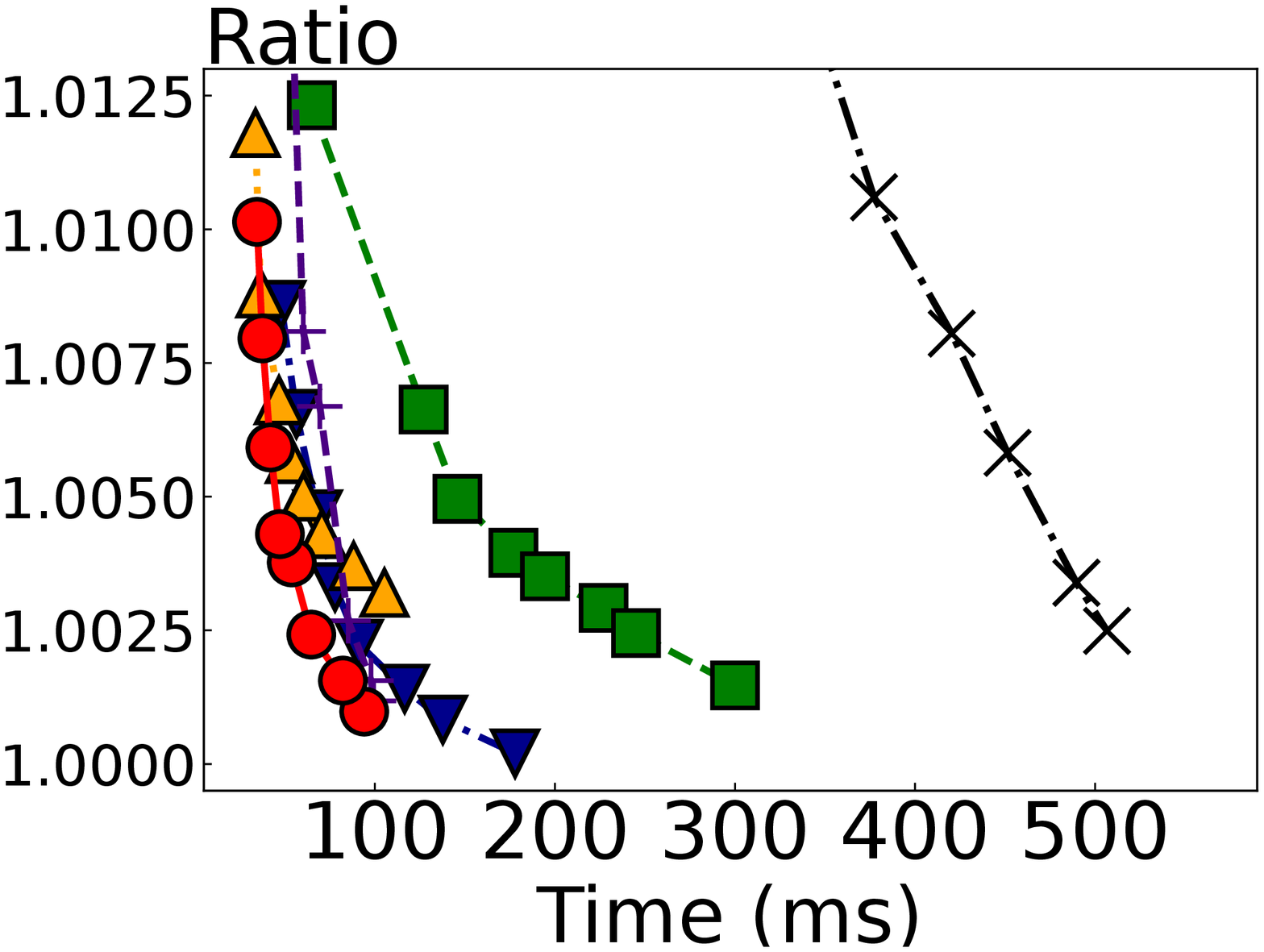}
		\end{minipage}
	}
	\subfigure[Ratio-Time on Deep]{
		\begin{minipage}[c]{0.3\linewidth}
			\centering
			\includegraphics[width=1\textwidth]{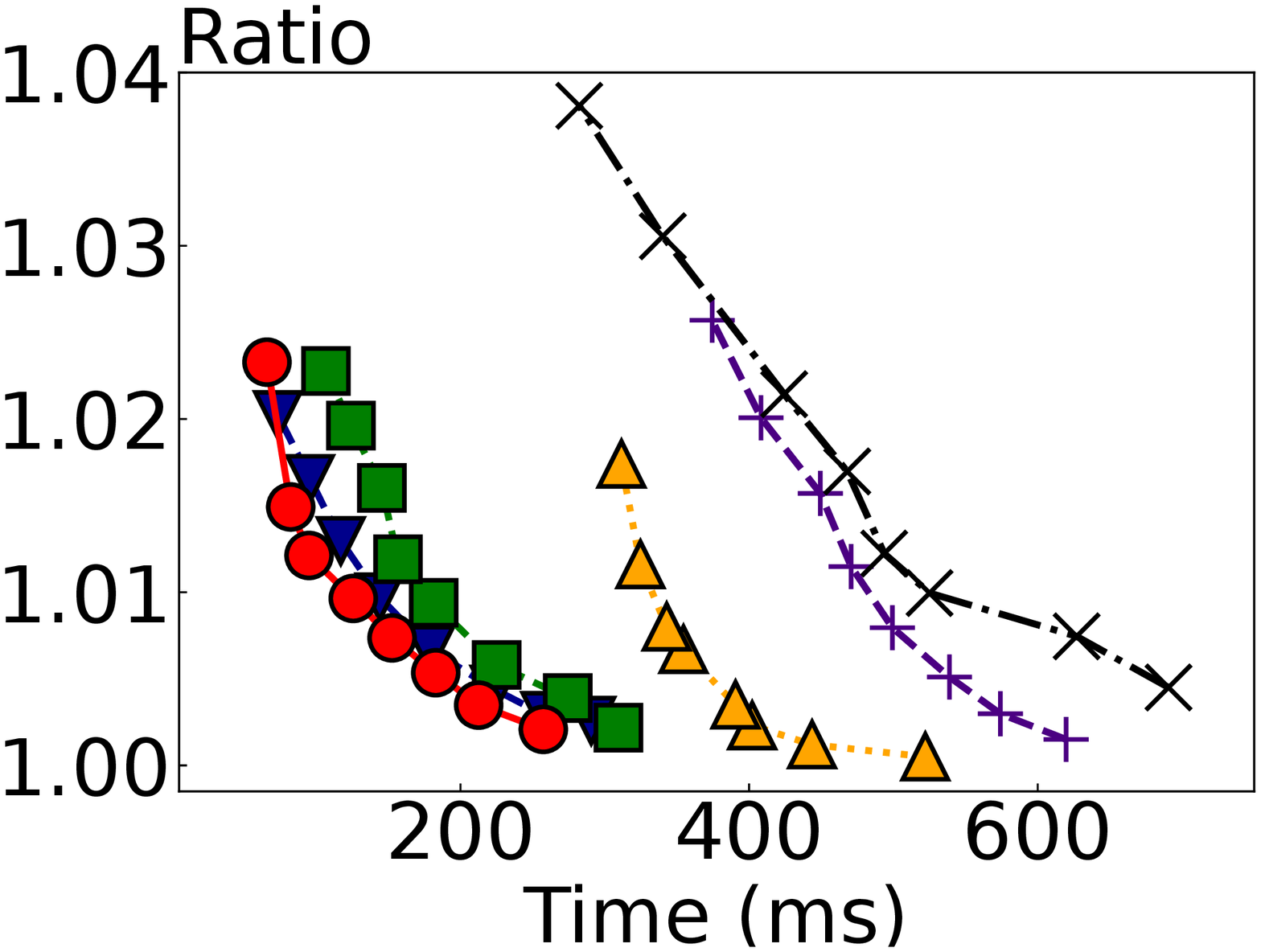}
		\end{minipage}
	}
	\caption{Ratio-Time Curve for NN Queries}
	\label{result_ratio_time}
\end{figure*}

\textbf{Effect of $k$}.
In this set of experiments, we study the performance when varying $k$ in $\{1, 10, 20,$ $ \cdots , 100\}$. Due to the space limitation, we only report the performance on \textit{Deep}, \textit{Cifar}, and \textit{Trevi}. The results are shown in Figs. \ref{result_Cifar}--\ref{result_trevi}.
In the \textit{Cifar} and \textit{Trevi} datasets, we can see that PM-LSH achieves the best performance on all metrics. SRS is the second-best algorithm. When using the \textit{Deep} dataset, PM-LSH has the smallest query time and overall ratio, and its recall is close to that of SRS.

As $k$ increases, all algorithms achieve a higher overall ratio and a smaller recall, but the query time is relatively steady. In fact, the algorithms return the best $k$ objects from a candidate set whose size exceeds $\beta n+k$. Therefore, a larger $k$ has little effect on the query time but obviously has an adverse effect on the result quality.

When considered across different datasets with different cardinality $n$ and dimension $d$, PM-LSH exhibits a consistent high accuracy. This is because PM-LSH is unaffected by the dimensionality of the datasets and because its cost is sublinear in the cardinality of the datasets. In contrast, Multi-Probe is affected significantly by the dataset dimensionality. The hash number of QALSH is $O(n\log n)$, so its query time increases super-linearly with the dataset cardinality. Similarly, when the dataset cardinality increases, SRS incurs a higher query cost to find an NN in the projected space.

To sum up, PM-LSH has the smallest query time among all competitors. In addition, the accuracy is high.
Only SRS is able to achieve a competitive recall in some cases but exhibits longer query times than PM-LSH.

\textbf{Recall-Time and OverallRatio-Time Curves}.
In this set of experiments, we evaluate the relationship between the recall or overall ratio and the query time for $(c,k)$-ANN queries on all the datasets when varying $c$ to obtain different query times. The results are shown in Figs. \ref{result_recall_time} and \ref{result_ratio_time}. As the tradeoff between the query quality and the query time is the key tradeoff, the LSH methods focus on returning relatively good results with much smaller query time than the exact NN algorithms. 
The results show that all algorithms return more accurate results when more query time is used. 
They also show that PM-LSH achieves superior efficiency and accuracy when compared to SRS, QALSH, and Multi-Probe. This can be explained as follows. First, PM-LSH has a better distance estimator than QALSH and Multi-Probe, so PM-LSH outperforms them with the same number of retrieved points. Second, PM-LSH needs lower time to obtain the same number of retrieved points since only one or two range queries are required. In contrast, SRS needs $T$ rounds of incremental NN search.

\subsection{Evaluation of CP Query Processing}
To evaluate PM-LSH for CP query processing, we first conduct an evaluation to determine the setting of $\gamma$ and compare two \textsf{Promote} methods. Then, we compare with the competitors by varying the parameter values. Finally, we consider the changes of the overall ratio and recall under different query times.

\textbf{Determining the Setting of $\gamma$.}
In this set of experiments, we study the effects of the node capacity $M$ and the dataset cardinality on choosing $\gamma$ in datasets \textit{Audio}, \textit{Trevi}, and \textit{NUS}. We choose $M=16$ and \textsf{m$\_$RAD} as defaults. We randomly sample $n'=10K$ points from each dataset. After we build a PM-tree, we compute the value of $\gamma$ for each pair and use the probability density distribution function $f_{\gamma}(x)$ to study the effects.

We first consider $f_{\gamma}(x)$ when varying the value of $M$ in $\{2, 16, 64\}$. As shown in Fig. \ref{fig:rca_m}, the tendency of $f_{\gamma}(x)$ remains nearly unchanged when varying $M$. However, the peak position, the peak value, and the gradient are affected slightly by $M$. To make $\Pr(\gamma)=0.85$, the settings for $\gamma$ are different. 
Note that when $M=2$, $f_{\gamma}(x)$ has the smallest peak position, the largest peak value, and the largest gradient. This indicates that a small $M$ yields a good partitioning. However, a small $\gamma$ increases the PM-tree size and leads to additional computational costs. To achieve a good tradeoff, we set $M=16$.

Next, we study $f_{\gamma}(x)$ when varying the number of sampled points $n'$ in $\{5000, 10000, 20000\}$. As shown in Fig. \ref{fig:rca_n}, $f_{\gamma}(x)$ changes slightly when varying $n$, which enables us to determine the setting of $\gamma$ by using only a subset that preserves the information of the whole dataset. 
The cost of computing $\gamma$ equals the time needed to compute the distances of $50$ million point pairs formed by $10K$ points, which is about 0.3s when we use $m=15$ hash functions for each dataset.

\begin{figure*}[t]
	\centering
	\subfigure[$f_{\gamma}(x)$ on Audio]{
		\begin{minipage}[c]{0.3\linewidth}
			\centering
			\includegraphics[width=1\textwidth]{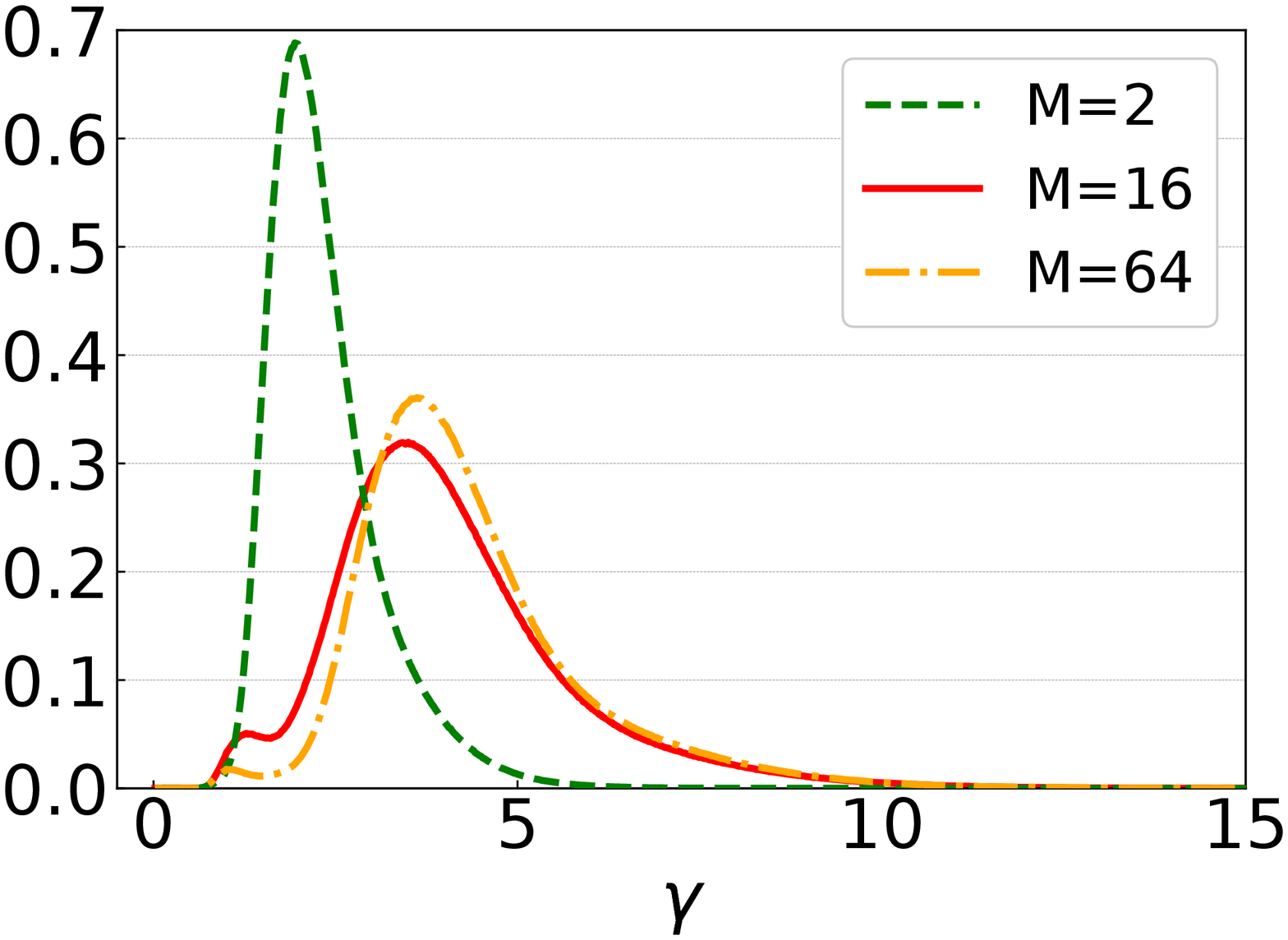}
		\end{minipage}
	}
	\subfigure[$f_{\gamma}(x)$ on Trevi]{
		\begin{minipage}[c]{0.3\linewidth}
			\centering
			\includegraphics[width=1\textwidth]{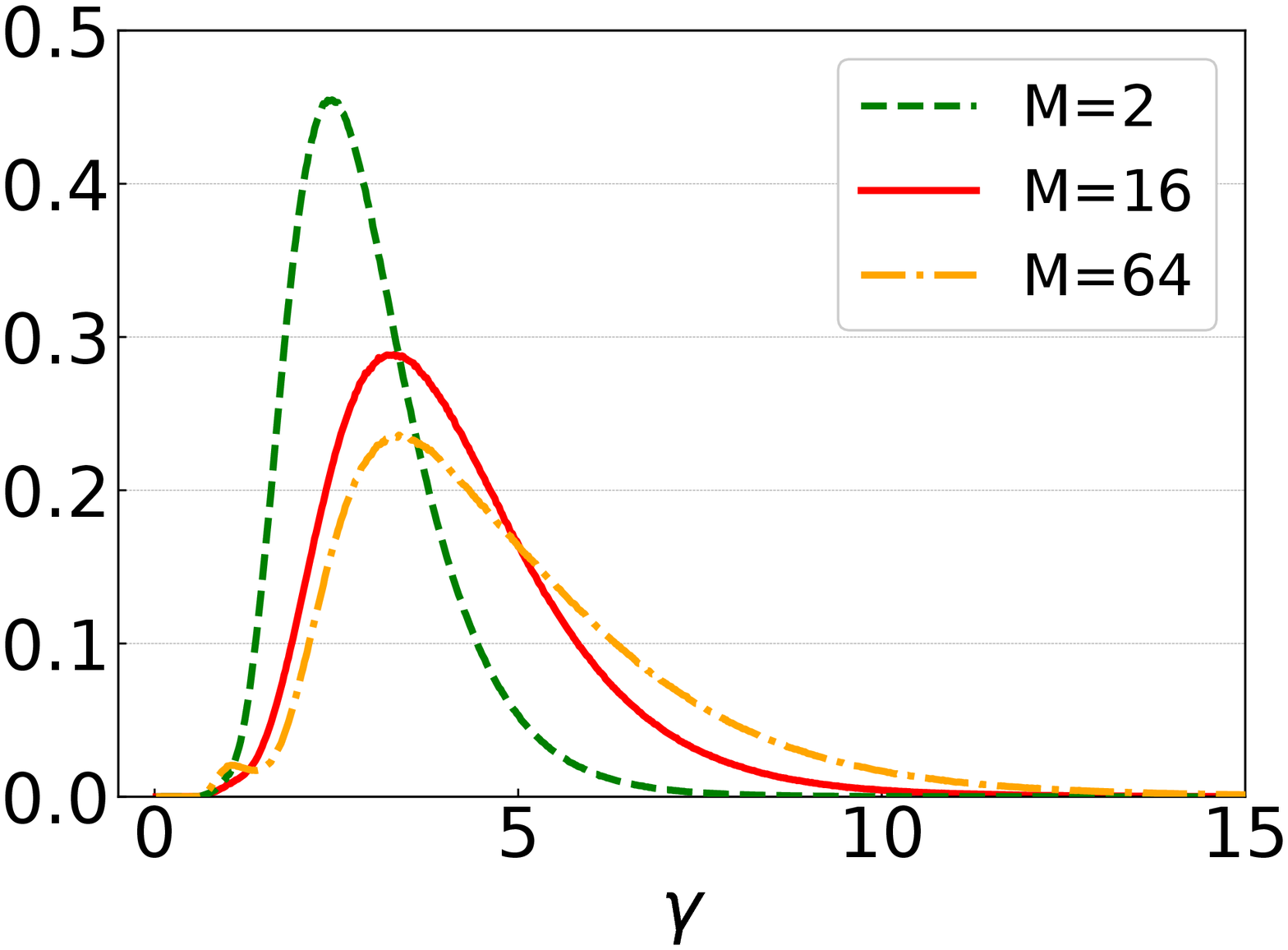}
		\end{minipage}
	}
	\subfigure[$f_{\gamma}(x)$ on NUS]{
		\begin{minipage}[c]{0.3\linewidth}
			\centering
			\includegraphics[width=1\textwidth]{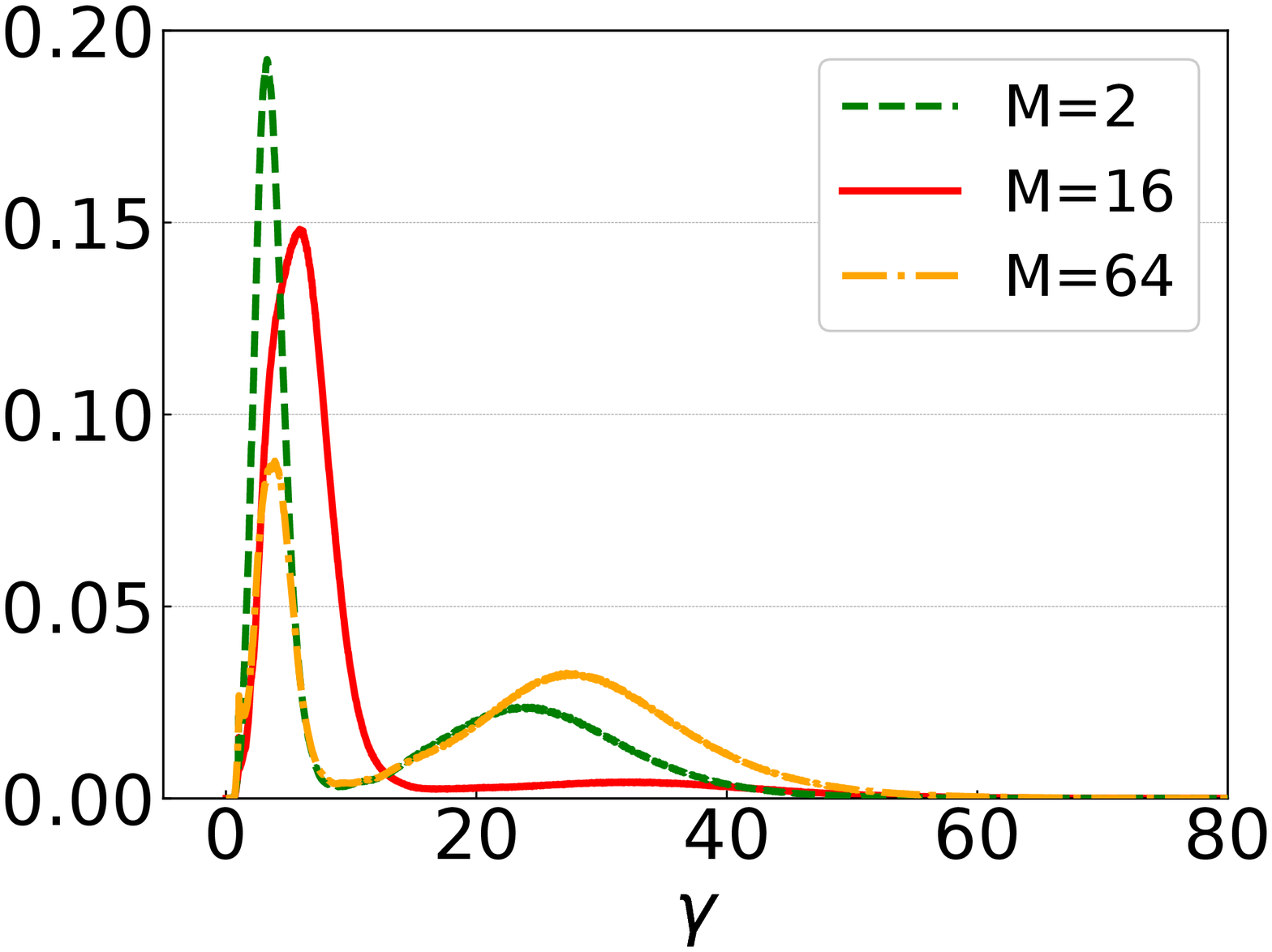}
		\end{minipage}
	}
	\centering
	\caption{Effect of $M$ for $f_{\gamma}(x)$}
	\label{fig:rca_m}
\end{figure*}

\begin{figure*}[t]
	\centering
	\subfigure[$f_{\gamma}(x)$ on Audio]{
		\begin{minipage}[c]{0.3\linewidth}
			\centering
			\includegraphics[width=1\textwidth]{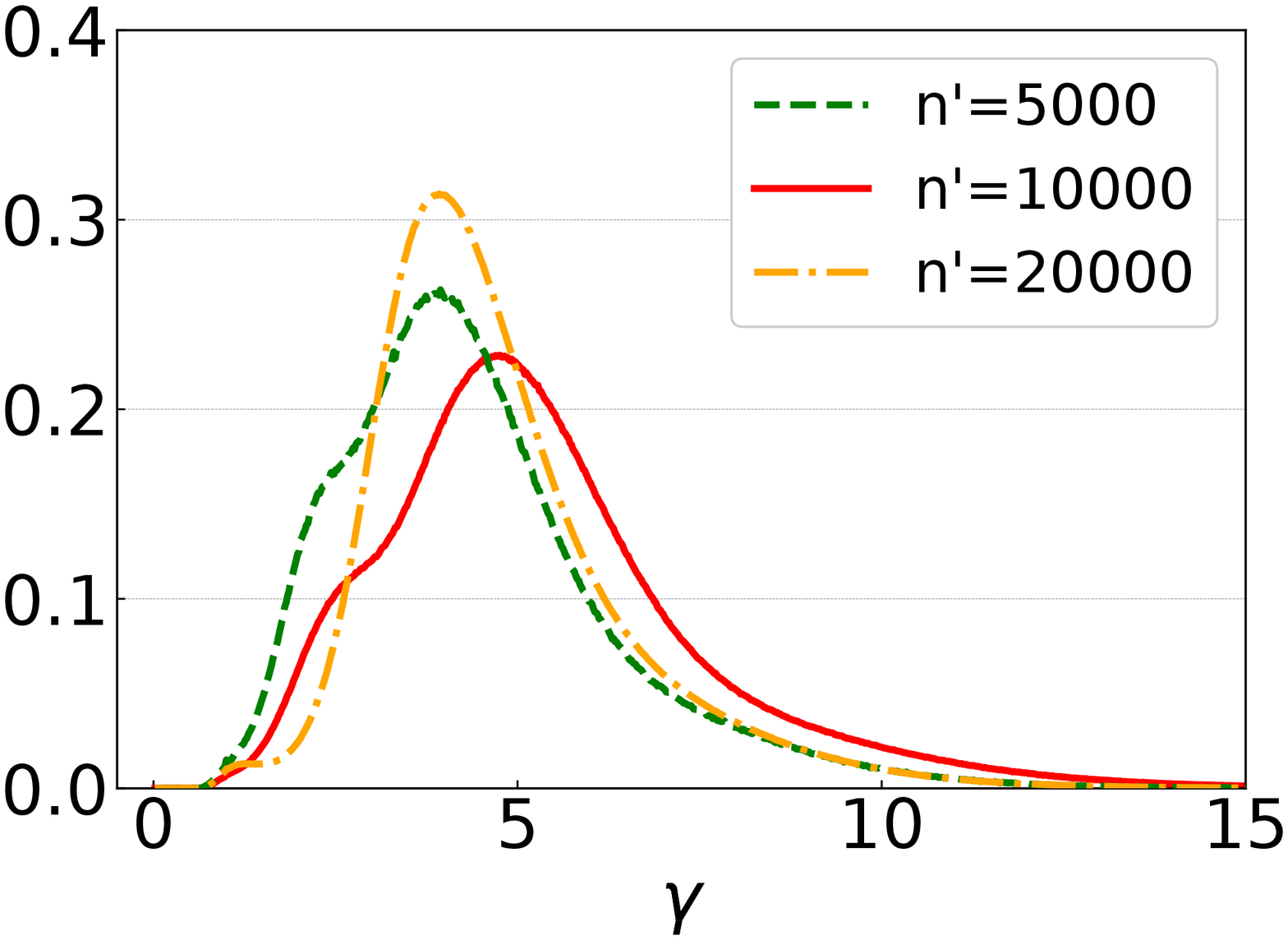}
		\end{minipage}
	}
	\subfigure[$f_{\gamma}(x)$ on Trevi]{
		\begin{minipage}[c]{0.3\linewidth}
			\centering
			\includegraphics[width=1\textwidth]{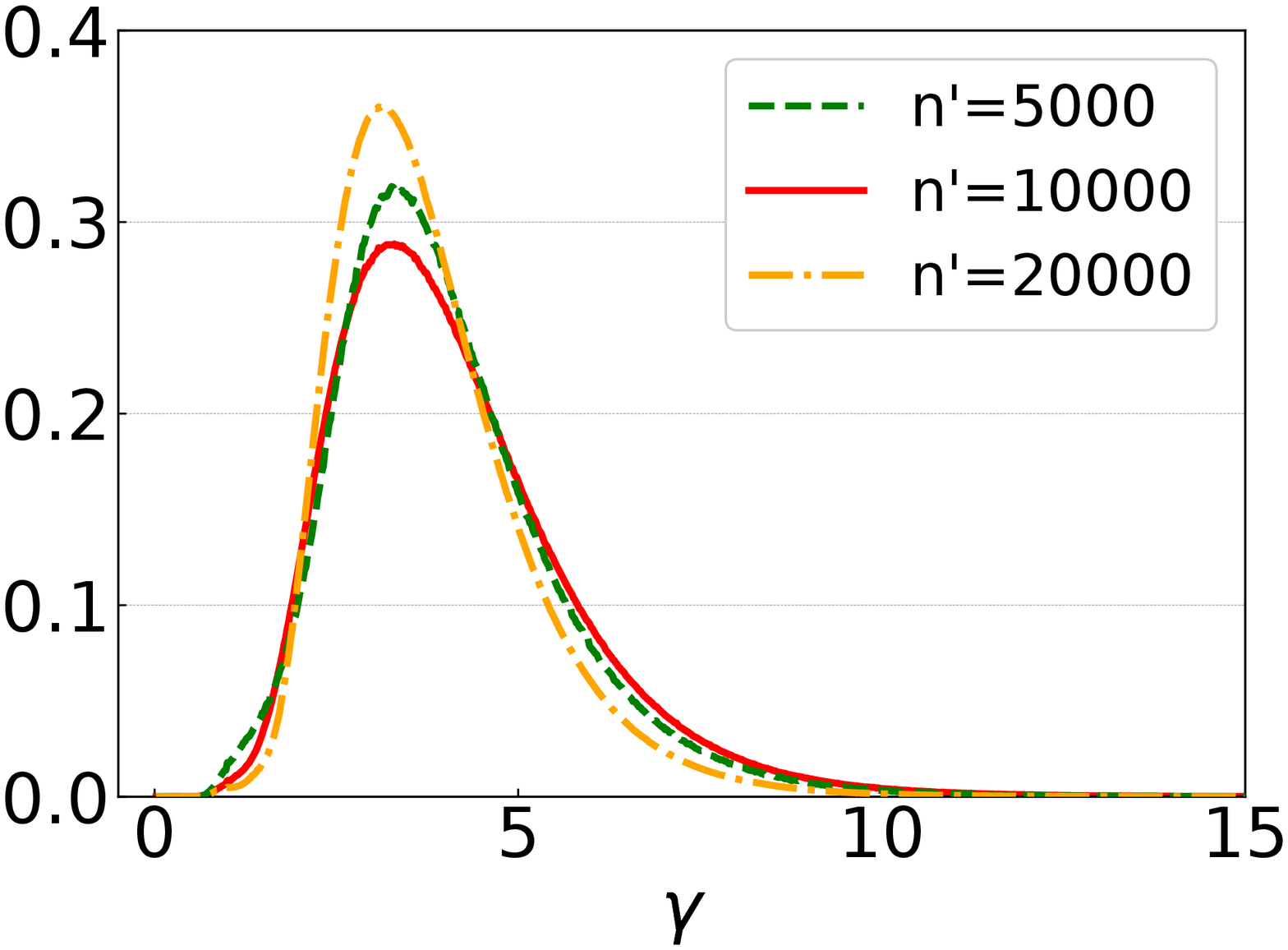}
		\end{minipage}
	}
	\subfigure[$f_{\gamma}(x)$ on NUS]{
		\begin{minipage}[c]{0.3\linewidth}
			\centering
			\includegraphics[width=1\textwidth]{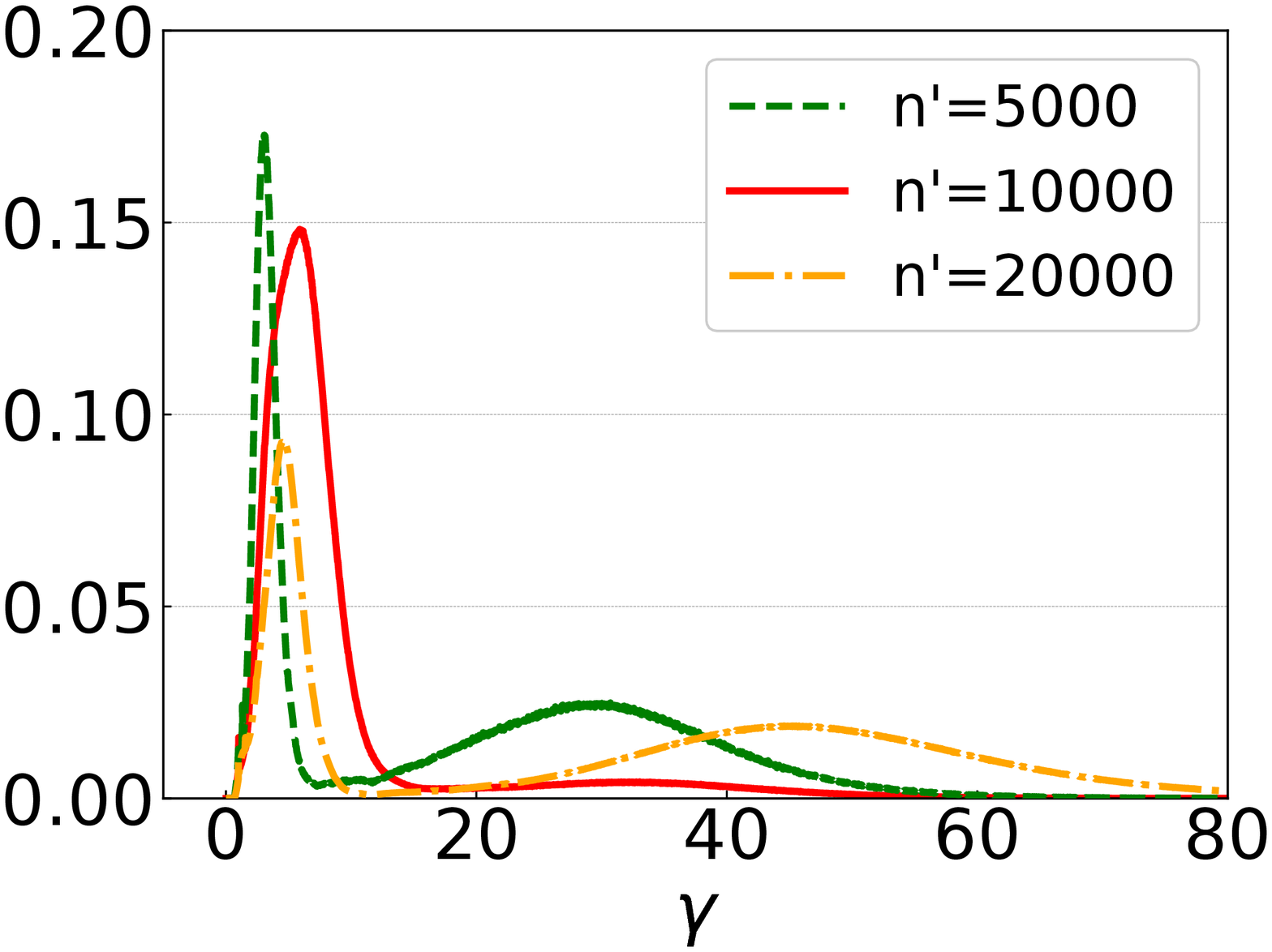}
		\end{minipage}
	}
	\caption{Effect of Dataset Cardinality for $f_{\gamma}(x)$}
	\label{fig:rca_n}
\end{figure*}

\begin{figure*}[t]
	\centering	
	\subfigure[Time (s)]{
		\begin{minipage}[c]{0.3\linewidth}
			\centering
			\includegraphics[width=1\textwidth]{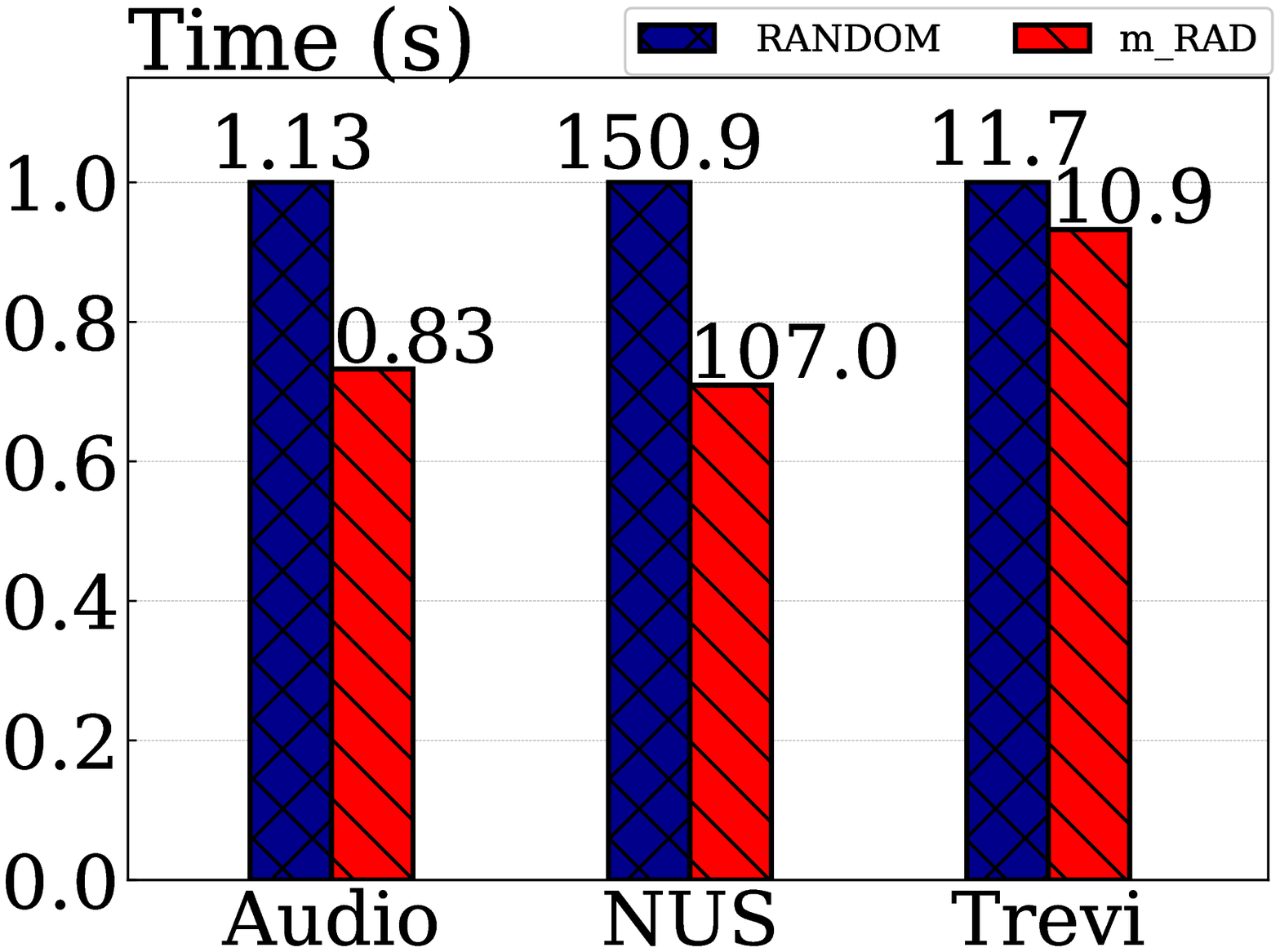}
		\end{minipage}
	}
	\subfigure[Recall]{
		\begin{minipage}[c]{0.3\linewidth}
			\centering
			\includegraphics[width=1\textwidth]{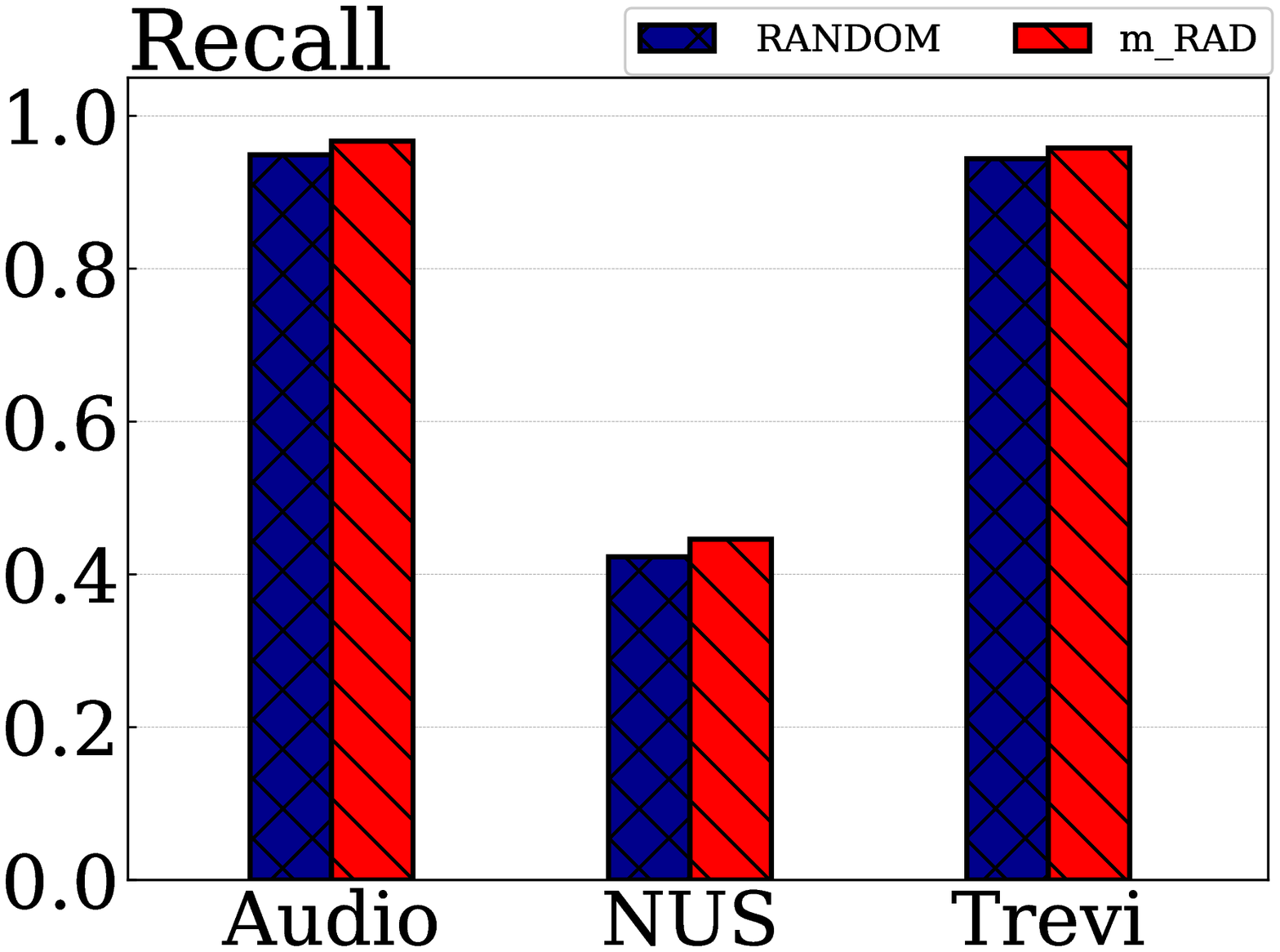}
		\end{minipage}
	}
	\subfigure[OverallRatio]{
		\begin{minipage}[c]{0.3\linewidth}
			\centering
			\includegraphics[width=1\textwidth]{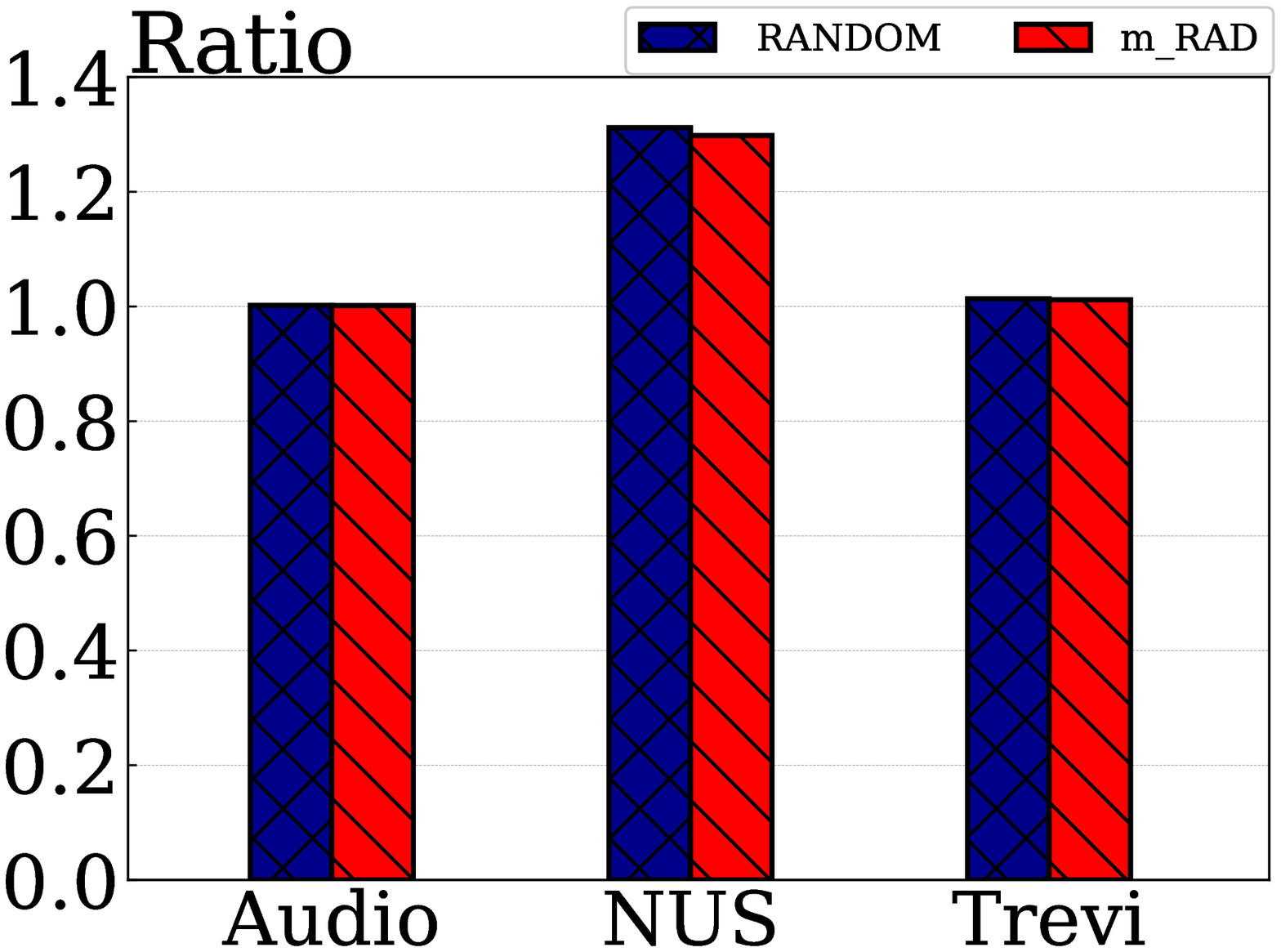}
		\end{minipage}
	}
	\caption{Effect of \textsf{Promote} methods}
	\label{fig:promote}
\end{figure*}

\textbf{Effect of \textsf{Promote} methods.}
We compare the performance of the two \textsf{Promote} methods, \textsf{m$\_$RAD} and \textsf{RANDOM}. Fig. \ref{fig:promote} shows that the recall and overall ratio are very similar for the two methods, but that the query time when using \textsf{m$\_$RAD} is smaller than that achieved when using \textsf{RANDOM}. This can be explained by the fact that the PM-tree constructed with \textsf{m$\_$RAD} has a better structure, meaning that fewer candidate pairs need to be verified to achieve a high recall. So we choose \textsf{m$\_$RAD} as the default \textsf{Promote} method. On the other hand, Table \ref{tb:indexTime} shows that the construction of the PM-tree with \textsf{m$\_$RAD} takes more time than with \textsf{RANDOM}, while still being acceptable.

\begin{table}
	\centering
	\caption{Construction Time of \textsf{m$\_$RAD} and \textsf{RANDOM}}
	\label{tb:indexTime}
	\begin{tabular}{|c|c| m{1.5cm}<{\centering}|}
		\cline{1-3} 
		\multirow{2}*{Dataset}  &\multicolumn{2}{c|}{Construction Time (s)} \\		
		\cline{2-3} 
		&\textsf{RANDOM}&\textsf{m$\_$RAD} \\		\cline{1-3}
		Audio&0.82&28.75 \\
		\cline{1-3}
		NUS&2.84&116.81  \\	
		\cline{1-3}
		Trevi&1.06&45.09  \\
		\cline{1-3}
	\end{tabular}
\end{table}

\begin{table*}[t]
	\centering
	\caption{Performance Overview of CP Queries}
	\label{tb:joinoverview}
	\begin{tabular}{|c|c| m{1.5cm}<{\centering} |m{1.5cm}<{\centering}|m{1.5cm}<{\centering}|m{1.7cm}<{\centering}||m{1.5cm}<{\centering}|}
		\cline{1-7}
		\multicolumn{2}{|c|}{} & \textbf{PM-LSH} & \textbf{LSB-tree}& \textbf{ACP-P} & \textbf{M$ k $CP}  & \textbf{NLJ} \\
		\cline{1-7}
		\multirow{3}*{\textbf{Audio}}&Query Time (s)&\textbf{0.83} &12.82 &384.60&756.26  &388.03 \\	
		\cline{2-7}
		\multicolumn{1}{|c|}{}&Overall Ratio&\textbf{1.002} &1.004 &1.004&1.083  &1.000 \\
		\cline{2-7}
		\multicolumn{1}{|c|}{}&Recall&\textbf{0.964} &0.911&0.930&0.288  &1.000 \\
		\cline{1-7}
		\multirow{3}*{\textbf{MNIST}}&Query Time (s)&\textbf{33.5}9 &38.80 &597.53&2946.45  &1900.42 \\
		\cline{2-7}
		\multicolumn{1}{|c|}{}&Overall Ratio&\textbf{1.004} &1.006 &1.005&1.103  &1.000 \\
		\cline{2-7}
		\multicolumn{1}{|c|}{}&Recall&\textbf{0.937} &0.911 &0.928&0.313  &1.000 \\
		\cline{1-7}
		\multirow{3}*{\textbf{NUS}}&Query Time (s)&\textbf{107.03} &179.43 &921.19&/  &23322.10 \\
		\cline{2-7}
		\multicolumn{1}{|c|}{}&Overall Ratio&\textbf{1.298} &3.904 &1.669&/  &1.000 \\
		\cline{2-7}
		\multicolumn{1}{|c|}{}&Recall&\textbf{0.446}&0.005 &0.190&/  &1.000 \\
		\cline{1-7}
		\multirow{3}*{\textbf{Trevi}}&Query Time (s)&\textbf{10.92} &66.96 &933.33&/  &28400.6 \\	
		\cline{2-7}
		\multicolumn{1}{|c|}{}&Overall Ratio&\textbf{1.014} &1.019 &1.016&/  &1.000 \\
		\cline{2-7}
		\multicolumn{1}{|c|}{}&Recall&\textbf{0.946} &0.905 & 0.918&/ &1.000 \\
		\cline{1-7}
		\multirow{3}*{\textbf{Cifar}}&Query Time (s)&\textbf{91.83} &106.18 &376.17&4140.29  &2609.30 \\	
		\cline{2-7}
		\multicolumn{1}{|c|}{}&Overall Ratio&\textbf{1.034} &1.070 &1.047&1.094  &1.000 \\
		\cline{2-7}
		\multicolumn{1}{|c|}{}&Recall&\textbf{0.721} &0.499 &0.619&0.449  &1.000 \\
		\cline{1-7}
		\multirow{3}*{\textbf{GIST}}&Query Time (s)&\textbf{81.77} &125.45 &985.02&/  &590321.43 \\
		\cline{2-7}
		\multicolumn{1}{|c|}{}&Overall Ratio&\textbf{1.101} &1.998 &1.283&/  &1.000 \\
		\cline{2-7}
		\multicolumn{1}{|c|}{}&Recall& \textbf{0.772}&0.16 &0.504&/  &1.000 \\
		\cline{1-7}
		\multirow{3}*{\textbf{Deep}}&Query Time (s)&\textbf{128.74} &132.73 &129.16&/  &174900.00 \\
		\cline{2-7}
		\multicolumn{1}{|c|}{}&Overall Ratio&\textbf{2.337} &2.420 &7.115&/ &1.000  \\
		\cline{2-7}
		\multicolumn{1}{|c|}{}&Recall&\textbf{0.445} &0.427 &0.192&/  &1.000 \\
		\hline
	\end{tabular}
\end{table*}

\textbf{Performance Overview of CP Query}.
We compare the algorithms with default settings on all datasets and report the query time (s), overall ratio, and recall in Table \ref{tb:joinoverview}. We observe that PM-LSH has the best performance for all evaluation metrics and datasets. 
To analyze what affects the query time of PM-LSH on different datasets, we notice that \textit{Cifar} takes more time than \textit{Trevi}. However, the cardinality and dimensionality of \textit{Cifar} are both smaller than those of \textit{Trevi}, indicating that the query time is not only affected by the dataset cardinality and dimensionality. Other factors, including the data distribution, also have an effect.
All algorithms exhibit a poor performance on \textit{NUS}. This can be explained by \textit{NUS} having a small RC value and a large LID value, which make it challenging to compute CP queries. 
M$ k $CP has the worst performance on all datasets. The reason is that M$ k $CP uses the M-Tree to index points directly, causing vulnerability to the curse of dimensionality. For high-dimensional datasets, the M$ k $CP query algorithm nearly degenerates to being a brute-force algorithm. In practice, operations such as computing lower bounds and maintaining priority queues incur additional costs.

\begin{figure*}[htbp]
	\centering
	\subfigure{
		\includegraphics[width=0.8\textwidth]{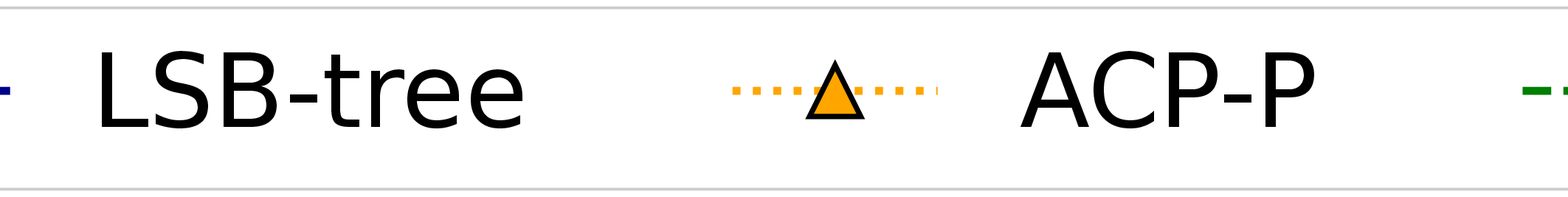}
	}

	\subfigure[Time on Audio]{
		\begin{minipage}[c]{0.3\linewidth}
			\centering
			\includegraphics[width=1\textwidth]{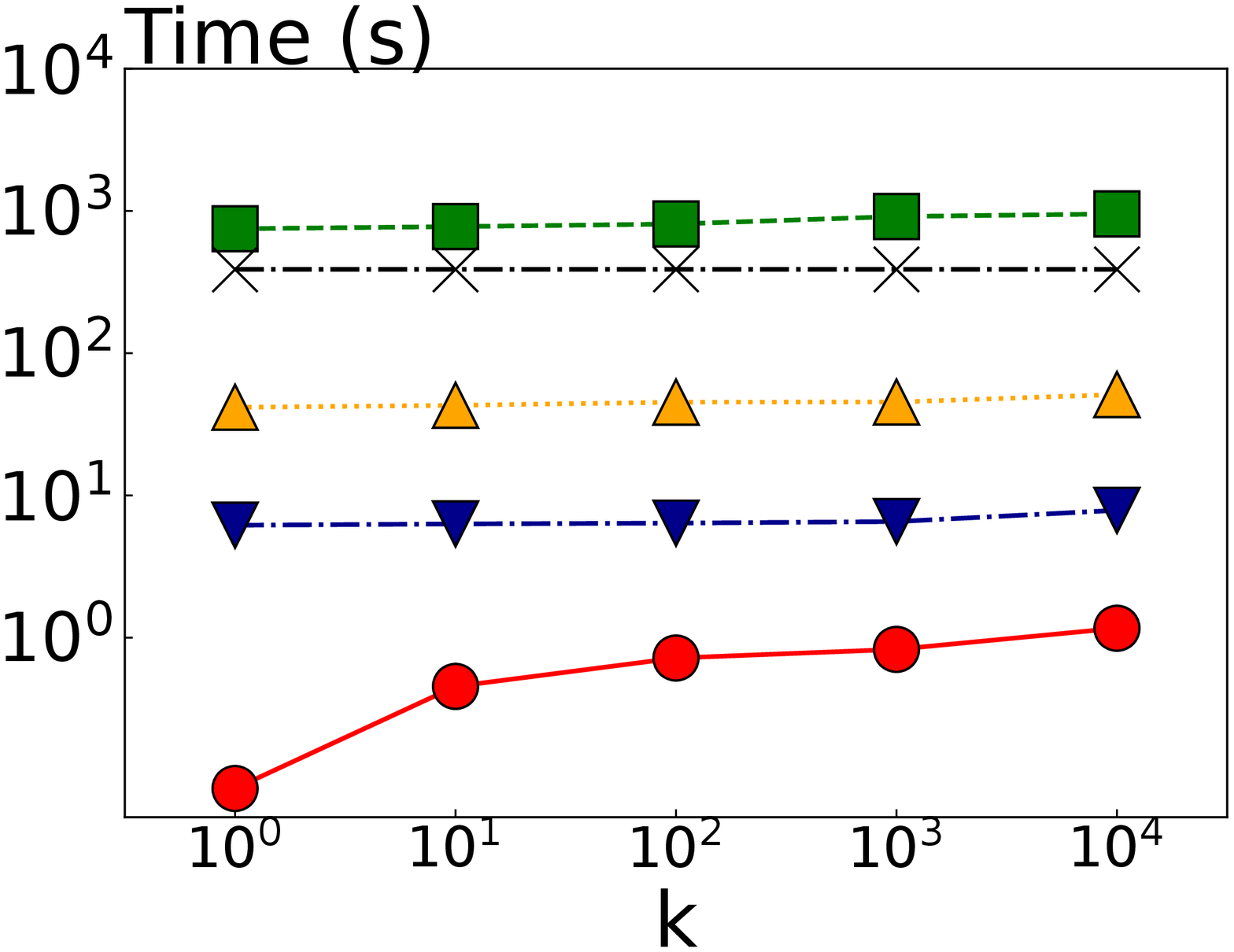}
		\end{minipage}
	}
	\subfigure[Recall on Audio]{
		\begin{minipage}[c]{0.3\linewidth}
			\centering
			\includegraphics[width=1\textwidth]{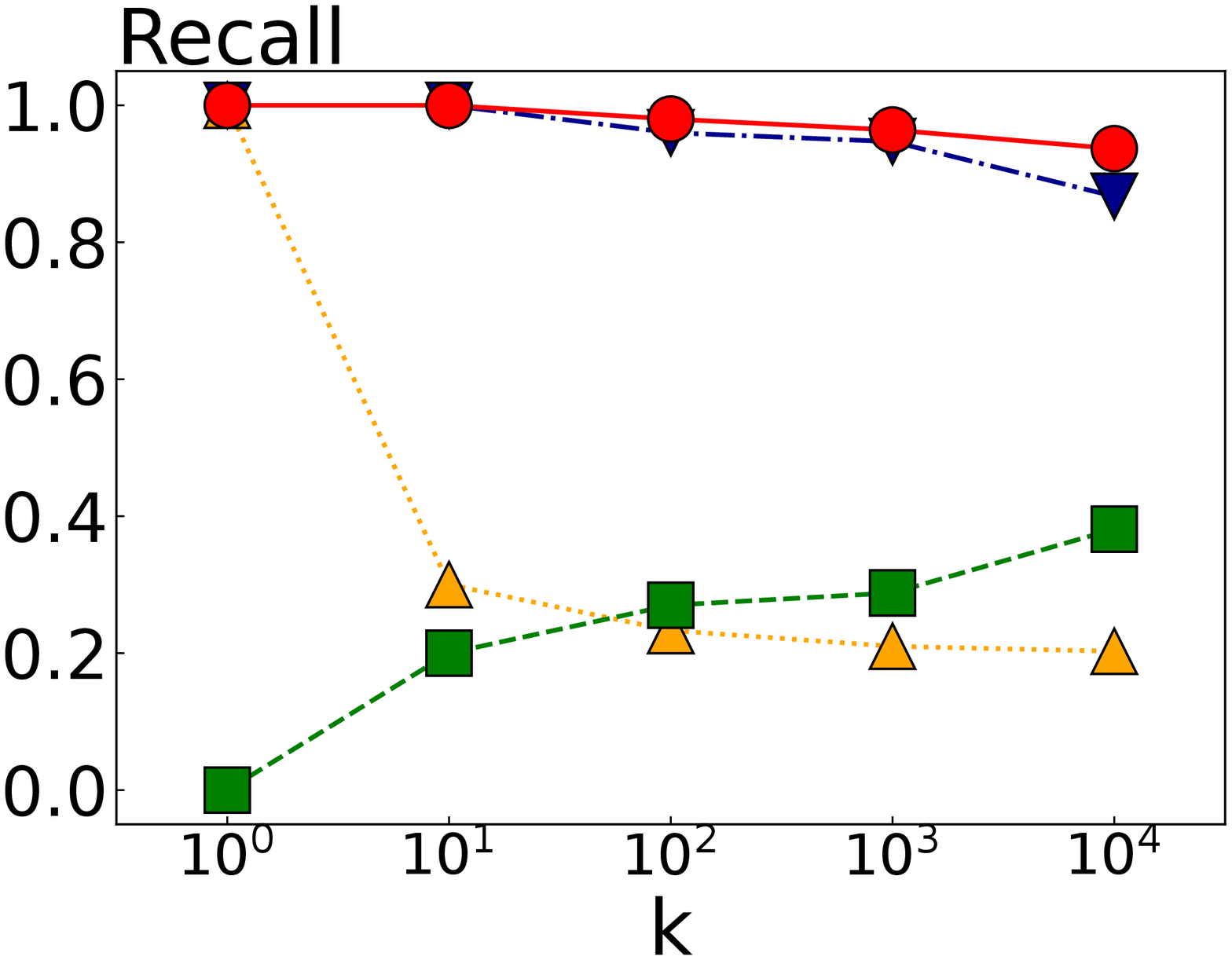}
		\end{minipage}
	}
	\subfigure[Ratio on Audio]{
		\begin{minipage}[c]{0.3\linewidth}
			\centering
			\includegraphics[width=1\textwidth]{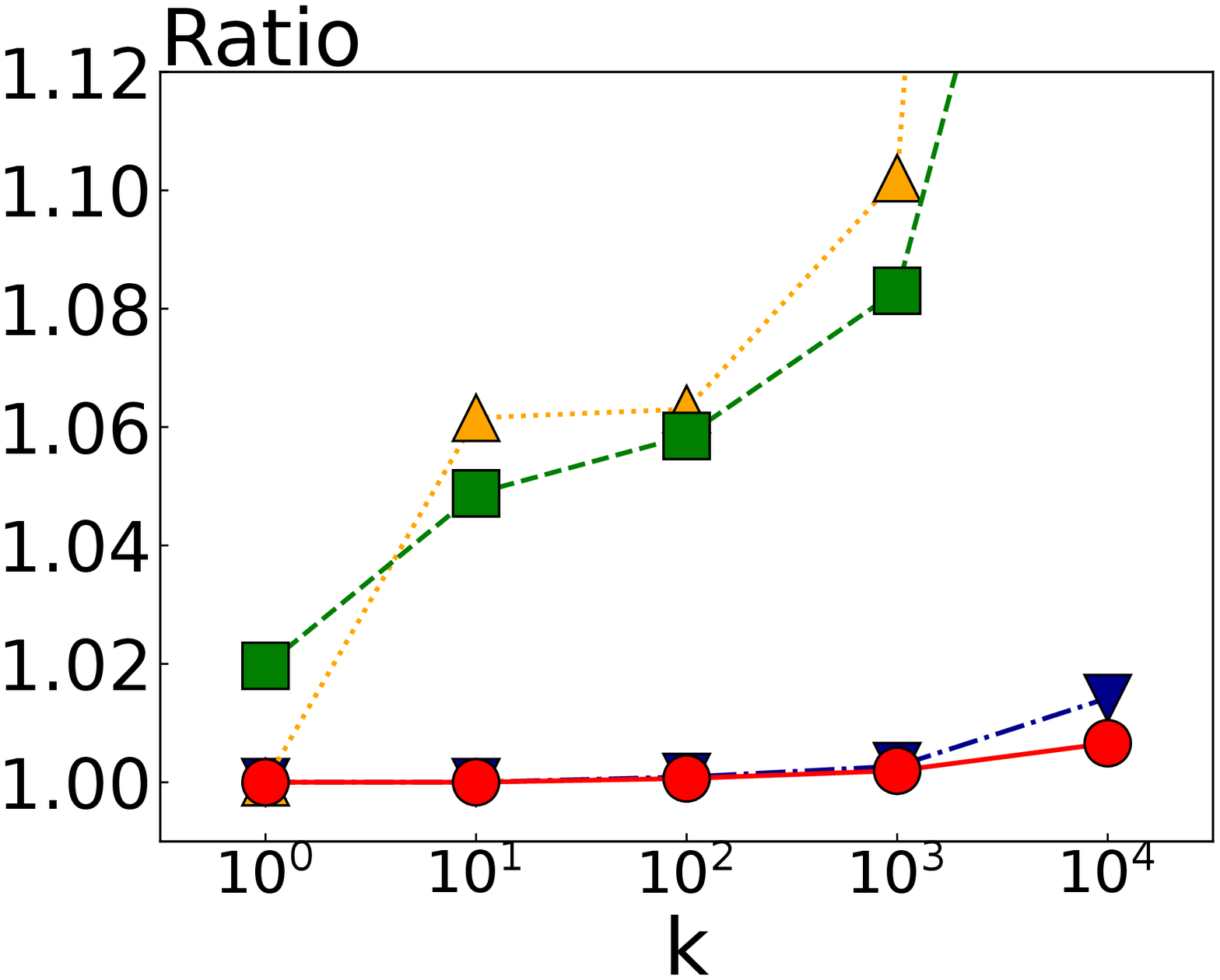}
		\end{minipage}
	}
	\caption{Performance on Audio when Varying $ k $ of CP Queries}
	\label{fig:cp_knn_Audio}
\end{figure*}

\begin{figure*}[htbp]
	\centering
	\subfigure[Time on NUS]{
		\begin{minipage}[c]{0.3\linewidth}
			\centering
			\includegraphics[width=1\textwidth]{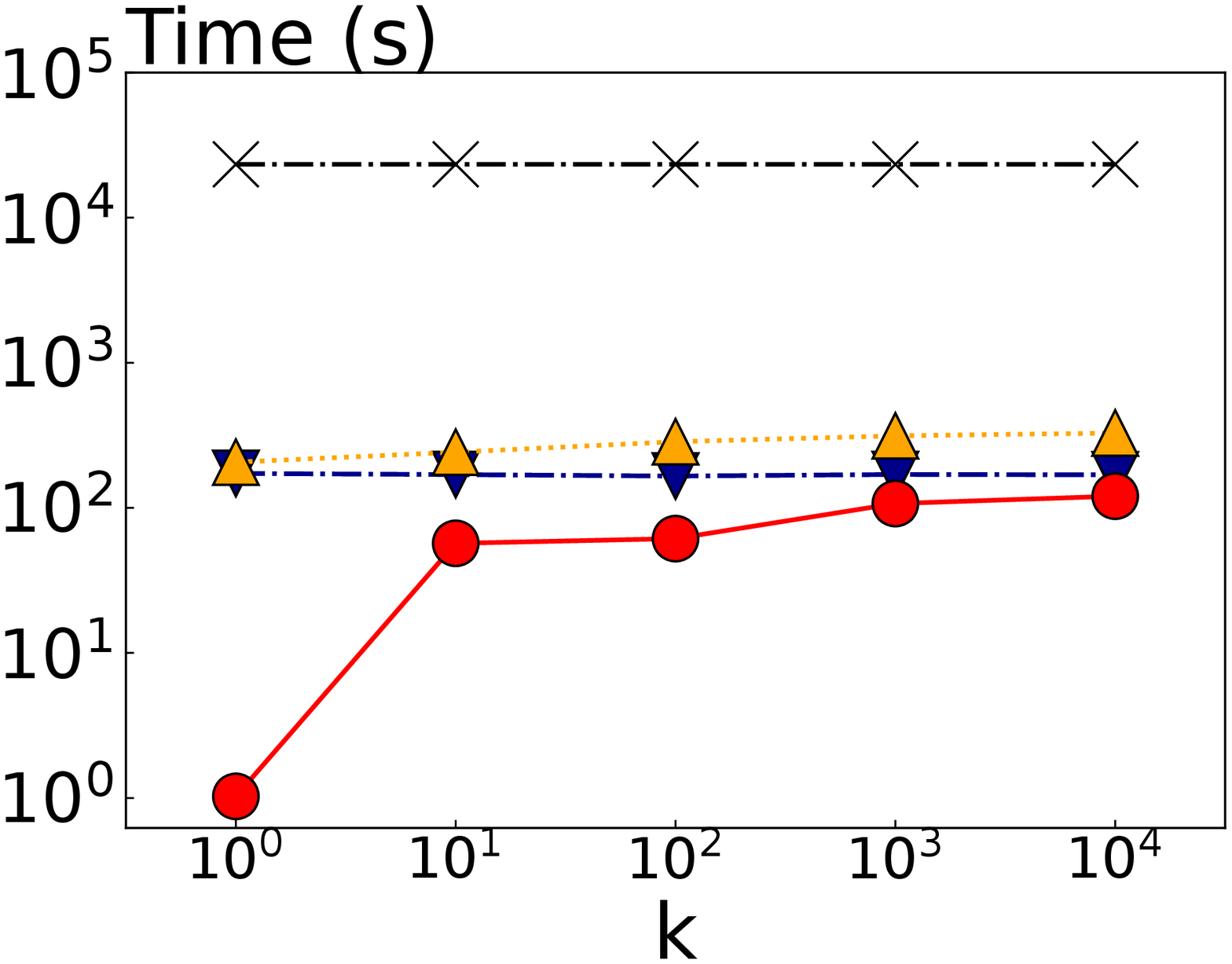}
		\end{minipage}
	}
	\subfigure[Recall on NUS]{
		\begin{minipage}[c]{0.3\linewidth}
			\centering
			\includegraphics[width=1\textwidth]{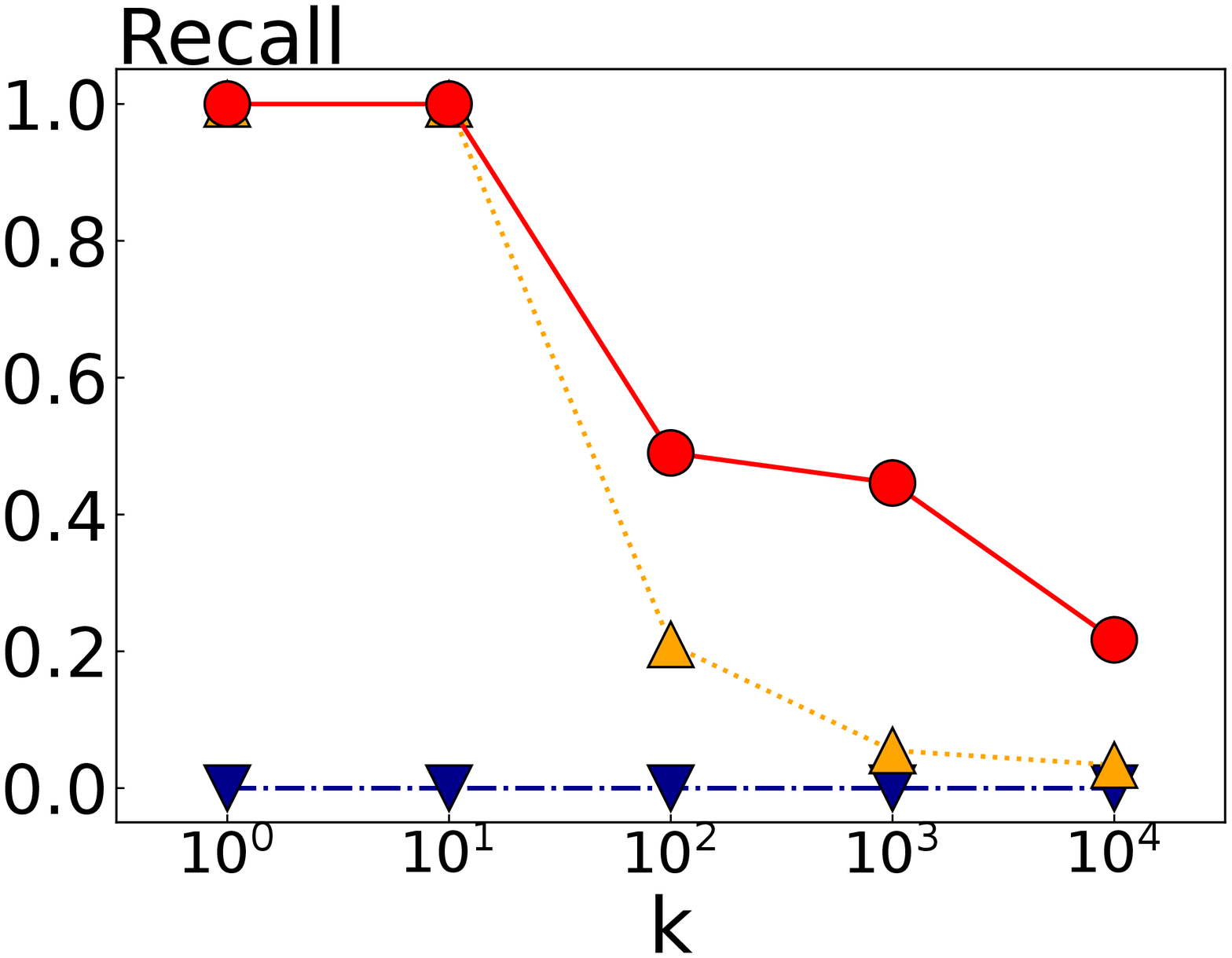}
		\end{minipage}
	}
	\subfigure[Ratio on NUS]{
		\begin{minipage}[c]{0.3\linewidth}
			\centering
			\includegraphics[width=1\textwidth]{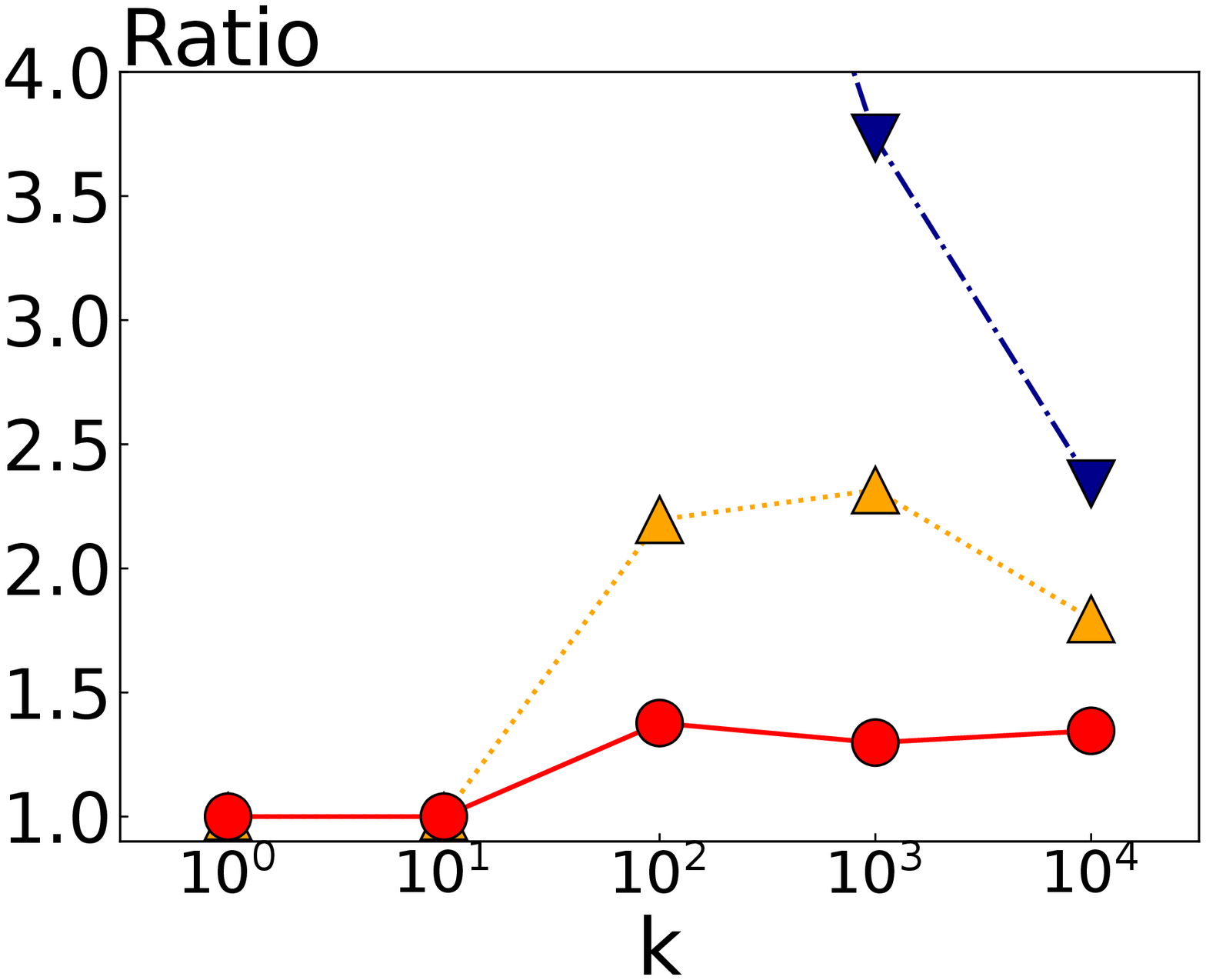}
		\end{minipage}
	}
	\caption{Performance on NUS when Varying $ k $ of CP Queries}
	\label{fig:cp_knn_NUS}
\end{figure*}

\begin{figure*}[htbp]
	\centering
	\subfigure[Time on Trevi]{
		\begin{minipage}[c]{0.3\linewidth}
			\centering
			\includegraphics[width=1\textwidth]{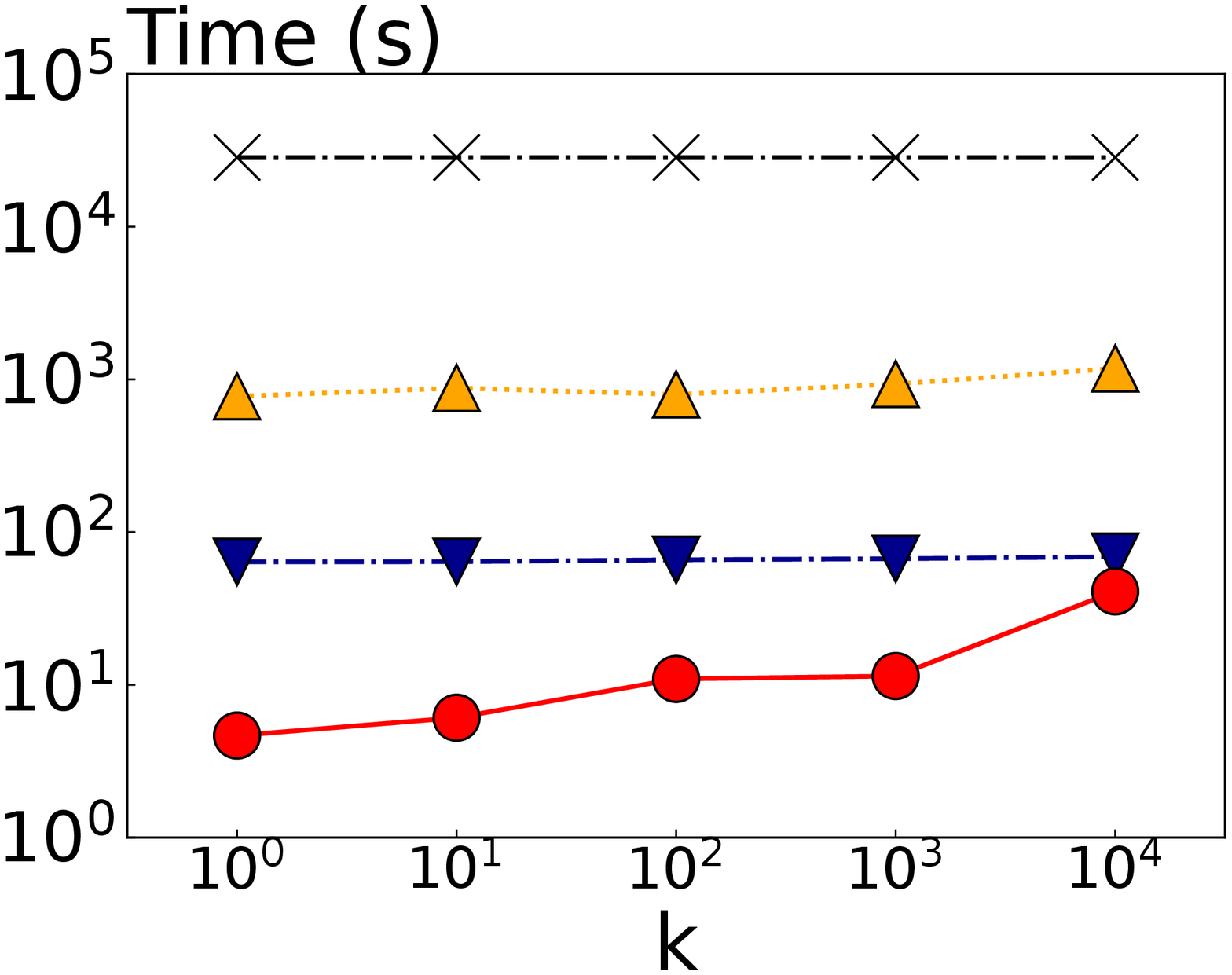}
		\end{minipage}
	}
	\subfigure[Recall on Trevi]{
		\begin{minipage}[c]{0.3\linewidth}
			\centering
			\includegraphics[width=1\textwidth]{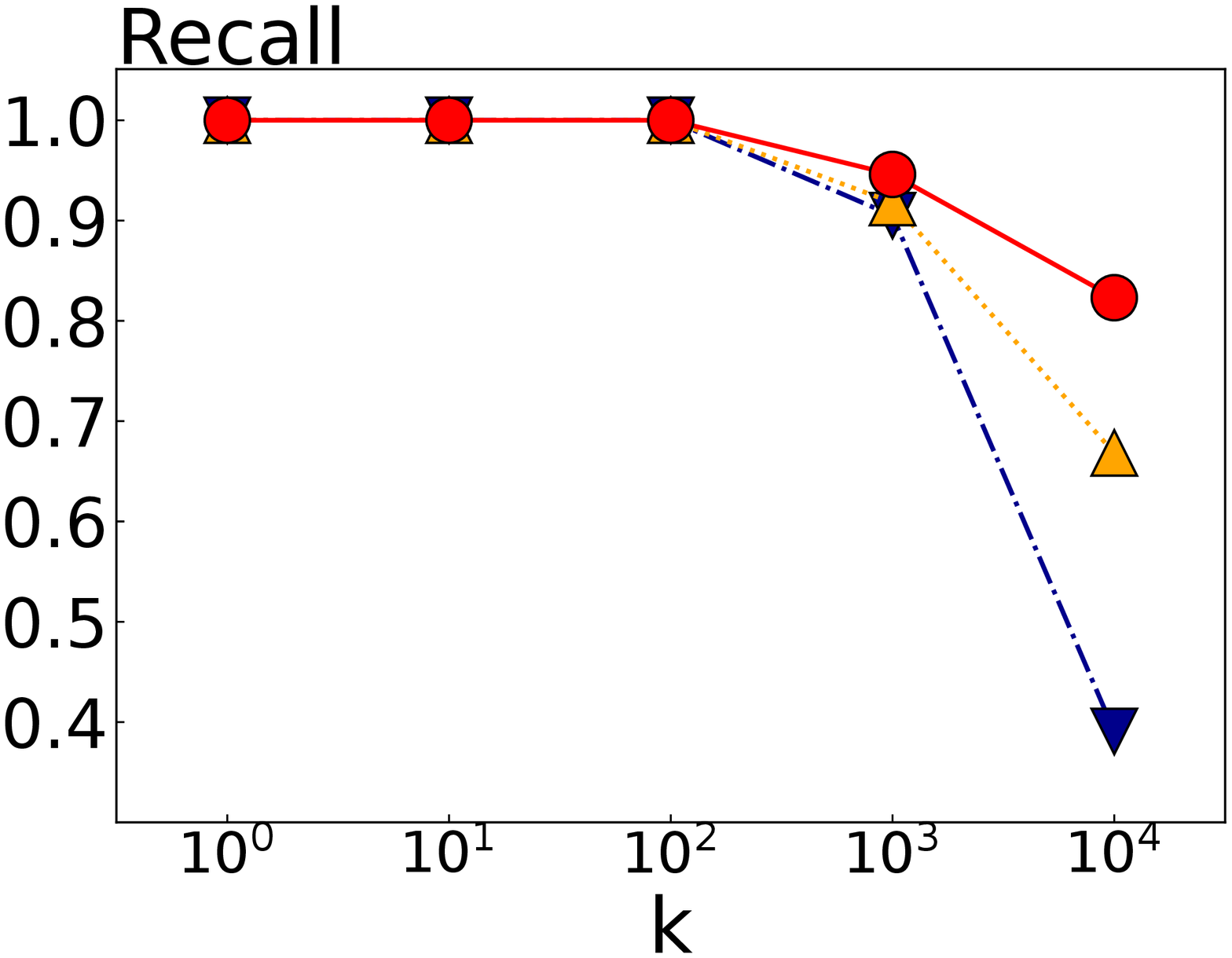}
		\end{minipage}
	}
	\subfigure[Ratio on Trevi]{
		\begin{minipage}[c]{0.3\linewidth}
			\centering
			\includegraphics[width=1\textwidth]{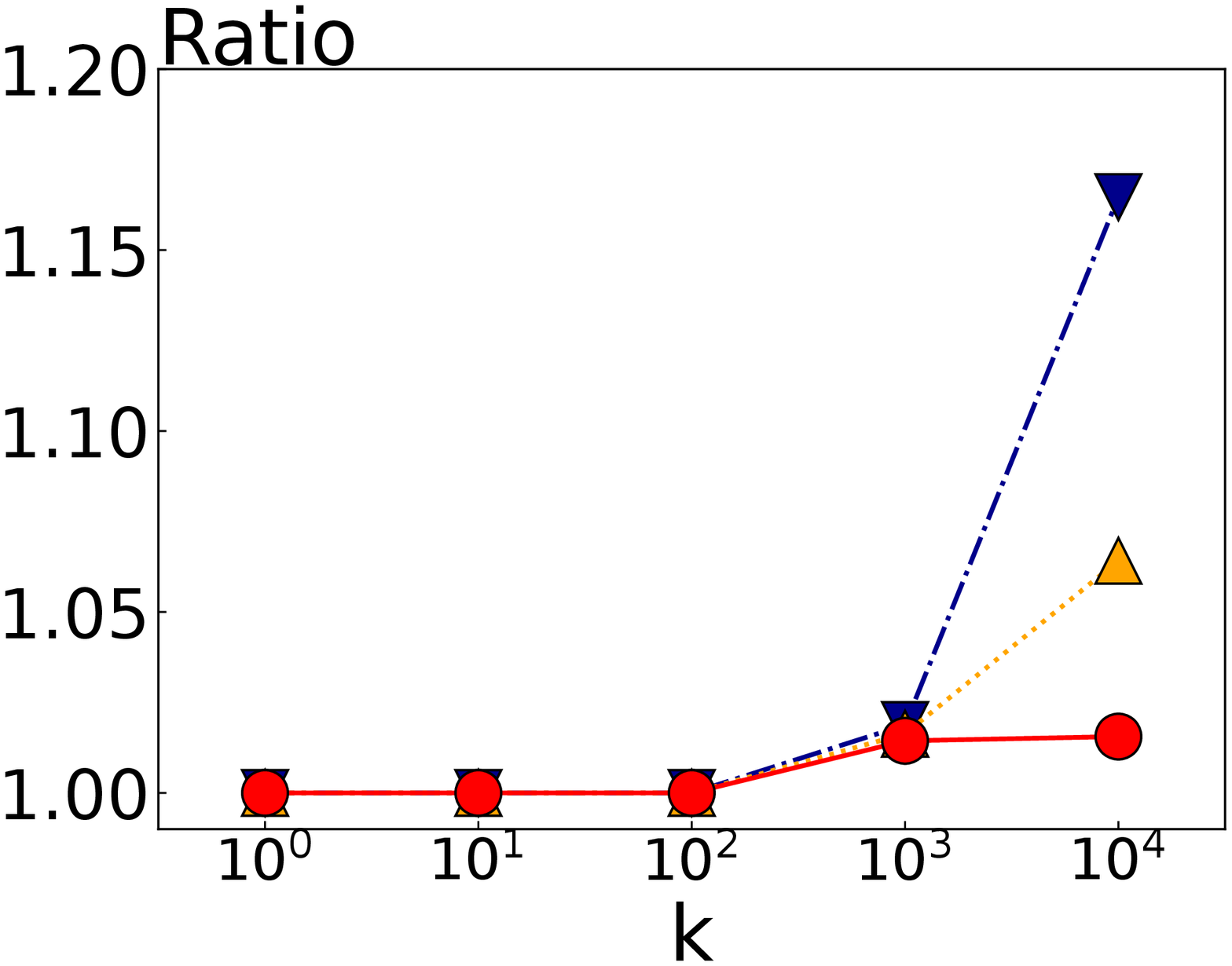}
		\end{minipage}
	}
	\caption{Performance on Trevi when Varying $ k $ of CP Queries}
	\label{fig:cp_knn_Trevi}
\end{figure*}

\begin{figure*}[htbp]
	\centering
	\subfigure[Recall-Time on Audio]{
		\begin{minipage}[c]{0.3\linewidth}
			\centering
			\includegraphics[width=1\textwidth]{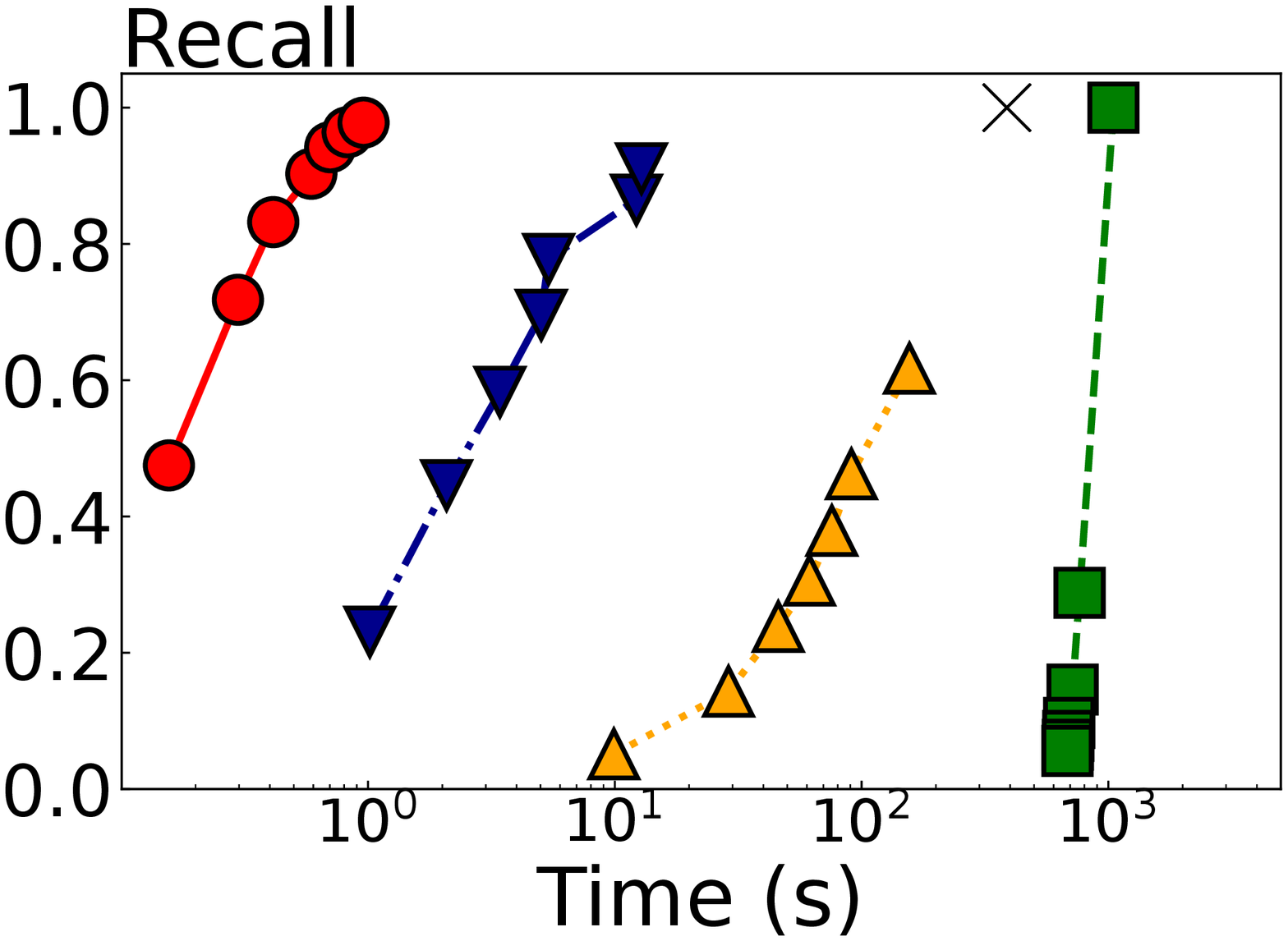}
		\end{minipage}
	}
	\subfigure[Recall-Time on Trevi]{
		\begin{minipage}[c]{0.3\linewidth}
			\centering
			\includegraphics[width=1\textwidth]{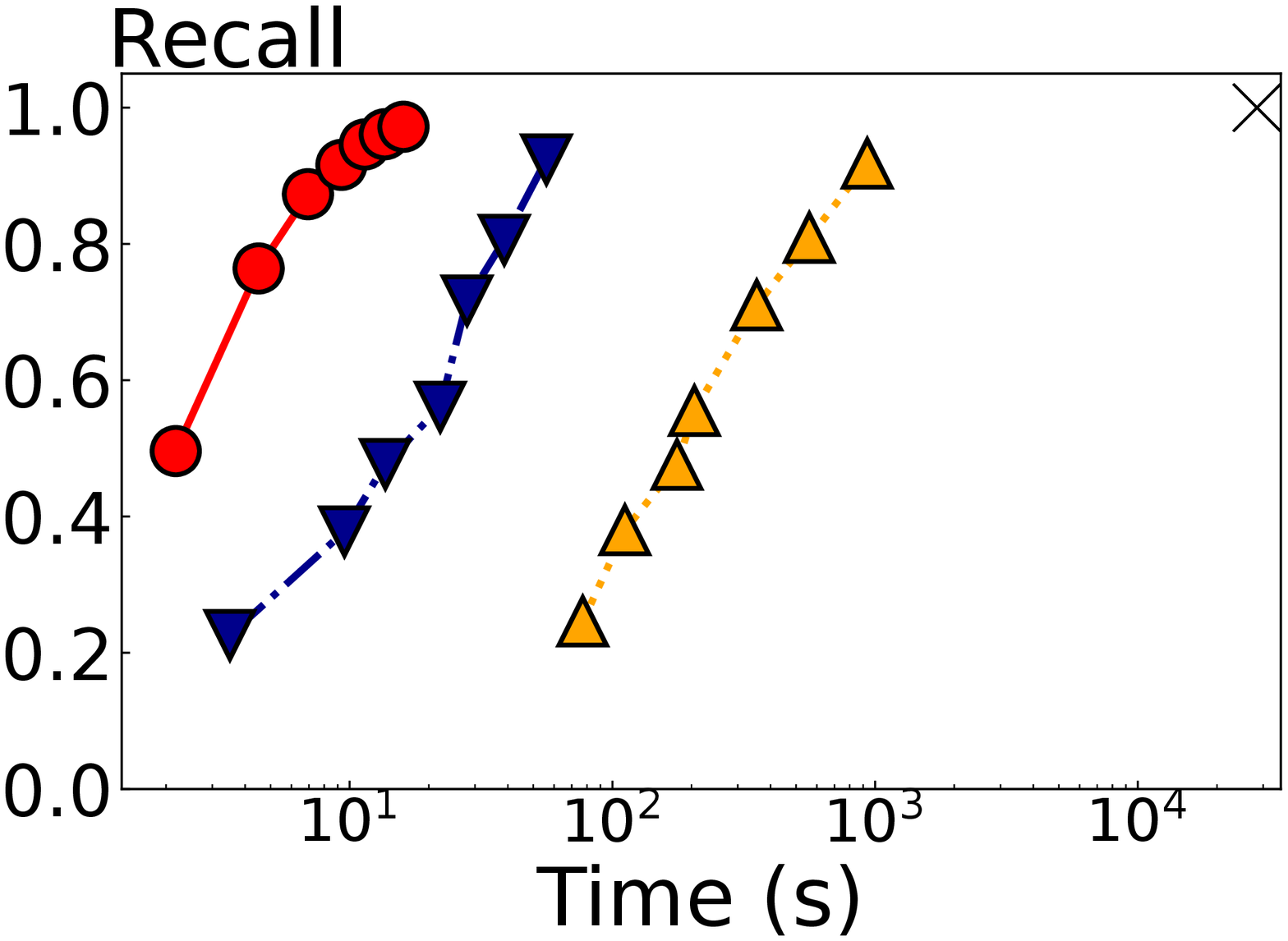}
		\end{minipage}
	}
	\subfigure[Recall-Time on NUS]{
		\begin{minipage}[c]{0.3\linewidth}
			\centering
			\includegraphics[width=1\textwidth]{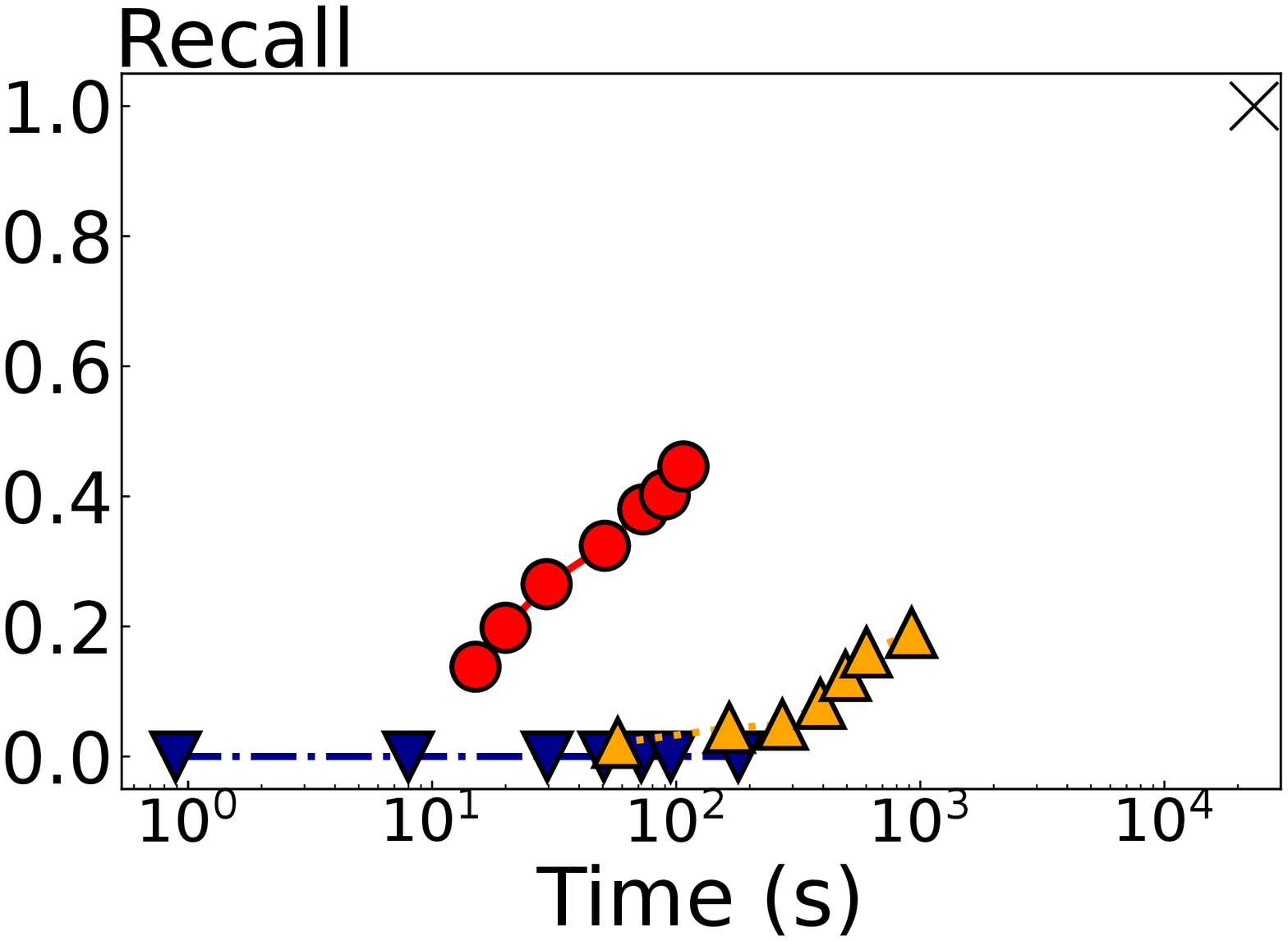}
		\end{minipage}
	}
	\caption{Recall-Time Curve for CP Queries}
	\label{fig:cp_result_recall_time}
\end{figure*}

\begin{figure*}[htbp]
	\centering
	\subfigure[Ratio-Time on Audio]{
		\begin{minipage}[c]{0.3\linewidth}
			\centering
			\includegraphics[width=1\textwidth]{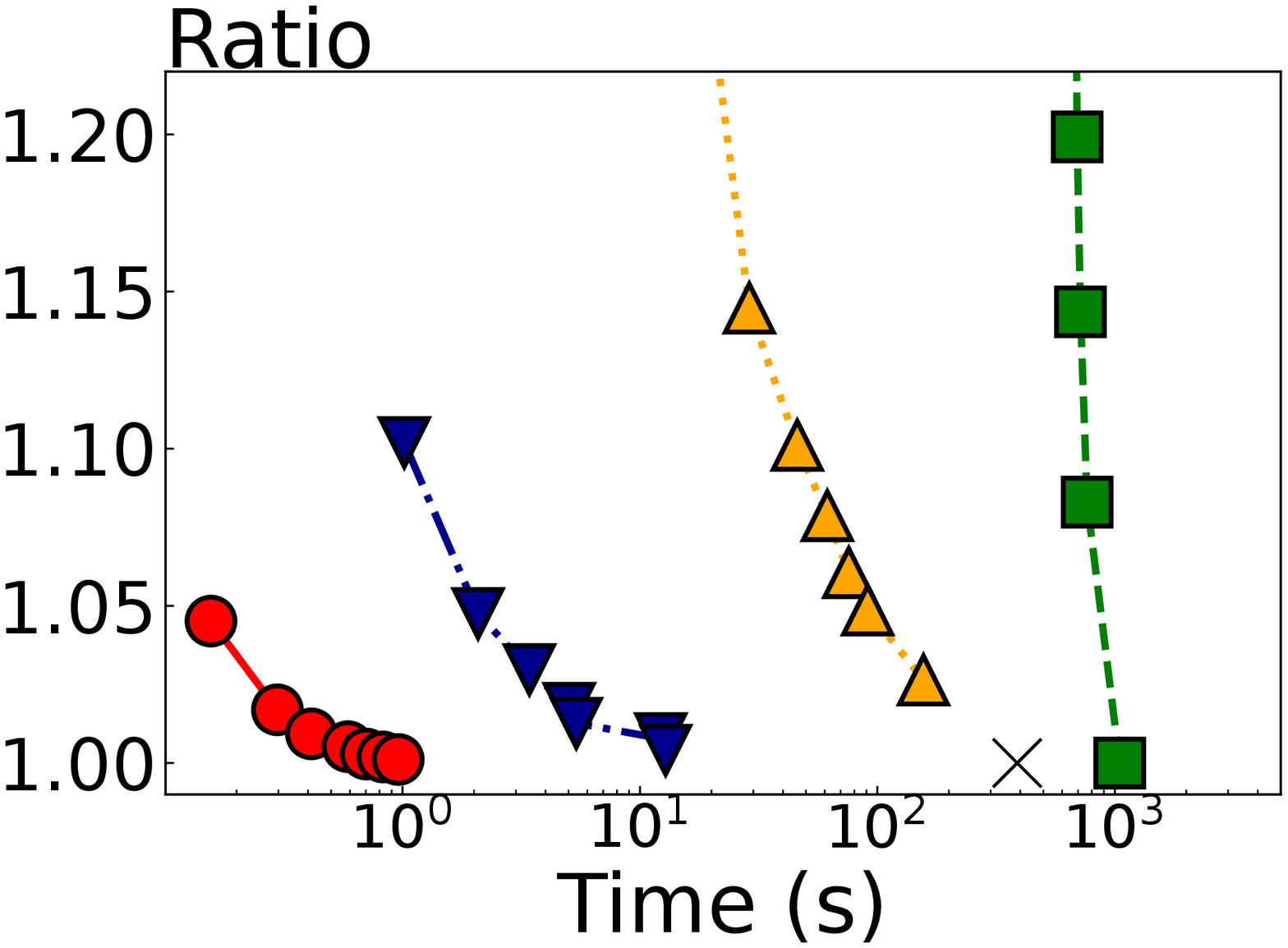}
		\end{minipage}
	}
	\subfigure[Ratio-Time on Trevi]{
		\begin{minipage}[c]{0.3\linewidth}
			\centering
			\includegraphics[width=1\textwidth]{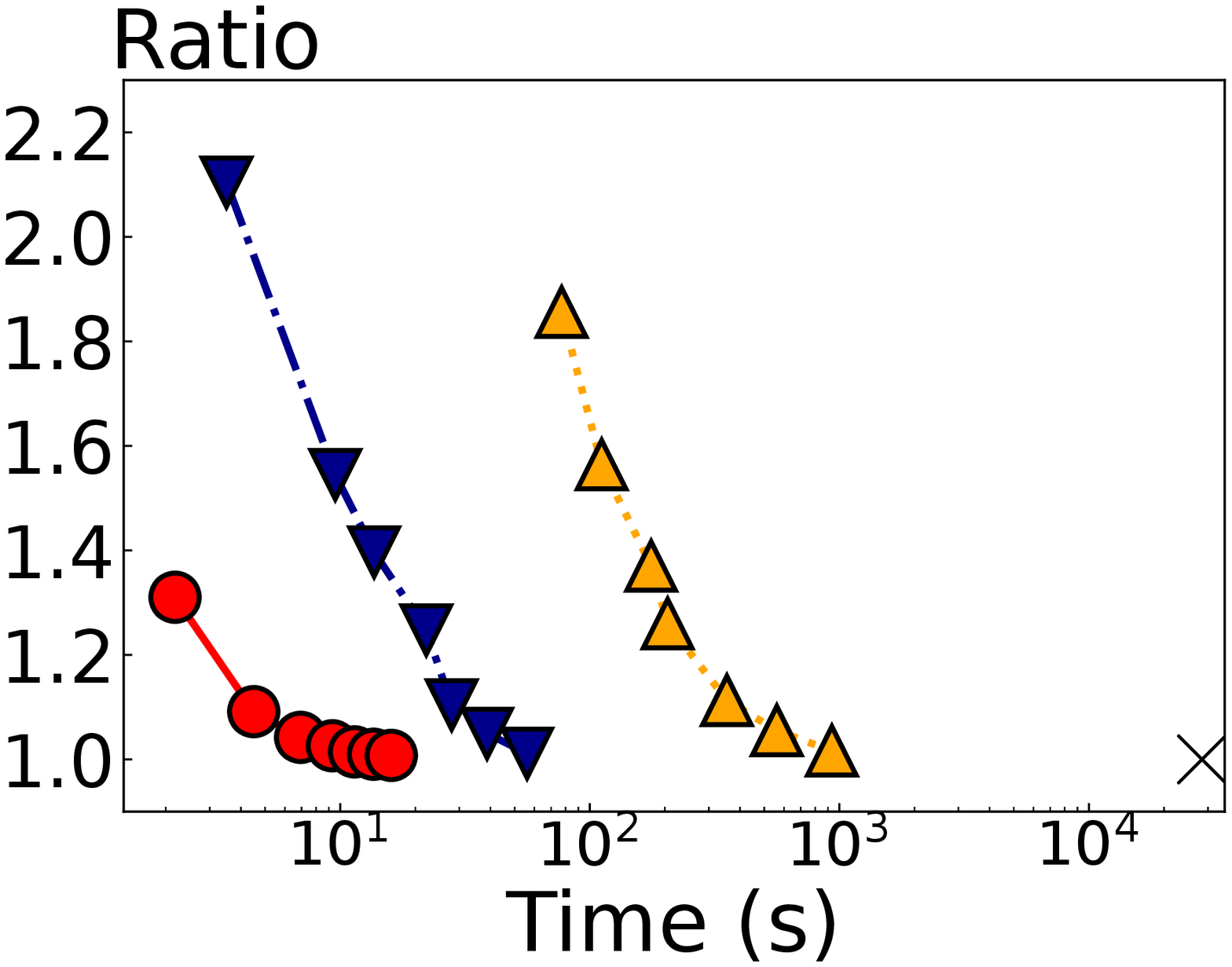}
		\end{minipage}
	}
	\subfigure[Ratio-Time on NUS]{
		\begin{minipage}[c]{0.3\linewidth}
			\centering
			\includegraphics[width=1\textwidth]{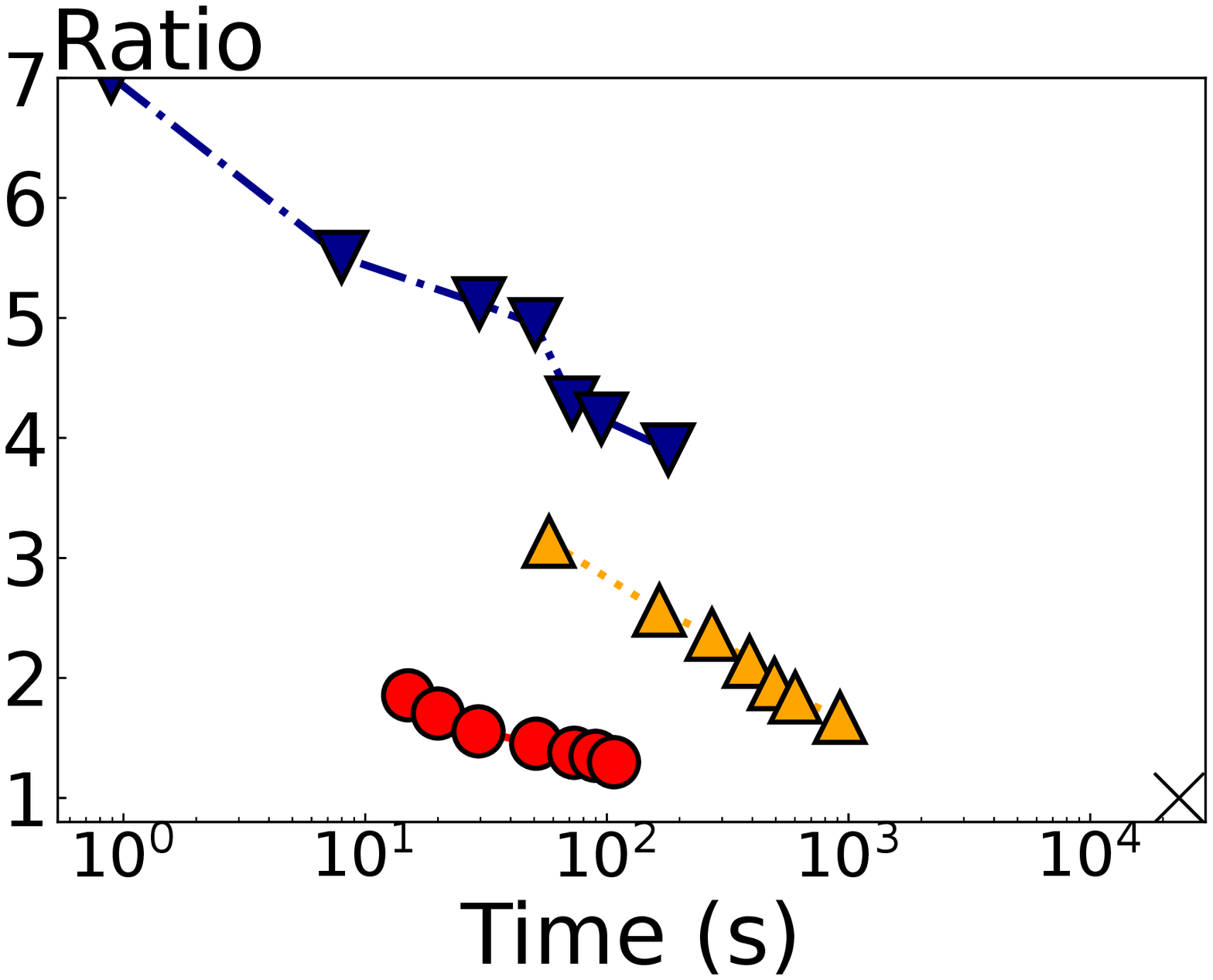}
		\end{minipage}
	}
	\caption{Ratio-Time Curve for CP Queries}
	\label{fig:cp_result_ratio_time}
\end{figure*}

\textbf{Effect of $ k $}.
Next, we study the performance when varying $ k $ in $\{1,10,10^2,$ $10^3,10^4\}$. For brevity, we only report the performance on datasets \textit{Audio}, \textit{Trevi}, and \textit{NUS}. We choose \textit{Audio} and \textit{NUS} instead of \textit{Cifar} and \textit{Deep} because M$k  $CP and ACP-P are inefficient for the latter two. The results are shown in Figs. \ref{fig:cp_knn_Audio}--\ref{fig:cp_knn_Trevi}.

With the increase of $ k $, most algorithms incur longer query times and worse recall and overall ratio.
The reason for a larger query time is that $ k $ affects the number of candidate pairs. PM-LSH, ACP-P, and M$ k $CP all use the $ k $-th smallest distance for pruning, so a large $ k $ means that more candidate pairs must be verified. The LSB-Tree returns the best $ k $ objects from a nearly fixed-size candidate sets, so its query time increases only slowly with $ k $. An exceptional case occurs for the LSB-tree on \textit{NUS}. The overall ratio improves with the increase of $k$. This is because many pairs have almost the same distances. 
When the result size increases, although the exact results are not found, the ratio of the distance of the $i$-th returned pair over that of the $i$-th exact pair decreases.

When considered across datasets, PM-LSH exhibits a consistent high accuracy. However, the query time of each algorithm varies substantially across the different datasets, which can be explained by three observations. (1) The query time is affected significantly by dataset cardinality $n$. For instance, the query times of PM-LSH, the LSB-tree, and ACP-P are subquadratic to $n$; the query time of M$k$CP is $O(n^2)$ in the worst case. (2) The query time is affected by dataset dimensionality $d$. All algorithms need to verify candidate pairs, and the cost is linear in $d$. (3) The data distribution also affects the query time, which is a key determining factor for when the algorithms terminate.

To sum up, PM-LSH has the smallest query time among all competitors. In addition, the accuracy is high.
Only the LSB-tree is able to achieve a competitive recall in some cases but incurs longer query time than PM-LSH.

\textbf{Recall-Time and OverallRatio-Time Curves}.
We proceed to study the relationship between the recall or overall ratio and the query time for $(c,k)$-ACP queries on all the datasets when varying their configurations to obtain different query times, such as $c$ for PM-LSH, $N$ for M$k$CP, $L$ for the LSB-tree, and repeat times for ACP-P. The results are shown in Figs. \ref{fig:cp_result_recall_time} and \ref{fig:cp_result_ratio_time}. As the query quality and the query time represent the key tradeoff, the algorithms focus on returning relatively good results with much smaller query times than those of exact CP algorithms. 
The results show that all algorithms return more accurate results when more query time is used. 
They also show that PM-LSH achieves superior efficiency and accuracy when compared to the LSB-tree, ACP-P, and M$k$CP. This can be explained as follows. First, PM-LSH has a better distance estimator than the LSB-tree and ACP-P, so PM-LSH outperforms them with the same number of retrieved points. Second, PM-LSH uses a radius filtering technique to generate candidate pairs, which reduces substantially the cost of generating candidate pairs and provides a well-designed condition to terminate the process early. Third, the hyper-ball and hyper-ring space partitioning help reduce unnecessary verification overhead. In addition, although M$k$CP also finds approximate closest pairs in a space partitioning tree, it indexes high-dimensional data directly, which makes pruning difficult. Therefore, its query time is much larger than those of the other methods.

\section{Related Work}\label{sec:relatedwork}

\subsection{LSH for Nearest Neighbor Search}

Locality-Sensitive Hashing (LSH) is a prominent approach to speeding up the processing of approximate nearest neighbor querying \cite{DBLP:conf/vldb/GionisIM99, DBLP:conf/compgeom/DatarIIM04, DBLP:conf/www/BawaCG05, DBLP:conf/vldb/LvJWCL07, DBLP:conf/cikm/DongWJCL08}. LSH was originally proposed by Indyk et al. \cite{DBLP:conf/stoc/IndykM98} for use in Hamming space, and it has since attracted substantial attention due to its excellent performance.
Datar et al. \cite{DBLP:conf/compgeom/DatarIIM04} propose an LSH function based on $p$-stable distributions in Euclidean space, which has become a mainstream method that yields low computation cost, a simple geometric interpretation, and a good quality guarantee. 
Since then, many LSH methods build on this work to choose hash functions \cite{DBLP:conf/vldb/LvJWCL07, DBLP:conf/edbt/HaghaniMA09, DBLP:conf/sigmod/TaoYSK09, DBLP:conf/sigmod/GanFFN12, DBLP:journals/pvldb/SunWQZL14, DBLP:journals/pvldb/HuangFZFN15}. 
In addition to the competitors introduced in Section \ref{sec:competitors}, other proposals deserve mention.
Based on a rigorous theoretical analysis, Panigrahy et al. \cite{DBLP:conf/soda/Panigrahy06} propose an entropy-based LSH, and Satuluri et al. \cite{DBLP:journals/pvldb/SatuluriP12} propose a BayesLSH.
The former tries to reduce the number of hash tables by using multiple perturbed queries, and the latter aims to reduce the query time by estimating the similarity between data and query objects based on Bayes rule. However, both yield limited performance improvements as the assumptions made on the underlying dataset are hard to satisfy and verify.
Another interesting proposal is LazyLSH \cite{DBLP:conf/sigmod/ZhengGTW16}, which supports queries in multiple $l_p$ spaces by using one index, thus effectively reducing the space overhead.
Another line of hashing-based methods is learning to hash (L2H) \cite{DBLP:journals/pami/WangZSSS18}, which is orthogonal to our work. LSH uses predefined hash functions without considering the underlying dataset, while L2H learns tailored, data dependent hash functions. Many learning algorithms have been proposed, such as iterative quantization (ITQ) \cite{DBLP:journals/pami/GongLGP13} and generate-to-probe QD ranking (GQR) \cite{DBLP:conf/sigmod/LiYZXCLNC18}. 

\subsection{High Dimensional Closest Pair Search}

Closest-Pair (CP) search is an important problem in the database domain. Early studies target mainly low-dimensional closest pair search \cite{DBLP:conf/sigmod/CorralMTV00,DBLP:conf/sigmod/HjaltasonS98,DBLP:conf/ssd/ShanZS03,DBLP:journals/dke/CorralMTV04,DBLP:journals/tkde/ShinML03,DBLP:journals/tkde/KimP10}. They adopt spatial index structures, such as the R-tree and Quadtree and their variants, to organize the data.
However, these methods fail to handle high-dimensional closest pair search due to the curse of dimensionality. Corral et al. \cite{DBLP:conf/adbis/CorralDMV05} propose a join method based on the VA-file, which is an array structure rather than a tree structure. Angiulli et al. \cite{DBLP:journals/dke/AngiulliP05} adopt the Z-curve to reduce the dimensionality and generate candidates in one-dimensional spaces. 
Tao et al. \cite{DBLP:journals/tods/TaoYSK10} propose an LSB-tree that uses a compound hash function to project points into a low-dimensional space. Next, they adopt the Z-curve to map the projected points into one-dimensional values that are indexed by a B-tree. 
Candidate point pairs are generated from the points with the same Z-values. However, $L=O(\sqrt{n})$ B-trees are required, thus causing a large space consumption. Mueen et al. \cite{DBLP:conf/sdm/MueenKZCW09} partition the data based on their distances to a pivot and thus map the high-dimensional data to a one-dimensional space. Other studies use LSH \cite{DBLP:journals/tkde/YuNLWY17,DBLP:journals/tkde/LiNXYH19} or random projection\cite{DBLP:conf/pakdd/CaiRZ18} to reduce the dimensionality. For instance, Cai et al. \cite{DBLP:conf/pakdd/CaiRZ18} project the data directly into a one-dimensional space. Nearby points in the projected space are considered as candidate point pairs. However, the distance estimation is inaccurate and leads to unnecessary verification overhead.

Unlike the previously covered methods that use dimension reduction, yet other studies organize the original data directly by means of novel index structures, such as the LTC index \cite{DBLP:journals/jda/ParedesR09}, the multi-ball \cite{kurasawa2011finding,DBLP:conf/sisap/FredrikssonB13}, and the eD-Index \cite{DBLP:conf/sisap/PearsonS14}. Specifically, Gao et al. \cite{DBLP:journals/vldb/GaoCLYC15} propose several efficient algorithms using the count M-tree.
However, these methods still suffer from the curse of dimensionality. 

In addition, distributed indexing based approaches \cite{DBLP:conf/kdd/WangMP13,DBLP:journals/tkde/LiNXYH19} are proposed to accelerate CP search. These enable in-memory processing of large scale datasets.
\section{Conclusion} \label{sec:conclusion}
We present a fast and accurate in-memory framework, called PM-LSH, for computing $(c,k)$-ANN and $(c,k)$-ACP queries with theoretical result quality guarantees. For NN queries, we first adopt the PM-tree to index the data points to be queried in a projected space. Second, in order to improve the distance estimation accuracy in the projected space, we develop a tunable confidence interval on the projected distance w.r.t. a given original distance. Finally, we propose an efficient algorithm to compute range queries using the PM-tree.
The experimental study using seven widely used datasets shows that PM-LSH is capable of outperforming five competitors in terms of
both query efficiency and result accuracy. Specifically, PM-LSH improves the query time by an average of 30\% when compared to the closest competitor. When all competitors are given approximately the same query time, PM-LSH improves the recall by about 10\% when compared to the closest competitor.

For computing CP queries, we also use the PM-tree to index the points in the projected space. We propose a radius filtering technique for finding closest pairs in the PM-tree. The experimental study shows that PM-LSH is capable of outperforming four competitors in terms of
both query efficiency and result accuracy. Specifically, PM-LSH improves the query time by an average of 40\% when compared to the closest competitor. When all the competitors are given approximately the same query time, PM-LSH improves the recall by about 50\% when compared to the closest competitor.

\section*{Acknowledgments}
This research is supported in part by the NSFC (Grants No. 61902134, 62011530437), the Hubei Natural Science Foundation (Grant No. 2020CFB871), and the Fundamental Research Funds for the Central Universities (HUST: Grants No. 2019kfyXKJC021, 2019kfyXJJS091).

\bibliographystyle{abbrv}
\bibliography{myRef}
\end{document}